\documentclass[11pt,fleqn,reqno]{article} 
\usepackage[margin=.65in]{geometry}
\usepackage{amsmath,amsthm,amssymb,mathrsfs,euscript}
\usepackage{enumerate}
\usepackage{xcolor}

\usepackage[numbers]{natbib}
\usepackage{hyperref}
\usepackage{graphicx}

\usepackage{makeidx}

\usepackage{hyperref}

\theoremstyle{plain}
\newtheorem{theorem}{Theorem}
\newtheorem{proposition}[theorem]{Proposition}
\newtheorem{lemma}[theorem]{Lemma}%
\newtheorem{corollary}[theorem]{Corollary}

\theoremstyle{definition}

\newtheorem{remark}[theorem]{Remark}

\newcommand{\R}{\mathbb{R}}
\newcommand{\C}{\mathbb{C}}
\newcommand{\Z}{\mathbb{Z}}
\newcommand{\N}{\mathbb{N}}
\newcommand{\bH}{\mathbb{H}}

\newcommand{\cE}{\mathcal{E}}
\newcommand{\E}{\mathcal{E}}

\newcommand{\cH}{\mathcal{H}}

\newcommand{\cK}{\mathcal{K}}
\newcommand{\cL}{\mathcal{L}}         
\newcommand{\cO}{\mathcal{O}}         

\newcommand{\cS}{\mathcal{S}}

\newcommand{\cX}{\mathcal{X}}
\newcommand{\cY}{\mathcal{Y}}
\newcommand{\cZ}{\mathcal{Z}}
\newcommand{\nc}{\newcommand}

\nc{\G}{\Gamma}
\nc{\g}{\gamma}
\nc{\al}{\alpha}
\nc{\be}{\beta}
\nc{\del}{\delta}
\nc{\io}{\iota}
\nc{\ka}{\kappa}
\nc{\lam}{\lambda}
\nc{\Lam}{\Lambda}
\nc{\w}{\omega}
\nc{\om}{\omega}
\nc{\Om}{\Omega}
\nc{\Oms}{\Omega^*}
\nc{\s}{\sigma}
\nc{\Si}{\Sigma}
\nc{\ta}{\tau}
\nc{\h}{\theta}
\nc{\z}{\zeta}

\newcommand{\vphi}{\varphi}

\newcommand{\fg}{\mathfrak{g}}

\nc{\ran}{\rangle}
\nc{\lan}{\langle}

\newcommand{\im}{\operatorname{Im}}
\renewcommand{\Re}{\operatorname{Re}}
\renewcommand{\Im}{\operatorname{Im}}
\newcommand{\ra}{\rightarrow}
\newcommand{\spec}{\operatorname{spec}}

\newcommand{\n}{\nabla}
\newcommand{\p}{\partial}

\newcommand{\Curl}{\operatorname{curl}}

\newcommand{\curl}{\operatorname{curl}}
\renewcommand{\div}{\operatorname{div}}
\newcommand{\divv}{\operatorname{div}}

\newcommand{\grad}{\operatorname{grad}}

\newcommand{\Null}{\operatorname{Null}}

\newcommand{\Ran}{\operatorname{Ran}}
\newcommand{\Tr}{\operatorname{Tr}}

\newcommand{\one}{\mathbf{1}}
\newcommand{\id}{{\bfone}}
\nc{\bfone}{{\bf 1}}

\nc{\Oml}{\Omega^\lat}

\newcommand{\LAT}{\mathcal{L}}
\newcommand{\lat}{\mathcal{L}}

\newcommand{\DETAILS}[1]{}

\newcommand{\LA}[2]{\vec{\mathscr{L}}_{}^{}(\tau)}
\newcommand{\HA}[2]{\vec{\mathscr{H}}_{}^{}(\tau)}

\nc{\hl}{\textcolor{red}}

\numberwithin{theorem}{section}
\numberwithin{equation}{section}
\begin{document}
\date{\DATUM}

\newcommand{\DATUM}{June, 2023} 
\pagestyle{myheadings}                         
\markboth{\hfill{Abrikosov lattices in the Weinberg-Salam model, June 7, 2013}}{{Symmetry breaking in the Weinberg-Salam model, June 30, 2023}\hfill}  %

\title
{Instability of electroweak homogeneous vacua in strong magnetic fields} 

\author{Adam Gardner\footnote{Artinus Consulting Inc., Ottawa, Canada, adampg@artinus.ai}\, and Israel Michael Sigal\footnote{Dept. of Math., U. of Toronto, Toronto, 
		Canada,  imsigal@gmail.com}}

\maketitle

\centerline{\it To Nicholas Ercolani and Gian Michele Graf, scientists and friends.}

\begin{abstract}
	We consider the  classical 
	 vacua of the  Weinberg-Salam (WS) model of electroweak forces. 
These are no-particle, static solutions to the WS equations minimizing  the WS energy locally. 

\DETAILS{These equations have (gauge-) translation invariant (homogeneous) solutions  for all constant external magnetic fields.		 For zero external magnetic field, the corresponding solution is the vacuum 
 of the $U(2)$-YMH equations. 

	We prove that (i) the homogeneous vacua  are stable for weak magnetic fields and unstable for strong ones and (ii) inhomogeneous solutions with lower energy per unit volume emerge at the transition point. 
	The latter solutions 
	have the discrete translational symmetry of a 2D lattice in the plane transversal to the external magnetic field.} 

We study the WS 
vacuum solutions exhibiting a non-vanishing average magnetic field of strength $b$, and prove that (i) there is a magnetic field threshold $b_*$ such that for $b<b_*$, the vacua are translationally invariant (and the magnetic field is constant), 
 while, for  $b>b_*$, they are not, (ii) 
for $b>b_*$, there are non-translationally invariant solutions with lower energy per unit volume 
 and with the discrete translational symmetry of a 2D lattice in the plane transversal to 
   $b$,   and (iii) 
 the lattice 
  minimizing the energy per unit volume approaches the hexagonal one as the magnetic field strength approaches the threshold $b_*$.

	In the absence of particles, the  Weinberg-Salam model reduces to  the Yang-Mills-Higgs (YMH) equations for the gauge group $U(2)$. 
	\DETAILS{The no-particle vacua 
	 minimize  the $U(2)$-YMH energy locally.} 
Thus our results can be rephrased as the corresponding statements about    the  $U(2)$-YMH equations.

MCS classes: 81T13 (primary), 35Q40, 70S15 (secondary)

	
\end{abstract}


\section{Introduction}

The Weinberg-Salam (WS) model of electroweak interactions \cite{Gl}\cite{SW}\cite{Wb} was the first triumph of the program to unify the four fundamental forces of nature. It is a key part of the standard model of elementary particles. 
It unifies electromagnetic and weak interactions, two of the three forces dealt with in the standard model. 
It involves particle, gauge 
 and the Higgs fields. 

While the gauge 
 fields  describe the electroweak interactions, the role of the Higgs field 
is to convert the original massless fields (zero masses are required by the relativistic invariance) to massive ones. This phenomenon is called 
 the Higgs mechanism. This mechanism, together with the Goldstone theorem, leads to all gauge particles but one acquiring mass, resulting in two massive bosons -- denoted W and Z -- and a massless one -- the photon. The  W and Z 
  particles where discovered experimentally 16 years after their theoretical prediction.

\DETAILS{It was shown in \cite{NO, AHN, Sk} that for large values of magnetic field, the Weinberg-Salam vacuum becomes unstable and suggested in \cite{AO1} that the true vacuum is inhomogeneous with a lattice symmetry,  
 similar to that occurring in superconductivity (\cite{Abr}). 
  This was investigated extensively in physics literature (see e.g. \cite{MT, CvDV, Mir, And} and references therein).} 

In this paper, we consider the vacuum solutions of the classical WS model with a non-vanishing {\it average magnetic field} $\langle\vec b\rangle$. These are  static, no-particle solutions minimizing  the WS energy locally for a fixed $\vec b$. 
They are also no-particle solutions of the entire standard model.\footnote{The no-particle sector  of the standard model splits into the $U(2)$-YMH (electroweak) and  $U(3)$-YM (strong, or QCD)  parts. Correspondingly, the vacuum of the standard model is the product of the electroweak and strong vacua and the vacuum energy is the sum of the corresponding energies.) }

We prove that (i) there is a magnetic field threshold $b_*$ such that for 
 $|\vec b|<b_*$, the vacua are translationally invariant, 
 while, for  $|\vec b|>b_*$, they are not, (ii) 
for $|\vec b|>b_*$, there are non-translationally invariant solutions with lower energy per unit volume 
 and with the discrete translational symmetry of a 2D lattice in the plane transversal to the  magnetic field,   and (iii) 
 the lattice 
  minimizing the energy of the latter solutions per unit volume approaches the hexagonal one as the magnetic field strength approaches the threshold $b_*$.	We expect that these solutions are stable under field fluctuations and, in fact, minimize the energy locally.

\DETAILS{We prove that for weak external constant magnetic fields, the homogeneous  
  vacuum solution is stable, but  for strong fields, it is unstable
i.e. 
 the translational symmetry is broken  spontaneously. We show that this leads to the emergence of 
 inhomogeneous solutions with the symmetry of a 2D lattice (in the plane transversal to the external magnetic field) and with lower energy per unit volume. 
We expect that these solutions are stable under field fluctuations and, in fact, minimize the energy locally.}

The phenomenon above 
was investigated extensively in the physics literature (see e.g. \cite{And, CvDV, MT, Mir} and the references therein). It is similar to the one occurring in superconductivity and the solutions whose existence we establish are analogous to the superconducting Abrikosov vortex lattices  (\cite{Abr}, see e.g. \cite{S}, for a review). 
It is estimated in \cite{MT} that the spontaneous symmetry breaking takes place at the critical 
 average magnetic field of approximately $10^{24}$  Gauss $=10^{20}$ Tesla. By comparison, the strongest magnetic field produced on Earth is $10^{14} $ Tesla. 

	Note that, in the absence of particles, the  WS system reduces to  the Yang-Mills-Higgs (YMH) one with the gauge group $U(2)$. So ultimately, these are the equations we deal with. 

 The only rigorous result (\cite{SY, SY2}) on the classical WS model  deals with the vortices in the self-dual regime, where the WS (or corresponding YMH) equations are equivalent to the first order equations, and it uses this equivalence 
 in an essential way.   (The self-dual regime in this context was discovered in \cite{AO2, AO3, AO4}, see also \cite{Sk1, Sk2}.)
  
  \medskip
  
\noindent {\bf Open problems and further directions:} 

(a)  
Stability of the emerging solutions.

(b)  Existence of 
vortex lattices at $|\vec b|\gg b_*$.

(c) 
Quantum corrections to  
the values of the classical critical magnetic field $b_*$ and the optimal lattice shape parameter $\tau_*$. 

 For the stability and existence problems, (a) and (b), see e.g.  \cite{ST1, ST3} and  \cite{ST2}, respectively. The last problem brings up the regime of `sparse' vortex lattices as opposite to the case of $|\vec b|$ close to (and $> $) $b_*$  resulting in densely packed vortices: the lattice step $\ra 0$ as $|\vec b|\ra b_*$ and $\ra \infty$ as $|\vec b|\ra \infty$. Hence the existence of vortex lattices at $|\vec b|\gg b_*$ is closely related to the problem of {\it existence of vortices} (elementary excitations). 
 
 For the quantum corrections, problem (c), it would be natural to start 
with a BCS-type, or quasi-free, version of 
 the WS 
 model  and a Bogoliubov-type expansion of a regularized (say, lattice) WS model  around it, see e.g. \cite{BBCFS, BenPorSchl}.  


  
  \medskip

 The paper is organized as follows. In Section \ref{sec:probl-res}, we formulate the problem and describe results.  In Sections \ref{sec:WZ-bosons} - \ref{Sec:rescaling}, we fix the gauge and pass from the original Yang-Mills 
  fields to the W and Z (massive boson) and A (photon) 
   fields and rescale the resulting equations. The proofs of the main results are given in Section \ref{sec:linear}   (Theorem \ref{thm:normal-instab}),  Sections \ref{sbp} - \ref{sec:asymp} 
  (Theorem \ref{thm:AL-exist'}) and Section \ref{sec:shape}  (Theorem \ref{thm:lattice-shape}). In Appendix \ref{sec:Cov-Deriv-Curv}, we discuss 
  various covariant derivatives used in the main text,
    and in Appendix \ref{sec:YM}, we review the time-dependent YMH equations and derive the expression for the conserved energy as well as the YMH equations used in the main text. Furthermore,  
  there we write the YMH equations in coordinate form and derive a convenient expression for the energy functional. 
  In  Appendices \ref{sec:en-expl} - \ref{sec:WSeqs2D}, we derive the WS equations in 3D and 2D, respectively, in terms of the fields $W$, $Z$, $A$ and $\vphi$.
In the remaining appendices, we carry out technical computations. 
  
 Throughout the paper, we use the Einstein convention of {\it summing over repeated indices}.


\bigskip

\paragraph{\bf Acknowledgements}
The second author is grateful to Nicholas Ercolani, J\"urg Fr\"ohlich, Gian Michele Graf and Stephan Teufel for many instructive and stimulating discussions of the YMH equations. Both authors thank the anonymous referee for many constructive remarks. 

\section{No-particle and vacuum sectors of the Weinberg-Salam model} \label{sec:probl-res} 


 The no-particle sector of the Weinberg-Salam (WS) model involves the interacting 
 Higgs and $SU(2)$ and $U(1)$ gauge fields,  $\Phi$ and $ V$ and $ X$, while the particle fields are set to zero. The field $\Phi$ is a vector-function defined on the 
 Minkowski space-time $\R^{3+1}$
  with values in $\C^2$, and the fields $ V$ and $ X$ are one-forms on $\R^{3+1}$ with values in 
 the algebras $\mathfrak{su}(2)$ and $\mathfrak{u}(1)$, respectively. 
We write \[Q=  g V+ g' X,\] where $g$ and $g'$ are coupling constants, which is a one-form with values in $\mathfrak{u}(2)$. 
We consider $SU(2)$ as a matrix group 
 and $U(1)$ as multiples of the identity matrix $\id$ acting on $\C^2$. 

These fields satisfy the WS equations, which are the  Euler-Lagrange equations for the action functional
\begin{equation} \label{WS-act}
\cS (Q, \Phi) =  \int_{M} \big( \lan \n_{Q}\Phi,  \n_{Q}\Phi \ran_{\Omega^1_V}^\eta - \frac{1}{2} \lam (\|\Phi\|_{\C^2}^2-\vphi_0^2)^2  
 + \lan F_Q,  F_Q \ran_{\Omega^2_\fg}^\eta\big), 
\end{equation}
where  $M$ is a bounded domain in spacetime $\R^{3+1}$ equipped with the Minkowski metric $\eta$ of signature $(-,+,+,+)$, $\lambda$ and $\vphi_0$ are positive parameters, and 
 the remaining symbols are defined as follows:

 $\n_{Q}$ is  the covariant derivative mapping $\C^2$-valued functions (sections) into $\C^2$-valued one-forms defined as 
 \begin{equation}\n_{Q} =d +Q, 
 \end{equation}  with 
  $d$,  the exterior derivative; 
 
  $F_{Q}$ is the curvature $2$-form of the connection one-form $Q$, 
  given by 
\begin{equation}\label{FQ}
 F_Q = dQ + \frac{1}{2g} [Q,Q], 
\end{equation}
 where $[A, B]$ is defined in local coordinates $\{x^i\}$ as
\begin{align}\label{commut-def'} [A, B]:=  [A_i,  B_j] dx^i\wedge dx^j = [B,A],\end{align}
with $A= A_i dx^i$ and $B=B_i dx^i$; 

\quad $\Omega^p_U\equiv  U\otimes\Omega^p$ denotes the space of $U$-valued $p$-forms with the Minkowski, indefinite inner product, 
\begin{equation} \label{inner-prod-ext'}
 \langle A,B \rangle_{\Omega^p_U}^\eta:=  \langle A_{\alpha}(x), B^{\alpha}(x) \rangle_U, 
\end{equation}
where $A=A_{\alpha}(x) dx^{\alpha}$ and $B=B_{\alpha}(x) dx^{\alpha}$ are $U$
-valued $p$-forms, $\alpha$ is a $p$-form index and $\langle \cdot, \cdot \rangle_U$ 
is the standard, positive definite inner product on $U$ with the indices raised and lowered with help of the Minkowski metric $\eta$ on $M$. For instance, for $U = \mathfrak{s u}(2)$, the inner product is given by 
\begin{align}\label{u2-inner-prod}
		\langle A,B\rangle_{\Omega^p_{\mathfrak{s u}(2)}}^\eta := 2\Tr(A_{\alpha}(x)^*B^{\alpha}(x)) = -2\Tr(A_{\alpha}(x)B^{\alpha}(x)).
	\end{align}

Solutions of the no-particle WS equations solve also the full WS system as well as that for the standard model of the particle physics.

 The vacuum sector of the Weinberg-Salam (WS) model consists of static, no-particle solutions.
The static  Higgs and $SU(2)$ and $U(1)$ gauge fields $\Phi$, $ V$ and $ X$ are now defined on the physical space $\R^3$  with the same respective values as in the time-dependent case.  Geometrically, $V, X$ and $Q$ can be thought of as connection one-forms on the trivial bundles $\R^3\times SU(2), \R^3\times U(1)$ and $\R^3\times U(2)$.

The  fields $\Phi$, $ V$ and $ X$ satisfy the static no-particle WS equations, which are the  Euler-Lagrange equations for the static WS energy functional originating in \eqref{WS-act}\footnote{For a discussion of the the time-dependent theory and a derivation of the energy functional \eqref{WS-energy0} see \cite{JT}, \cite{Mir}, \cite{Rub}, \cite{Schw} and Appendix \ref{sec:YM}.}
\begin{equation} \label{WS-energy0}
E_{N} (Q, \Phi) := \int_{N}\big(
 \| \n_{Q}\Phi \|_{\Omega^1_{\C^2}}^{2} + \frac{1}{2} \lam (\|\Phi\|_{\C^2}^2-\vphi_0^2)^2 + \frac{1}{2}\| F_Q\|_{\Omega^2_{\mathfrak{u}(2)}}^2\big), 
\end{equation}
where $N$ is a bounded domain in $\R^3$ with appropriate boundary conditions (specified in \eqref{gauge-per} below) and  $\|\cdot\|_{\Omega^p_{U}}$ is the standard norm 
on the space $\Omega^p_{U} := U \otimes \Omega^p$ 
  of $U$-valued $p$-forms at  $x\in N$ (e.g. for $B=B_i(x)dx^i\in \Omega^1_{U}$, we have $\|B\|_{\Omega^1_{U}}:=(\sum_i\|B_i(x)\|^2_U)^{1/2}$ with the usual Euclidean metric and with the  indices running through $1, 2, 3$), while now, \eqref{inner-prod-ext'} (and \eqref{u2-inner-prod}) become the usual inner products. The symbols $\n_{Q}$ and $F_{Q}$ are as defined above but without the time component.

Since $Q=gV+g'X$ and $X$ has the values in the centre, $u(1)$, of the algebra $u(2)$, 
we have $F_{Q} = g F_{V} + g' F_{X}$, where
\begin{equation}F_V := 
dV + \frac {g}2 [V,V]\ \text{ and }\ F_X := 
 dX\end{equation} are the curvatures of the connections $V$ and $X$\footnote{For more discussion of covariant derivatives and their curvatures, see Appendices \ref{sec:Cov-Deriv-Curv} for the general case, or Appendix \eqref{sec:WS-coord}, 
  for the case of the gauge group $G = U(2)$.} 
 and $\| F_{Q}\|^2_{\Omega^2_{\mathfrak{u}(2)}}=\| F_{V}\|^2_{\Omega^2_{\mathfrak{u}(2)}} +\| F_{X}\|^2_{\Omega^2_{\mathfrak{u}(1)}}$. 

 We introduce  the covariant derivative  $d_{Q}$  mapping $\mathfrak{u}(2)$-valued $k$-forms into $\mathfrak{u}(2)$-valued $(k+1)$-forms,  $k\ge 1$, as
 \begin{equation}d_Q B  :=  d B +   [Q, B ]= d_V B :=  d B + g  [V, B ] .\end{equation} 
 This formula originates in the equation $(\del_Q F_Q)( B)=d_Q B$, where $\del_Q$ is the G\^ateaux derivative with respect to $Q$. For $0$-forms, we set $d_Q=\n_Q$.

The Euler-Lagrange equations for energy functional \eqref{WS-energy0} 
are given by (see Appendix \ref{sec:YM}\footnote{These equations could be converted formally back into the time-dependent ones by taking the adjoints in 
 the Minkowski metric instead of the Euclidian one, see \eqref{YMH-eqs-psi}-\eqref{YMH-eqs-A},  Appendix \ref{sec:YM}.})
\begin{align} \label{WS-eq1}
  &  \n_{Q}^* \n_{Q}\Phi  = \lam (\vphi_0^2 - \|\Phi\|^2)\Phi,\\   
  \label{WS-eq2}   &d_{Q}^* F_Q = J(Q, \Phi), 
\end{align}
 where $\n_{Q}^*$ 
  is the adjoint of $\n_{Q}$ 
   and maps $\C^2$-valued one-forms into $\C^2$-valued functions,  $d_{Q}^*$ 
   is the adjoint of $d_{Q}$ 
    and maps $\mathfrak{u}(2)$-valued two-forms into $\mathfrak{u}(2)$-valued one-forms, and  $J(Q, \Phi)$ is the electroweak current, which is the $\mathfrak{u}(2)$-valued one-form given by
\begin{align}\label{curr}
J(Q, \Phi) &:= 
-\frac{ig}2\tau_a \im \langle \tau_a\Phi,  \n_{Q}\Phi\rangle - \frac{ig'}{2}\tau_0\im\langle \tau_0\Phi, \n_{Q}\Phi\rangle,
\end{align}
where summing over repeated indices is understood, $\tau_0:= \one$ and $\tau_a, a=1, 2, 3,$ are the Pauli matrices,
 \begin{align}\label{pauli-m}  \tau_1:= \left(\begin{array}{cc}
  0 & 1 \\ 1 & 0 \end{array} \right),\
  \tau_2:= \left(\begin{array}{cc}
  0 & -i \\ i & 0 \end{array} \right),\
  \tau_3:= \left(\begin{array}{cc}
  1 & 0 \\ 0 & -1 \end{array} \right).
 \end{align}
 (The Pauli matrices, multiplied by $-i/2$, form an orthonormal basis in $\mathfrak{su}(2)$ with the inner product $\langle g,h\rangle_{\mathfrak{su}(2)} := 2\Tr(g^* h) = -2\Tr(gh)$.) We call system \eqref{WS-eq1}-\eqref{WS-eq2} the (static) WS equations.

The energy functional \eqref{WS-energy0} and  Euler-Lagrange equations \eqref{WS-eq1} - \eqref{WS-eq2} are invariant under the group of rigid motions and the gauge transformations (gauge symmetry)
\begin{align}\label{gauge-transf-gen} 
&(V(x),  X(x),  \Phi(x)) \mapsto (V_\g(x),  X_\g(x),  \Phi_\g(x)),
\end{align}
where $\g=\g(x)=h_1(x)h_2(x)$, with $ \,h_1(x)\in SU(2),\ h_2(x)\in U(1),$ and 
\begin{equation} \notag 
	\left\{ 
	\begin{array}{ll}
    V(x) \mapsto h_1(x)  V(x) h_1^{-1}(x) - i\frac{2}{g}  h_1(x) d h_1^{-1}(x), \\
	 X(x) \ra X(x) - i\frac{2}{g'} h_2(x) d h_2^{-1}(x),\\ 
\Phi(x) \ra  h_1(x)  h_2(x)\Phi(x).&
	\end{array} \right.
\end{equation}
\DETAILS{
\begin{align}\label{gauge-transf-gen} \notag   
&(V(x),  X(x),  \Phi(x)) \\ \notag \mapsto &\big(h_1(x)  V(x) h_1^{-1}(x) - i\frac{2}{g}  h_1(x) d h_1^{-1}(x), \,
  X(x) - i\frac{2}{g'} h_2(x) d h_2^{-1}(x),\, h_1(x)  h_2(x)\Phi(x)\big), \\  
\forall &\,h_1(x)\in SU(2),\ h_2(x)\in U(1).
\end{align}
}
\DETAILS{\begin{align}\label{gauge-transf-V-gen} 
&V(x) \mapsto h_1(x)  V(x) h_1^{-1}(x) - i\frac{2}{g}  h_1(x) d h_1^{-1}(x),\\ 
& X(x) \ra X(x) - i\frac{2}{g'} h_2(x) d h_2^{-1}(x),\\
 &\Phi(x) \ra  h_1(x)  h_2(x)\Phi(x), 
\end{align}
for every $ \,h_1(x)\in SU(2),\ h_2(x)\in U(1).$}
\DETAILS{\noindent the translations
\begin{equation}\label{translations}
 T^{transl}_s :    (\Psi(x), A(x)) \mapsto (\Psi(x + s), A(x + s)),\qquad \forall t \in \R^2;
\end{equation}
\noindent the rotations and reflections,
\begin{align}\label{rotations-reflections}
 T^{rot}_R :    (\Psi(x), A(x)) \mapsto (\Psi(R^{-1}x),  RA(R^{-1}x)),\qquad \forall R \in  O(n) .
\end{align}}

The physical quantities here are (a) the Higgs field density $\|\Phi\|$, (b) the magnetic field $\Tr F_{Q}$ and (c) the YM current $J(Q, \Phi)$. It is easy to check that these quantities are gauge invariant. We say that a solution $(Q, \Phi)$ to \eqref{WS-eq1}-\eqref{WS-eq2} is {\it homogeneous} if $\|\Phi\|$,  $\Tr F_{Q}$ and $J(Q, \Phi)$ are independent of $x$. (We say that $\Tr F_{Q}$ is independent of $x$, if it is a multiple of a constant $2$-form, see \eqref{curv-const}.) Otherwise, we say that  $(Q, \Phi)$  is {\it inhomogeneous}.

Furthermore, we say  that a solution $(Q, \Phi)$  is {\it gauge-translation invariant} if it is invariant under translations up to gauge transformations.

Clearly, a solution $(Q, \Phi)$ which is gauge-translation invariant is also homogeneous. The converse in general might not be true.

We are interested in  the vacuum solutions of the  WS equations with a non-vanishing average magnetic field, 
\[\lim_{R\ra \R^3}\frac1{|R|}\int_R\Tr F_Q=-i e\sum_{(ijk)}b_i dx^j\wedge dx^k, 
  \] 
 i.e. solutions minimizing  the WS energy locally under the constraint above. 
  In physical field theories, one expects the vacua to have the maximal available symmetry. Consequently, we first 
 consider gauge-translation invariant solutions with a fixed (constant) magnetic field. 

For $\vec b=(b_1, b_2, b_3)\neq 0$, 
 Eqs. \eqref{WS-eq1} - \eqref{WS-eq2} have the gauge-translation invariant 
 solution given (up to a gauge symmetry)  by
 \begin{equation}\label{vac-hom}
 U^{\vec b}_* := (Q^{\vec b}, 
 \Phi^{\vec b}),
 \end{equation}
  where $\Phi^{\vec b}$ is a constant field 
 and $Q^{\vec b}$ is a connection with a constant magnetic field 
 \begin{equation}\label{curv-const}\Tr F_{Q^{\vec b}}=-i e\om_{\vec b},\ \text{ where }\ \om_{\vec b}:=\sum_{(ijk)}b_i dx^j\wedge dx^k,\end{equation}
  with the sum taken over all cyclic permutations of $(1, 2, 3)$, and  $e:= 
  \frac{g g'}{\sqrt{g^2+{g'}^2}}.$ 
 ($e$ turns out to be the electron charge.) 
 We specify this solution at the end of this section in equations \eqref{Qb} and \eqref{FQb}. (For it, $Q^{\vec b}$ solves the YM equation $d_{Q}^* F_Q =   0$.) 
  \DETAILS{ $A^b(x)$ is a magnetic potential of the constant magnetic field of strength $b$ 
   and $\theta$ is \emph{Weinberg's angle}, given by $\tan\theta=g'/g$.\footnote{Indeed, $d_{Q}\Phi_0 = (g V+ g' X)\Phi_0 = ( g A^b \sin \theta \tau_{3}+ g' A^b\cos \theta\tau_{0})\Phi_0 = g' A^b\cos \theta(\tau_{3}+ \tau_{0})\Phi_0$. Since $ (\tau_3 + \tau_0) \Phi_0 = 0$, this implies $d_{Q}\Phi_0=0$. From $d_Q\Phi_0 = 0$, it is easy to see that \eqref{vac-true} solves \eqref{WS-eq1} - \eqref{WS-eq2}.} This solution is  gauge-translationally invariant, i.e. invariant under translations up to a gauge symmetry. 
  It corresponds to  the `total vacuum' 
with the constant magnetic field $d A^b(x)=\sum_{(ijk)}b_i dx^j\wedge dx^k$.} 

\DETAILS{In what follows, we 
 look for 
  solutions with a fixed average magnetic field (trace of the YM magnetic field, or curvature) 
   \[\lim_{R\ra \R^3}\frac1{|R|}\int_R\Tr F_Q=-i e\om_{\vec b}. 
  \] 
   This}
   Fixing the  average magnetic field breaks the full special Euclidean symmetry (i.e. translations and rotations but not reflections) but maintains the special Euclidean symmetry in the plane orthogonal to $\vec b$ and the translational symmetry along $\vec b$. 
      Looking for the simplest non-trivial solutions, we consider
 solutions which {\it do not depend on the coordinate along $\vec b$} and look for solutions spontaneously breaking the transversal translational symmetry. 

  With the notation $b=|\vec b|$, 
  we show that for appropriate perturbations: 
 \begin{enumerate}[(i)]
 \item \eqref{vac-hom} is linearly stable for $b<b_*$ and  unstable for $b>b_*$, where $b_*:=  g^2\varphi_0^2/2e$;
 \item At $b=b_*$, a new inhomogeneous solution (breaking the gauge-translational invariance) 
 bifurcates, and this solution has the discrete translational symmetry of a lattice in the plane orthogonal to $\vec b$ and has lower energy per unit area; 
  \item The lattice shape minimizing the energy per unit area approaches the hexagonal lattice as $b$ approaches $b_*$.
 \end{enumerate}
  
To formulate these results precisely, we introduce some definitions. Since we consider solutions which {\it do not depend on the coordinate along $\vec b$},  we can restrict our analysis to the plane $\perp\vec b$. We choose the $x^3$-axis along $\vec b$ and identify the plane $\perp\vec b$ with $\R^2$.

We fix a lattice $\cL$ in $\R^2$ 
and say a triple $( \Phi(x), V(x), X(x))$ is
 $\cL$-\emph{gauge-periodic}, or, $\cL$\emph{-equivariant}, 
   if and only if it satisfies the equation
 \begin{equation}\label{gauge-per}
	(T^{gauge}_{\g_s})^{-1} T^{trans}_s (V, X, \Phi) = (V, X, \Phi),\ \quad \forall s\in \cL,
\end{equation}
for some $\g_s \in C^1(\R^2, SU(2)\times U(1))$. 
Here $T^{gauge}_{\g }$ is given by \eqref{gauge-transf-gen} 
  and $T^{trans}_s$ is the group of translations, $T^{trans}_s f(x)=f(x+s)$. 
(When $\cL$ is clear, we omit it from the definition above.) 

We denote by $\cH^s_{\cL},\ s\in \N,$ the  
Sobolev space of $\cL$-equivariant triples $U\equiv (V, X, \Phi)$ on $\R^2$, with the norm 
  \begin{align}\label{Sob-norm}
  	\|U\|_{\mathcal{H}_{\mathcal{L}}^s} := \Big(\frac{1}{|\Omega|}\sum_{k=0}^s\int_{\Omega} \|d_Q^k U\|^2 
	\Big)^{\frac12},
  \end{align}
where $\Omega$ is an arbitrary fundamental domain of $\cL$, $d_Q^k$ is the $k$-th iterate of the covariant derivative $d_Q$ and $ \|\cdot\|$ is the (fiber) norm in the space $\Omega^{k+1}_{\mathfrak{s u}(2)}\times \Omega^{k+1}_{\mathfrak{u}(1)}\times \Omega^k_{\C^2}$, see \eqref{inner-prod-ext'}	
 (and \eqref{u2-inner-prod}), and  with corresponding the inner product. 
 Note that $L^2_{\cL}=\cH^0_{\cL}$. 
 
 The resulting  Sobolev spaces $\mathcal{H}_{\mathcal{L}}^{s}$ are independent (up to isomorphism) of the choice of the fundamental domain, $\Omega$.  All  Sobolev embedding theorems are valid for $\mathcal{H}_{\mathcal{L}}^s$. They can be proven by passing to a vector bundle over the torus $\R^2/\cL$ and then to the local charts and then using standard Sobolev embedding theorems. By the Sobolev embedding $\mathcal{H}_{\mathcal{L}}^{1}\subset L_{\mathcal{L}}^{p},\ p<\infty$, and the definitions \ref{WS-energy0} and \ref{Sob-norm},  
 \begin{align}\label{en-finite} E_{\Omega} (Q, \Phi)<\infty\ \text{ on }\ \mathcal{H}_{\mathcal{L}}^1  \end{align} (recall that $\Omega\subset \R^2$) and is independent of a choice of $\Omega$.

We say a solution $U_*:=  (V_*, X_*, \Phi_*)$ of the WS system \eqref{WS-eq1} - \eqref{WS-eq2} is \emph{energetically stable} 
if and only if it is a local minimum of the WS energy $E_{N} $, in the sense that the spectrum of the \emph{$L^2$-Hessian} of $E_{N} $ 
 at $U_*$ on  $L^2_{\cL}$ (which is real) is non-negative. 
  $U_*$ is said to be \emph{unstable} if it is a saddle point of $E_{N} $ (so that the spectrum of its Hessian has a negative part).  

For an $\cL$-equivariant triple $U$ and a fundamental domain $\Omega$ of $\cL$, we define the energy per fundamental cell by
\begin{equation}\label{En-per}E^{\cL}(U) :=  \frac{1}{|\Omega|} E_{\Omega} (U),\end{equation}
where $|\Omega|$ denotes the area of $\Omega$. This energy is independent of the choice of $\Omega$. 

In what follows,  $\Omega$ denotes an arbitrary (but fixed throughout) fundamental domain of $\cL$, and $|\cL|$, the area of a fundamental cell of $\cL$, which is independent of the choice of the cell $\Omega$ 
  and is called the covolume of $\cL$.  

 Let $M_W :=   \frac1{\sqrt{2}} g\varphi_0$, $M_Z :=   \frac1{\sqrt{2}\cos\theta} g\varphi_0$ and $M_H :=  \sqrt{2}\lambda\varphi_0$, where $\theta$ is the Weinberg angle defined by $ \cos \theta=
  \frac{g }{\sqrt{g^2+{g'}^2}}$. These are the masses of the W, Z and Higgs bosons, respectively 
(this nomenclature will be explained in the discussion  following Eq. \eqref{WS-energy'}). 
  Finally, let
\begin{equation} \label{b-crit}
 b_*:=  \frac{g^2\varphi_0^2}{2e}=\frac{M_W^2}{e},\ \quad e:=g\sin \theta.
\end{equation}

With the above definitions, we will
prove the following:

\begin{theorem}\label{thm:normal-instab} The gauge-translational invariant 
 solution \eqref{vac-hom} is energetically stable for $b<b_*$ and  unstable for $b>b_*$. \end{theorem}

\begin{theorem}\label{thm:AL-exist'}
Let $\cL$ be a  lattice satisfying $0 < 1-\frac{M_W^2}{2\pi}|\cL| \ll 1$ and assume that 
$M_Z < M_H$.\footnote{This assumption is justified experimentally since $M_Z = 91.1876 \pm 0.0021 GeV/c^2$ \cite{mz} and $M_H = 125.09 \pm 0.31 GeV/c^2$ \cite{mh}} Then there exist $\delta >0$ such that
 \DETAILS{for every lattice $\cL$ satisfying    
\begin{equation}\label{LAT-cond} |\Omega| = \frac{2\pi n}{eb} 
\end{equation}
with $n=1$,}the following holds:
	\begin{enumerate}[(a)]
\item Equations \eqref{WS-eq1} - \eqref{WS-eq2} have an inhomogeneous 
 solution $U_{\cL}\in \cH^2_{\cL}$ 
in the $\del$-ball 	$B_{\mathcal{H}_{\mathcal{L}}^2}( U_*^{\vec{b}}; \delta)$ in $\mathcal{H}_{\mathcal{L}}^2$ around the homogeneous solution \eqref{vac-hom};  
	\item $U_{\cL}$ 
is the unique, up to gauge symmetry transformation, inhomogeneous solution in the $\del$-ball $B_{\mathcal{H}_{\mathcal{L}}^2}( U_*^{\vec{b}}; \delta)$; 
	\item $U_{\cL}$ has energy per unit area less than vacuum solution \eqref{vac-hom}: $E^{\cL}(U_{\cL}) < E^{\cL}(U^b_*)$.
	\end{enumerate}	
\end{theorem}	
\DETAILS{The next result uses the fact that all the lattices in $\R^2$ are parameterized by the shape parameter $\tau$ in the fundamental domain, $\bH/SL(2, \Z)$, of  the modular group  $SL(2, \Z)$ acting on   the Poincar\'e half-plane  $\bH$. 
Explicitly  
\begin{align}\label{fund-domSL2Z} \{\tau\in \C: \Im\tau > 0,\ |\tau| \geq 1,\ -\frac{1}{2} < \Re\tau \leq \frac{1}{2} \}. 
\end{align} 
(See Appendix \ref{sec:ls} for more details.) This gives the space of lattices a topology.} 

The solutions described in this theorem can be reinterpreted geometrically as representing  sections ($ \Phi(x)$) and  connections ($(V(x), X(x))$) on a $U(2)$ vector bundle over a torus (cf. \cite{CERS}). 
 However, 
 a vector bundles over a torus is   topologically 
equivalent 
to a direct sum of line bundles. In our case, this equivalence follows from equations \eqref{Z,A-fields} - \eqref{gauge-transf'} below.

\DETAILS{ For the next result, we introduce the standard  parameterization of lattices in $\R^2$. Identifying $\R^2$ with $\C$ via $(x_1,x_2) \leftrightarrow x_1 + i x_2$, we can view a lattice $\cL\subset\R^2$ as a subset of $\C$. It is a well-known fact (see e.g. \cite{Ahlfors}) that any lattice $\cL\subset\C$ can be given a basis $r,r'$ such that the ratio $\tau = \frac{r'}{r}$ belongs to the set
\begin{align}\label{fund-domSL2Z} \{\tau\in \C: \Im\tau > 0,\ |\tau| \geq 1,\ -\frac{1}{2} < \Re\tau \leq \frac{1}{2} \}, 
\end{align} 
which  is  the fundamental domain, $\bH/SL(2, \Z)$, of  the modular group  $SL(2, \Z)$ acting on the Poincar\'e half-plane  $\bH$.   
 For a given $\cL$, the parameter $\tau$ is unique and 
 is used as a  parameterization (up to scaling) of the lattices.  This gives the space of (normalized) lattices a topology.}

 For the next result, we use the topology on the space of (normalized) lattices induced by the standard  parameterization of lattices defined as follows. 
 Identifying $\R^2$ with $\C$ via $(x_1,x_2) \leftrightarrow x_1 + i x_2$ and viewing a lattice $\cL\subset\R^2$ as a subset of $\C$ and using 
  a translation and a rotation, any lattice $\cL\subset\C$ can be reduced to the form $\cL=r\cL_\tau$, where $r>0$, $\cL_\tau:=\Z+\tau \Z$ and $\tau \in \bH:=\{\tau\in \C: \Im\tau > 0 \}$. 
Furthermore, any two $\tau$'s produce the same lattice iff they are related by an element the modular group  $SL(2, \Z)$ acting on the Poincar\'e half-plane  $\bH$ (see e.g. \cite{Ahlfors}). Hence, it suffices to restrict $\tau$ to   the fundamental domain of $SL(2, \Z)$, 
 \begin{align}\label{fund-domSL2Z} \bH/SL(2, \Z)=\{\tau\in \bH: 
 |\tau| \geq 1,\ -\frac{1}{2} < \Re\tau \leq \frac{1}{2} \}. 
\end{align} 


	\begin{theorem}\label{thm:lattice-shape} For $M_Z < M_H$, the lattice 
$\cL_*$ minimizing the average energy, $E^{\cL}(U_{\cL})$, 
	 approaches the hexagonal lattice $\cL_{\rm hex}$ as $b \to b_*$ in the sense that the shape parameter $\tau_*$ of the lattice $\cL_*$ approaches $\tau_{\rm hex} = e^{i\pi/3}$ in $\mathbb{C}$.
\end{theorem} 

\DETAILS
{\bf(The function $e(\tau)\equiv E^{\cL}(U_{\cL})$ is defined originally on  the fundamental domain, $\bH/SL(2, \Z)$, of  $SL(2, \Z)$. (It can be extended by periodicity to  the Poincar\'e half-plane  $\bH$. The extension is invariant under  the modular group  $SL(2, \Z)$ acting on $\bH$.) Hence, we expect that $e(\tau)$ has only two critical points: at $\tau = e^{i\pi/3}$ and  $\tau = e^{i\pi/2}$. This does not contradict Theorem \ref{thm:lattice-shape}, but makes it more precise. 
}

 \DETAILS{
$\cL$-equivariant functions and one-forms are in one-to-one correspondence with sections and connections on a line bundle over the torus $\R^2/\cL$.  We will see that such a vector bundle factorizes into  a line bundle and a trivial bundle.
}

Now, we construct  explicitly the solution \eqref{vac-hom}. We define 
 \begin{align}\label{Qb}
 Q^{\vec b}=(-\frac{i}{2}\tau_{3} A^{\vec b} \sin \theta,  -\frac{i}{2}\tau_{0} A^{\vec b}\cos \theta)\ \text{ and }\ \Phi^{\vec b}\equiv\Phi_0:= (0, \vphi_0),
 \end{align}
  where  $A^{\vec b}(x)$ be a ($U(1)$-) magnetic potential of the constant magnetic field $d A^{\vec b}=\om_{\vec b}$ and $\theta$ is \emph{Weinberg's angle}, given by $\tan\theta=g'/g$. 
   We have
  \begin{lemma}\label{eqs-equiv} 
The pair $(Q^{\vec b}, \Phi_0)$ satisfies \eqref{WS-eq1} - \eqref{WS-eq2}. Moreover, the connection $Q^{\vec b}$ 
has the constant curvature 
\begin{align}\label{FQb}F_{Q^{\vec b}}=-\frac{i}{2}e (\tau_{3} +\tau_{0})\om_{\vec b},\ \text{ (with the magnetic field }\ \Tr F_{Q^{\vec b}}=-ie\om_{\vec b}). \end{align}
 \end{lemma}

\begin{proof}\eqref{FQb} follows easily from $d A^{\vec b}=\om_{\vec b}$.     
 To check that $(Q^{\vec b}, \Phi_0)$ satisfies \eqref{WS-eq1} - \eqref{WS-eq2}, we observe that $d_{Q^{\vec b}}\Phi_0 = (g V^{\vec b}+ g' X^{\vec b})\Phi_0 = ( g A^{\vec b} \sin \theta \tau_{3}+ g' A^{\vec b}\cos \theta\tau_{0})\Phi_0 = eA^{\vec b} (\tau_{3}+ \tau_{0})\Phi_0$. Since $ (\tau_3 + \tau_0) \Phi_0 = 0$, this implies $\n_{Q}\Phi_0=0$. This gives \eqref{WS-eq1} and reduces \eqref{WS-eq2} to   $d_{Q^{\vec b}}^* F_{Q^{\vec b}} = 0$, which follows easily from \eqref{FQb}.  \end{proof}
 
\DETAILS{Furthermore, using that $d A^{\vec b}=\om_{\vec b}$, one shows readily 
    that $F_{Q^{\vec b}}=-\frac{i}{2}g'\cos \theta(\tau_{3} +\tau_{0})\om_{\vec b}$, which implies the latter equation giving \eqref{WS-eq2}.} 

 Our approach is based on a careful examination of the linearization of the WS equations on the homogeneous vacuum. The spectrum of the linearized problem determines the domains of the linear, or energetic, stability and the transition threshold. In the instability domain, we apply an equivariant bifurcation theory. This gives  Theorem \ref{thm:AL-exist'}(a) and (b). For Theorems \ref{thm:AL-exist'}(c) and \ref{thm:lattice-shape}, we carefully study the asymptotic behaviour of the energy functions for small values of the bifurcation parameter.  

\section{Gauge fixing and $W$ and $Z$ bosons} \label{sec:WZ-bosons}

In this section, we choose a particular gauge and pass from the fields (one-forms) $V$ and $X$ to more suitable gauge fields. 
\DETAILS{  A choice of a vacuum state $\Phi=\Phi_{\rm vac},$ a constant vector with $ | \Phi_{\rm vac}|^2=\vphi_0^2,$ breaks the $U(2)\approx SU(2)\times U(1)$ symmetry to a $U(1)$ one. Take for simplicity the gauge $\Phi_{\rm vac}=(0, \vphi)$, with $\vphi= \vphi_0$ real.}
We eliminate a part of the gauge freedom by assuming that the Higgs field $\Phi$ is of the form \begin{equation}\Phi=(0, \vphi),\end{equation} with $\vphi$ real (this can be done using only the $SU(2)$ part of the gauge group). 
Then
\begin{align} \label{sym-break}
\tau_a \Phi\ne 0,\ a=0, 1, 2, 3,
\end{align}
where, recall,  $\tau_a,\ a=1, 2, 3,$ are the Pauli matrices generating the Lie algebra $su(2)$, and $\tau_0=\one$. However, there is one linear combination of $\tau_a$'s (unique up to a scalar multiple) which annihilates $\Phi$:
\begin{align} \label{sym-surv}
(\tau_3 + \tau_0) \Phi = 0. 
\end{align}
Thus, for the gauge $\Phi=(0, \vphi)$ 
the symmetries generated by $\tau_1, \tau_2, \tau_3 - \tau_0$ are broken and the $U(1)$ symmetry generated by $\tau_3 + \tau_0$ remains unbroken. The unbroken gauge symmetry is given by transformations \eqref{gauge-transf-gen} with 
\begin{align}\label{gauge-choice} 
h_1(x) := & e^{-\frac i2 \g (x)\tau_3}\in SU(2),\ h_2(x) :=  e^{-\frac i2 \g (x) \tau_0}\in U(1), 
\end{align}
where $\g\in C^1(\R^3, \R)$.

Continuing in the gauge  $\Phi=(0, \vphi)$ and writing $V= -\frac i2 \tau_a V^a$\footnote{Note that the lower indices $i, j, k$, as in  $A=A_{i} dx^{i}$, refer to vectorial components and  run through $1, 2$, while the upper indices $a, b, c$, as in  $V= -\frac i2 \tau_a V^a$, refer to $U(2)$-algebra components.} and $X = -\frac i2 \tau_0 X^0$, where $X^0$ and $V^a, a=1, 2, 3,$ are real fields (since $V$ takes values in $su(2)$ and therefore $V^*=-V$), 
we pass to the new fields corresponding to the broken and unbroken generators, $
\tau_3 - \tau_0$ and $\tau_3 + \tau_0$, respectively:
\begin{align} \label{Z,A-fields}
Z  = V^3 \cos \theta - X^0 \sin \theta {\hbox{\quad and \quad}}
A = V^3  \sin \theta + X^0 \cos \theta,
\end{align}
where, recall, $\theta$ is Weinberg's angle, defined by  $\tan\theta=g'/g$. 
Note that $Z$ and $A$ are real fields 
Moreover, it is convenient to pass from the remaining two components, $V^1, V^2$, of $V$ to a single complex field 
\begin{align} \label{W-field}
W  =\frac{1}{\sqrt 2}  (V^1 - i V^2).
\end{align}

The gauge invariance 
 of the original field equations with the unbroken gauge symmetry  given by transformations \eqref{gauge-transf-gen} with \eqref{gauge-choice} leads to the invariance under following gauge transformations:
\begin{align}\label{gauge-transf'} 
\tilde T^{gauge}_\g :   (W, A, Z, \vphi) \mapsto (e^{i \g } W,  A -\frac{1}{e} d \g, Z, \vphi),
\end{align}
for $ \g\in C^1(\R^3, \R)$, where   $e^{i \g } W=\sum e^{i \g } W_i d x^i$ for $W=\sum W_i d x^i$, $e$ 
 is the electron charge. 
Here, we replaced $\Phi:= (0,  \vphi)$ by $\vphi$.

 The WS energy in terms of $W, Z, A$ and $\vphi$ fields in 3D is given in \eqref{WS-energy}, Appendix \ref{sec:en-expl}. The WS equations in terms of $W, Z, A$ and $\vphi$ in 3D can be found by taking variational derivatives of this energy w.r.to different fields.

In terms of $W, A, Z$ and $\vphi$ fields, the vacua \eqref{vac-hom} of the Weinberg-Salam model become (up to a gauge symmetry): 
\begin{equation}\label{vac} 
 (0, A^b(x), 0, \vphi_0),
\end{equation}
where, recall, $A^b(x)$ is a magnetic potential for the constant magnetic field of strength $b$ in the $x^3$-direction, $d A^b(x)=b dx_1\wedge dx_2$, and $\vphi_0$ is a positive constant from \eqref{WS-energy0}. 
We choose the gauge so that $A^b(x)$ is of the form
\begin{equation}\label{vac'}
A^b(x) = \frac b2 
(-x_2 dx_1 + x_1 dx_2). 
\end{equation} 

We will show that for a large magnetic field $b$, these homogeneous vacua become unstable  and new, inhomogeneous vacua emerge from them. This is a bifurcation problem from the branch of gauge-translationally invariant (homogeneous) solutions, \eqref{vac}. 



Since we consider the WS system with the fields independent of the third dimension $x^3$, i.e.  in $\R^2$, we can choose the gauge with $V_3 = X_3=0$ (and hence $W_3 = A_3=Z_3=0$). 

Also, we will work in a fixed coordinate system, $\{x^i\}_{i=1}^2$ and write the fields as $W = W_{i} dx^i,$ $ Z =  Z_{i} dx^i$  and $A =  A_{i} dx^i$. 
For ease of comparing our arguments with earlier results, and given that we use the standard Euclidean metric in $\R^2$, we identify (complex) one-forms  $W, Z$ and $A$ with the  (complex) vector fields $(W_1, W_2), (Z_1, Z_2)$ and $(A_1, A_2)$. With this,
we  show in Appendix \ref{sec:WSeqs2D} that in this case,  WS energy functional \eqref{WS-energy0} can be written as 
\begin{align} \label{WS-energy'}
E_{\Omega}(W, A, Z, \vphi) = &\int_{\Omega }  \big[ |\curl_{gV^3} W|^2 +  \frac{1}{2}  |\curl Z|^2 +\frac{1}{2} |\curl A|^2  \notag \\ 
&+\frac{1}{2} g^2 \vphi^2|W|^2+\frac12\kappa g^2 \vphi^2|Z|^2 \notag +\frac{g^2}{2} |\overline W\times W|^2 \notag  \\ 
&+ ig(\curl V^3)\overline W\times W +  |\n \vphi|^2  + \frac{1}{2} \lam (\vphi^2-\vphi_0^2)^2\big], 
\end{align}
where $\kappa := \frac{g^2}{2\cos^2\theta}$, $\curl_{U} W:= (\n_U)_1 W_2 - (\n_U)_2 W_1$, $(\n_U)_i :=  \partial_i -  i U_i $, $\p_i \equiv \p_{x^i}$ 
(for a $\mathfrak{u}(1)-$valued vector-field $U$), $\xi\times \eta:=\xi_1 \eta_2- \xi_2 \eta_1$ and 
$\curl V^3 := \p_1 V^3_2 - \p_2 V^3_1$. 
It follows from \eqref{Z,A-fields} that $V^3=  Z \cos \theta + A \sin \theta$.

Expanding \eqref{WS-energy'} in $\vphi$ around $\vphi_0$, we see that the $W$, $Z$ and $\phi$ (Higgs) fields have the masses $M_W := \frac{1}{\sqrt{2}} g \vphi_0$, $M_Z := \frac{1}{\sqrt 2\cos \theta}g \vphi_0$ and $M_H = \sqrt{2}\lambda\varphi_0$, respectively.

Using the relation $\xi\times \eta = J\xi \cdot \eta$, where $\cdot$ denotes the Euclidean scalar product in $\R^2$ and $J$ is the symplectic matrix, 
\begin{align}\label{J}
J := \begin{pmatrix}
0 & - 1 \\
1 & 0
\end{pmatrix},
\end{align}  
we find the Euler-Lagrange equations for \eqref{WS-energy'}, which give the WS system \eqref{WS-eq1} - \eqref{WS-eq2} in 2D in terms of the fields $W$, $A$, $Z$ and $\vphi$
\begin{align} \label{WS-eq1'}
&[\curl_{gV^3}^* \curl_{gV^3} + \frac{g^2}{2}\varphi^2 - ig(\curl V^3)J + g^2(\overline{W} \times W)J]W = 0, \\
\label{WS-eq2'} &\curl^* \curl A + 2e\Im[(\curl_{gV^3}W)J\overline{W} - \curl^*(\overline{W}_1 W_2)] = 0, \\
\label{WS-eq3'}  &[\curl^* \curl\! + \kappa\varphi^2] Z  
+ 2g\cos\theta \Im[(\curl_{gV^3}W)J\overline{W} - \curl^*(\overline{W}_1 W_2)] = 0, \\
\label{WS-eq4'} &[-\Delta + \lambda (\varphi^2-\vphi_0^2)+\frac{g^2}{2}|W|^2 + \frac12\kappa|Z|^2]\varphi  = 0,
\end{align}
where, recall, $\kappa = \frac{g^2}{2\cos^2\theta}$, $V^3=  Z \cos \theta + A \sin \theta$ and $\Delta$ is the standard Laplacian. 
 (For a derivation of \eqref{WS-eq1'} - \eqref{WS-eq4'} from \eqref{WS-energy'}, see Appendix \ref{sec:WSeqs2D} and also \cite{MT, SY}.) Of course, \eqref{WS-eq1'} - \eqref{WS-eq4'} can also be derived directly from WS system \eqref{WS-eq1} - \eqref{WS-eq2}.



In terms of the $(W, A, Z, \vphi)$ fields, the lattice gauge - periodicity \eqref{gauge-per} is expressed as
\begin{equation}\label{gauge-per'}
(\tilde T^{gauge}_{\g_s})^{-1} T^{trans}_s (W, A, Z, \vphi) = (W, A, Z, \vphi), 
\end{equation}
for all $s\in \cL,$ where $\gamma_s \in C^1(\R^2,\R)$ for all $s\in\cL$, $ \tilde T^{gauge}_\g$ given in \eqref{gauge-transf'} and $T^{trans}_s$ is the group of translations, $T^{trans}_s f(x)=f(x+s)$. We say that $(W,A, Z, \vphi)$ satisfying \eqref{gauge-per'}  is an $\cL$\emph{-equivariant} state. By evaluating the effect of translation by $s+t$ in two different ways, we see that  the family of functions $\g_s$ has the co-cycle property\footnote{A function $\g_s: \cL \times \R^2\ra G$ satisfying the co-cycle property \eqref{cocycle-cond} is called the automorphy exponent and $e^{i \g_s}$, the automorphy factor.}
\begin{align}\label{cocycle-cond}
\g_{s+t}(x) - \g_s(x+t) - \g_t(x) \in 2\pi\Z,\ \quad \forall s, t\in \cL.
\end{align}
Since $T^{trans}_s$ is an Abelian group,  the co-cycle condition \eqref{cocycle-cond} implies that,  for any basis $\{j_1, j_2\}$ in $\cL$,  the quantity 
\begin{equation}\label{cgams}
c(\g_s) =\frac{1}{2\pi} (\g_{j_2}(x+j_1) + \g_{j_1}(x) - \g_{j_1}(x+j_2)  - \g_{j_2}(x)) 
\end{equation}
is independent of $x$ and of the choice of the basis $\{j_1, j_2\}$, and is an integer.  This topological invariant is equal to the degree of the corresponding line bundle.


Using Stokes' Theorem, one can show, for any $A$ satisfying  \eqref{gauge-per'} - \eqref{cgams}, that the magnetic flux through any fundamental domain $\Omega$ of the lattice $\cL$ is quantized: 
\begin{equation}\label{flux-quant}
\frac{e}{2\pi }\int_{\Omega}d A = n,
\end{equation}
where $e$ is defined after \eqref{gauge-transf'} and $n=c(\g_s)\in\Z$ defined in \eqref{cgams}. The left-hand side of \eqref{flux-quant} is called the \emph{Chern number} of the line bundle corresponding to $\g_s$. (We note that $n$ is independent of the choice of $\Omega$.)

The vacuum state \eqref{vac} is $\cL$-equivariant if and only if the magnetic field $b$ is given by the relation 
\begin{equation} \label{b-rel}
b = \frac{2\pi }{ e|\cL|}n,
\end{equation} 
where, by definition, $|\cL|= |\Omega|$ for any fundamental cell $\Omega$. In particular, $b$ is quantized. For such $b$, the vector field $\frac1e A^b$ satisfies \eqref{flux-quant}.

Furthermore, due to the reflection symmetry of the problem, we may assume that $b \geq 0$.
Clearly, we have:
\begin{lemma}\label{eqs-equiv} Equations \eqref{WS-eq1} - \eqref{WS-eq2} for $\cL$\emph{-equivariant} fields \eqref{gauge-per} in the gauge $\Phi=(0, \vphi)$ are equivalent to Equations \eqref{WS-eq1'} - \eqref{WS-eq4'} for $\cL$\emph{-equivariant} fields \eqref{gauge-per'}, with the equivalence realized by the transformation \eqref{Z,A-fields} - \eqref{W-field}. \end{lemma} 

Finally, we use the invariance of \eqref{WS-eq1'} - \eqref{WS-eq4'} under the gauge transformation \eqref{gauge-transf'}  
to choose a convenient gauge for the fields $W(x)$ and $A(x)$. We say that the fields $(W,A, Z, \vphi)$ and $(W',A', Z', \vphi')$ are \emph{gauge-equivalent} if there is $\gamma\in C^1(\R^2, \R)$ such that 
\[(W', A', Z', \vphi') = \tilde T^{gauge}_{\gamma}(W, A, Z, \vphi).\] Clearly, if $(W, A, Z, \vphi)$ and $(W',A', Z', \vphi')$ are gauge-equivalent, then $(W, A, Z, \vphi)$ solves \eqref{WS-eq1'} - \eqref{WS-eq4'} if and only if $(W',A',Z',\vphi')$ solves \eqref{WS-eq1'} - \eqref{WS-eq4'}.
The following proposition was first used in  
\cite{Odeh} and proven in \cite{Tak} (an alternate proof 
is given in Appendix A of \cite{TS1}):

\begin{proposition}\label{prop:gauge-fix} Let $(W',A', Z', \vphi')$ be an $\cL$-equivariant state and let $b$ be given by  \eqref{b-rel}. 
	Then there is a $\cL$-equivariant state $(W,A, Z, \vphi)$, gauge-equivalent to $(W',A', Z', \vphi')$, which satisfies \eqref{gauge-per'}, with $\chi_s(x)=\frac{eb}2 s\wedge x + k_s$, i.e. such that, $\forall s\in\cL$,
	\begin{align} \label{W-gauge-fix}
	&W(x+s)=e^{i(\frac{eb}2 s\wedge x + k_s)} W(x), \\ 
	\label{A-gauge-fix}&A(x+s)=A(x)+\frac b2 Js,\\ 
	\label{divA}&\div A=0, \\ 
	\label{Zvphi-gauge-fix}&Z(x+s)=Z(x), \quad \vphi(x+s)=\vphi(x). 
	\end{align}	Here $k_s$ satisfies the condition $k_{s+t}-k_s-k_t-\frac {eb}2 s\wedge t\in 2\pi\mathbb{Z}$, for all $s,t\in\cL$, the matrix $J$ is given in \eqref{J}.
\end{proposition}

Note that with the gauge \eqref{divA}, the homogeneous vacua \eqref{vac} satisfy \eqref{W-gauge-fix} - \eqref{Zvphi-gauge-fix}.

Our goal is to prove the instability of the vacuum state \eqref{vac} and the existence of $\cL-$equivariant (in the sense of \eqref{gauge-per'})  solutions to transformed WS system \eqref{WS-eq1'} - \eqref{WS-eq4'} having the properties described in Theorems \ref{thm:AL-exist'} and \ref{thm:lattice-shape}.

\section{Rescaling}\label{Sec:rescaling}

In this section, we rescale transformed WS system \eqref{WS-eq1'} - \eqref{WS-eq4'} to keep the lattice size  fixed. Specifically, we define the rescaled fields $(w,a,z,\phi)$ to be
\begin{align} \label{rescaling}
(w(x), a(x), z(x),\phi(x))&:=  (rW(rx),rA(rx),rZ(rx),r\vphi(rx)), \\ \label{r} r&:= \sqrt{\frac{n}{eb}} = \sqrt{\frac{|\Omega|}{2\pi}}.
\end{align}
where in the second equality \eqref{r}, we used  \eqref{b-rel}. Clearly, $(W(x), A(x), Z(x),\vphi(x))$ is $\cL$-equivariant if and only if $(w(x),a(x), z(x),\phi(x))$ is $\cL'$-equivariant, where \[\cL':= \frac{1}{r}\cL.\] Now, the rescaled lattice $\cL'$ is independent of $b$ and the size of a fundamental domain, $\Omega'$, of $\cL'$ is fixed as $|\Omega'|=2\pi$.

Plugging the rescaled fields into \eqref{WS-eq1'} - \eqref{WS-eq4'} gives the rescaled Weinberg-Salem equations: 
\begin{align}
&[\curl_{\nu}^* \curl_{\nu} + \frac{g^2}{2}\phi^2 - i (\curl \nu)J + g^2(\overline{w} \times w)J]w = 0, \label{WS-eq1-resc} \\
&\curl^* \curl a + 2e\Im[(\curl_{\nu}w)J\overline{w} - \curl^*(\overline{w}_1 w_2)] = 0, \label{WS-eq2-resc} \\
&[\curl^* \curl + \kappa\phi^2] z + 2g\cos\theta \Im[(\curl_{\nu}w)J\overline{w} 
- \curl^*(\overline{w}_1 w_2)] = 0, \label{WS-eq3-resc} \\
&[-\Delta + \lambda (\phi^2- \xi^2)
+ \frac{g^2}{2}|w|^2 + \frac12 \kappa|z|^2 ]\phi = 0, \label{WS-eq4-resc}
\end{align}
where $\xi := r\vphi_0$
(with $r$ given in \eqref{r}), 
$\nu :=  g(a\sin\theta + z\cos\theta)$ and, recall, $\curl_{q} w = \n_1 w_2 - \n_2 w_1$, $\n_i  :=   \partial_i - iq_i $, $\p_i \equiv \p_{x^i}$ (for a $\mathfrak{u}(1)-$valued vector-field $iq$) and, recall, $\overline{w} \times w:=\overline{w}_1 w_2- \overline{w}_2 w_1$. We define the rescaled energy by
\begin{equation} \label{WS-energy-resc-def}
\E_{\Omega'}(w, a, z, \phi;r)  :=   r^2  E_{\Omega} (W, A, Z, \vphi).
\end{equation}
with  $(W, A, Z,\vphi)$ related to  $(w,a, z,\phi)$ by \eqref{rescaling} and $E_{\Omega} (W, A, Z, \vphi)$ given in  \eqref{WS-energy'}. Explicitly, we have
\begin{align} \label{WS-energy-resc}
\E_{\Omega'}(w, a, z, \phi; r) &= \int_{\Omega'}  
\big(|\curl_{\nu} w|^2  +\frac{1}{2} |\curl a|^2 +  \frac{1}{2}  |\curl z|^2 \notag \\ 
&+\frac{1}{2} g^2 \phi^2|w|^2+\frac12 \kappa \phi^2|z|^2 \notag +\frac{g^2}{2} |\overline w\times w|^2 \notag  \\ 
&+ i (\curl \nu)\overline w\times w +  |\n \phi|^2  + \frac{1}{2} \lam (\phi^2-\xi^2)^2\big).
\end{align}

We note that after rescaling,  the average magnetic flux per fundamental domain  becomes $n/e$ and the vacuum solution \eqref{vac},
\begin{align}\label{gs-resc}
m^{n, r} :=  (0,\frac 1e a^n, 0, \xi),
\end{align}
where $a^n(x)\equiv  A^n(x) =\frac{n}{2} J x$, 
. Furthermore, \eqref{gauge-per'} and Proposition \ref{prop:gauge-fix} imply that $(w, a, z, \phi)$ satisfy
\begin{align}
&w(x+s)=e^{i(\frac{n}{2} s\times x + c_s)} w(x) \text{ for all } s\in\cL', \label{resc-gauge-1} \\
&a(x+s)=a(x)+\frac{n}{2e} J s \text{ for all } s\in\cL', \label{resc-gauge-2} \\
\label{resc-gauge-3}	&\div a=0,\\
\label{resc-gauge-4}&z(x+s)=z(x), \quad \phi(x+s)=\phi(x)	 \text{ for all } s\in\cL',
\end{align}
where $c_s$ satisfies the condition $c_{s+t}-c_s-c_t-\frac n2 s\times t\in 2\pi\mathbb{Z}$, for all $s,t\in\cL'$. 

 Finally, the Sobolev spaces here, denoted again by  $\mathcal{H}_{\mathcal{L}'}^{s}$, can be obtained by rescaling  the Sobolev spaces defined above or defined directly, again as above, see \eqref{Sob-norm} and the text around it. Similarly to \eqref{en-finite}, by a  Sobolev embedding theorem, the rescaled energy is finite,
 \begin{align}\label{en'-finite} E_{\Omega'} (w, a, z, \phi; r)<\infty\ \text{ on }\ \mathcal{H}_{\mathcal{L}'}^1  \end{align} 
  and is independent of a choice of $\Omega'$.


\section{The linearized problem}\label{sec:linear} 
In this section we prove Theorem \ref{thm:normal-instab}, describing the stability properties 
 of the vacuum \eqref{vac}. Equivalently, we will investigate the energetic stability of the rescaled vacuum  solution \eqref{gs-resc} of  the rescaled WS equations \eqref{WS-eq1-resc} - \eqref{WS-eq4-resc}.

\DETAILS{
	In what follows, we choose the gauge for the magnetic potential so that $A^b(x):= \frac b2 (-x^2 dx^1 - x^1 dx^2)$. Substituting $A = A^b$ into \eqref{flux-quant}, we see that
	\begin{equation} \label{flux'}
	b = \frac{2\pi e}{|\Omega|}n.
	\end{equation}
	Then the vacuum solution \eqref{vac} satisfies \eqref{gauge-per'} for $\g_s(x)$ given by
	\begin{equation}
	\g_s(x) = \frac b2 s\times x + c_s
	\end{equation}
	with $c(\g_s)=n$, where the map $s\in\cL \to c_s\in\R$ satisfies the condition
	\begin{equation}
	c_{s+t} - c_s - c_t - \frac b2 s\times t \in 2\pi\Z,\ \quad \forall s, t\in \cL.
	\end{equation}
}

  Let $ m:= (w, a, z,\phi)$ and denote by $G(b, m) \equiv  G(m)$ the map $G: \mathcal{H}_{\mathcal{L}'}^2 \to \mathbb{C}^7$ given by the left-hand side of \eqref{WS-eq1-resc} - \eqref{WS-eq4-resc}, written explicitly as
\begin{align}\label{G-expl}
& G(b, m) \equiv   G(m)  = 
 (G_1(m), \dots, G_4(m)),\\
& G_1(m) :=  [\curl_{\nu}^* \curl_{\nu} + \frac{g^2}{2}\phi^2 - i (\curl \nu)J + g^2(\overline{w} \times w)J]w, \label{G1-expl} \\
& G_2(m) :=  \curl^* \curl a + 2e\Im[(\curl_{\nu}w)J\overline{w} - \curl^*(\overline{w}_1 w_2)], \label{G2-expl} \\
& G_3(m) :=  [\curl^* \curl + \kappa\phi^2] z + 
2g\cos\theta \Im[(\curl_{\nu}w)J\overline{w}  
 - \curl^*(\overline{w}_1 w_2)], \label{G3-expl}  \\  
& G_4(m) :=   [-\Delta  + \lambda (\phi^2-\xi^2)  + \frac{g^2}{2}|w|^2 + \frac{1}{2}\kappa |z|^2 ]\phi, \label{G4-expl}
\end{align}
where, recall,  $J$ is the symplectic matrix given in \eqref{J}, 
 $\xi := r\vphi_0$ (with $r$ given in \eqref{r}), $\nu := g(a\sin\theta + z\cos\theta)$, $\Delta$ is the standard Laplacian and the parameter $b$ enters through periodicity conditions \eqref{resc-gauge-1} - \eqref{resc-gauge-4}. Now, the WS system can be written as 
 \begin{align}\label{m-eq}G(m)=0.\end{align}
 
 
Recall the definition of stability given above Eq. \eqref{En-per}. To apply it to the rescaled WS equations \eqref{WS-eq1-resc} - \eqref{WS-eq4-resc}, we observe that the map $G$  is the $L^2$-gradient, $\grad_{L^2}\E_{\Omega'}$,  of the energy $\E_{\Omega'}$, see \eqref{WS-energy-resc}, considered as a functional of $u=(w,  a, z, \phi)$.  
Namely, $\lan G(m), \xi\ran_{L^2}=\del\E_{\Omega'}(m)  \xi$, where $\del\E_{\Omega'}(m)$ is the G\^ateau  derivative
\begin{align}\label{Gateau}
	\del\E_{\Omega'}(u)\xi \equiv \frac{d}{d\tau} \E_{\Omega'}(u + \tau\xi)|_{\tau=0},
\end{align}
of $E_{\Omega'}$ at $m$, defined on the space of variations $\cY$ 
 tangent to the space  
of $L^2_{\rm loc}$ functions of the form $(w,  a, z, \phi)$ 
  satisfying the gauge - periodicity conditions \eqref{resc-gauge-1} - \eqref{resc-gauge-4}:  
\begin{align}\label{Ydef}
\cY:=  L_n^2 \times L^2_0 \times L^2_0 \times L^2.
\end{align}
Here $L_n^2, L^2_0$ and $L^2$ are given by
\begin{align}
\label{L2n-space}& L_n^2  :=   \{w\in L^2_{loc}(\R^2,\C^2) : w(x+s)=e^{i(\frac{n}{2} s\times x + c_s)} w(x)\ \forall s\in\cL'\}, \\
\label{L20-space}& L_0^2  :=   \{\alpha\in L^2_{loc}(\R^2,\R^2) : \alpha(x+s)=\alpha(x)\ \forall s\in\cL',\; \div\alpha = 0\}, \\
\label{L2-space}& L^2  :=   \{\psi\in L^2_{loc}(\R^2,\R) : \psi(x+s)=\psi(x)\ \forall s\in\cL'\}
\end{align}
(see \eqref{resc-gauge-1} - \eqref{resc-gauge-3}).

 Since $G(m)=\grad_{L^2}\E_{\Omega'}(m)$, the $L^2$-Hessian for $\E_{\Omega'}$ and $m$ is the  formally symmetric operator 
\[\E_{\Omega'}''(m):=\del\grad_{L^2}\E_{\Omega'}(m)=\del G(m),\] 
 
Denote the $L^2$-Hessian at the vacuum solution $m^{n, r}$ (see \eqref{gs-resc}) 
  by 
\[L_{n, \mu}:= \del G(m^{n, r}).\] 
 As seen from  its explicit form given below, the operator $L_{n, \mu}$, acting on the space $\cY$,  is self-adjoint and therefore its spectrum is real. 

Thus, applied to the rescaled WS equations \eqref{WS-eq1-resc} - \eqref{WS-eq4-resc}, the definition of stability 
can be rephrased as: 

the vacuum solution $m^{n, r}$ 
 is {\it energetically stable} (respectively, {\it unstable}) if and only if $\inf\spec(L_{n, \mu})\ge 0$ (respectively, $\inf\spec(L_{n, \mu})< 0$).

We consider the operator $L_{n, \mu}$ on the space $\cY$, with the domain 
\begin{equation}\label{space}
\cX :=  \cH_n^2 \times \cH^2_0 \times \cH^2_0 \times \cH^2,
\end{equation}
where $\cH_n^s$, $\cH_0^s$ and $\cH^s$ are the respective Sobolev spaces for the $L^2$-spaces \eqref{L2n-space}-\eqref{L2-space}, 
 with inner products given (for $s\in\Z_{\geq 0}$) by
\begin{align}
\label{Hns-ip}&\langle w, w' \rangle_{\cH_n^s}  :=  \frac{1}{|\Omega'|} \sum\limits_{i=1}^2\sum\limits_{|\gamma|\leq s}\int_{\Omega'} \overline{(\nabla_{a^n})^{\gamma} w_i} (\nabla_{a^n})^{\gamma}  w'_i, \\
&\langle a,a' \rangle_{\cH^s_0}  :=   \frac{1}{|\Omega'|} \sum\limits_{i=1}^2\sum\limits_{|\gamma|\leq s}\int_{\Omega'} \partial^{\gamma} a_i \partial^{\gamma} a'_i, \\ 
&\langle \psi,\psi' \rangle_{\cH^s}  :=  \frac{1}{|\Omega'|} \sum\limits_{|\gamma|\leq s}\int_{\Omega'} \partial^{\gamma} \psi \partial^{\gamma} \psi',
\end{align}
 where $w^\#=(w_1^\#, w_2^\#), a^\#=(a_1^\#, a_2^\#)$, 
  $\Omega'$ is an arbitrary fundamental domain of the lattice $\cL'$ and $\gamma$ is a multi-index. The $\cL'$-equivariance of the above functions implies that these inner products do not depend on the choice of fundamental domain $\Omega'$.


We compute  the linear operator  $L_{n, \mu}$ explicitly. In what follows we use the notation $\oplus_{j} A_j$ for diagonal operator-matrices with the operators $A_j$ on the diagonal. 
Passing 
  from  the parameter $\xi= r\varphi_0$, or $r$, to the parameter  $ \mu :=  g^2 \xi^2/2$ and  using that $\nu\big|_{a=a^n/e, z =0}=\frac 1e a^n g\sin\theta=  a^n $, we find 
 \begin{align} \label{Ln}& L_{n, \mu} =  \oplus_{j=1}^4 H_j,\\ 
 \label{Hw}& H_1(\mu) :=   
\curl_{ a^n}^* \curl_{ a^n} + \mu 
 - n iJ, \\ 
\label{Ha} & H_2(\mu) :=  \curl^* \curl,   \\
\label{Hz} & H_3(\mu) :=   \curl^* \curl + \frac{\mu}{\cos^2\theta},\\ 
\label{Hvphi}& H_4(\mu) :=   -\Delta + \frac{4\lambda \mu}{g^2}, 
\end{align}
where, recall, 
 $\curl_{q} w = (\n_q)_1 w_2 - (\n_q)_2 w_1$, $(\n_q)_i :=  \partial_i - iq_i$,  $\p_i \equiv \p_{x^i}$.   (Note that the matrix  $ iJ$ 
  is self-adjoint.)

The gauge invariance of Eq. \eqref{m-eq} and the partial symmetry breaking  of vacuum solution \eqref{gs-resc}  
 imply that $L_{n,\mu=n}$ has the gauge zero mode: 
\begin{align}\label{WS-lin-op}
L_{n,\mu=n} \G_f= 0,\ \G_f :=  (0,  \n f, 0, 0).
\end{align}


For a null vector $\G_f$ defined in \eqref{WS-lin-op} to be in $\cX$, $f$ must satisfy $\div(\n f) = -\Delta f = 0$. This implies that $f$ is a linear function, $f(x) = c\cdot x + d$ for some $c\in\R^2$ and $d\in\R$, and so 
\begin{align}\label{Gf0}
 \G_f\in \cX \Longrightarrow  \G_f= (0, c,0,0).
\end{align}

 In this section we shall prove the following result implying Theorem \ref{thm:normal-instab}:
 \begin{theorem}\label{thm:WS-lin-op-lowestEV} The operator $L_{n, \mu}$ 
 on the space $\cX$ has purely discrete spectrum. For $\mu \neq n$, $L_{\mu,n}$ has the multiplicity $2$ eigenvalue $0$ with the eigenfuctions $(0, e_i,0,0)$, $i=1, 2, e_1=(1, 0), e_2=(0, 1)$ (see \eqref{Gf0}).
 
 Furthermore, the smallest non-zero eigenvalue given by $\mu - n$, 
 having multiplicity $n$. 
 For $\mu=n$, the eigenvalue $0$ has the multiplicity $n+2$.\end{theorem}


Theorem \ref{thm:WS-lin-op-lowestEV}  follows from Propositions \ref{non-neg-prop} and \ref{H_w-prop} given below. \qquad  \qquad  \quad $\Box$

\DETAILS{
Define
\begin{align}
& h_w :=   \curl_{a^n}^* \curl_{a^n}  - 2n iJ, \label{h-op-w} \\ 
&  h_a :=  \curl^* \curl, \label{h-op-a} \\
&  h_z :=  \curl^* \curl+\frac{g^2}{2\cos^2\theta}\xi^2, \label{h-op-z} \\ 
&  h_{\phi} :=  -\Delta. \label{h-op-phi}
\end{align}
}

\begin{proposition} \label{non-neg-prop}
	The operators $H_2(\mu)$, $H_3(\mu)$ and $H_4(\mu)$ 
	have purely discrete spectra. Furthermore, $H_3(\mu)$ and $H_4(\mu)$ are strictly positive and $H_2(\mu)$ is non-negative and has the null space $\{(0, c, 0, 0): c\in \R^2\}$ of dimension $2$. 
\end{proposition}

\begin{proof}
	The strict positivity of $H_3(\mu)$ and $H_4(\mu)$ and the non-negativity of $H_2(\mu)$ are obvious.
	The discreteness of the spectra and the form of the null space of $H_2(\mu)$ follow from the discreteness of the spectrum of the Laplacian on compact domains and the identity $\curl^*\curl v = -\Delta v$ when $\div(v)=0$. To compute the null space of $H_2(\mu)$,  we observe that the solutions of the equations $\Delta v=0$ and $\div(v)=0$ 
	 are constant vectors in $\R^2$.
\end{proof}


 Let $\n_{q} := \n - iq = ((\nabla_{ q})_1, (\nabla_{ q})_2)$, $(\nabla_{ q})_j :=  \partial_j - iq_j$, and $\Delta_q:=\n_{q}^2=-\n_{q}^*\n_{q}$. We also introduce the complexified covariant derivative 
	$\bar\p_{ q}  :=  (\nabla_{ q})_1 + i(\nabla_{ q})_2.$ 
 We have
\begin{proposition} \label{H_w-prop}
	(i) $H_1(\mu)$ is a self-adjoint operator on $\cH^2_n$ and its spectrum is given by
	\begin{equation} \label{spec-h_w}
	\sigma(H_1(\mu)) = \{(m-1)n + \mu: m\in\Z_{\geq 0}\} \cup \{\mu\},
	\end{equation}
	where $n := eb|\cL |/2\pi$.
	
	(ii) The eigenspace of the eigenvalue $-n+ \mu$ is $n$-dimensional and is spanned by functions of the form\footnote{$\beta$ can be expressed in terms of the Jacobi theta function, see Proposition \ref{prop:Landau-ham-spec} and Appendix \ref{saol}}
	\begin{equation} \label{chi-def} 
	\chi = (\beta, i\beta),\qquad
	\curl_{ a^n} \chi  =i\bar\p_{ a^n} \beta= 0, 
	\end{equation}
		and the eigenspace of the eigenvalue $\mu$ is of the form 
	\begin{equation} \label{H1-grad}
	\Null{(H_1(\mu)-\mu)}=\{\nabla_{ a^n}f: f\in\cH^3_n\}.
	\end{equation}
\end{proposition} 
In the proof of this proposition, we use the following standard result whose proof, for reader's convenience, is given in Appendix \ref{saol}:
\begin{proposition}\label{prop:Landau-ham-spec}  The operator $-\Delta_{ a^n}$ is self-adjoint on its natural domain and its spectrum is given by
	\begin{equation}\label{spec-Landau-app}	\sigma(-\Delta_{ a^n}) = \{\, (m + 1) n : m\in\Z_{\geq 0} \,\},
	\end{equation}
	with each eigenvalue is of the multiplicity $n$. Moreover,  
	\begin{equation} \label{nullL'}
	\Null (-\Delta_{ a^n} - n) = \Null \bar\p_{ a^n}.
	\end{equation}	
In more detail, with $z= (x^1+i x^2)/ \sqrt{\frac{2\pi}{\im\tau} }$ and $\tau$ coming from $\cL'=\Z+\tau\Z$, we have	
	\begin{align} \label{Vn-space} \Null (-\Delta_{ a^n} - n) =e^{\frac{in}{2}x^2(x^1 + ix^2) } V_n,\end{align} 
	where $V_n$ is spanned by functions of the form 
	\begin{align} 
	\label{theta-repr} &\theta (z, \tau) :=  \sum_{m=-\infty}^{\infty} c_m e^{i2\pi m z},\  c_{m + n} = e^{-in\pi z} e^{i2m\pi\tau} c_m.  \end{align}
 \end{proposition}
\begin{remark} Functions of the form \eqref{theta-repr} are determined entirely by the values of $c_0,\ldots,c_{n-1}$ and therefore form an $n$-dimensional vector space..
\end{remark}

\begin{proof}[Proof of Proposition \ref{H_w-prop}]
	First, we will show that $\cH^2_n = \cY \oplus \cZ$ (the Hodge decomposition), where
	\begin{align}
	\cY  :=  & \{w\in\cH^2_n: \div_{ a^n}w=0\}, \\
	\cZ  :=  & \{w\in\cH^2_n: w=\nabla_{ a^n}f \text{ for some } f\in\cH^3_n\},
	\end{align}
	with $\div_{ a^n}w  :=   
	(\nabla_{ a^n})_1 w_1 + (\nabla_{ a^n})_2 w_2=-\nabla_{ a^n}^*$. We write any $w\in\cH^2_n$ as $w = w_0 + \nabla_{ a^n}f$, where $f$ solves the equation $\Delta_{ a^n} f = \div_{ a^n} w$ and  $w_0$ is defined by this relation.  By Proposition  \ref{prop:Landau-ham-spec}, $0$ is not in the spectrum of $\Delta_{ a^n}$ and therefore the equation $\Delta_{ a^n} f = \div_{ a^n} w$ has the unique solution $f\in\cH^3_n$. 
	Then, since 
		$\Delta_{ a^n}:=\div_{ a^n} \nabla_{ a^n}$, 
		we have $\div_{ a^n}w_0=0$. 
This proves $\cH^2_n = \cY \oplus \cZ$.

 Now, recall that the operator $H_1(\mu)$ acts on complex vectors $w=(w_1, w_2)$. 
The definition $H_1(\mu)  :=  \curl_{ a^n}^* \curl_{ a^n} - n iJ  + \mu$ and  the relations $\curl_{ a^n}^* = -J\nabla_{ a^n}$ and
  \[\curl_{ a^n}\nabla_{ a^n} = [(\nabla_{a^n})_1, (\nabla_{a^n})_2] = -in\] yield that $(H_1(\mu)-\mu)\nabla_{ a^n}f = 0$, 
	which proves that the $\mu$-eigenspace of $H_1(\mu)$ is of the form \eqref{H1-grad} giving  the second part of (ii).  
	

	By the above the subspace $\cY$ is invariant under  $H_1(\mu)$. To compute the spectrum of the operator $H_1(\mu)$ on the subspace $\cY$, we 
	\DETAILS{write it as 
	\begin{equation} \label{H1-deco}
	H_1(\mu)  :=   h_1 + \mu,\ h_1  :=  \curl_{ a^n}^* \curl_{ a^n} - n iJ. 
	\end{equation}
Now, using}
use the definitions of $\curl_{ a^n}$ and $\curl_{ a^n}^*$ and recall the relation $ [(\nabla_{a^n})_1, (\nabla_{a^n})_2] = -in$ to compute  
\[\curl_{ a^n}^* \curl_{ a^n}=-\Delta_{ a^n} - n iJ+ \nabla_{a^n}\div_{ a^n}.\] 
By above, we have $H_1(\mu) w_0 =(-\Delta_{ a^n} - 2n iJ -\mu)w_0$, for any $w_0\in\cY$.
(We check using  $\div_{a^n} (-\Delta_{ a^n} - 2n iJ) w_0 = (-\Delta_{ a^n}) \div_{ a^n} w_0 =0$,	  that  
$H_1(\mu)$ sends $\cY$ to $\cY$ and hence, $\cY$ is invariant under $H_1(\mu)$.) 
Thus, we
 conclude that 
	\begin{align}
	H_1(\mu) 
	&(w_0 \oplus 0) = (h_{1} -\mu) w_0 \oplus 0, \\
	&h_{1}  :=   -\Delta_{ a^n} - 2n iJ.
	\end{align}
	
	Identifying one-forms with vector-fields, we compute
	\begin{equation}
	U^*(iJ)U = \begin{pmatrix}
	-1 & 0 \\
	0 & 1
	\end{pmatrix}, \;\;
	U:= \frac{1}{\sqrt{2}}\begin{pmatrix}
	1 & 1 \\
	-i & i
	\end{pmatrix},
	\end{equation}
	which gives
	\begin{equation} \label{h_w-diag}
	U^* h_{1} U = \left(\begin{array}{cc}
	-\Delta_{ a^n} + 2n & 0 \\ 0 & -\Delta_{ a^n} - 2n \end{array} \right).
	\end{equation}
	By Proposition \ref{prop:Landau-ham-spec}, we know that
	\begin{equation}\label{spec-Landau}	
	\sigma(-\Delta_{ a^n}) = \{\, (m + 1)n : m\in\Z_{\geq 
		0}\}
	\end{equation}
	and so the spectrum of $H_1(\mu)$ on $\cY$ is given by the first set on the r.h.s. of \eqref{spec-h_w}. Hence, by  $\cH^2_n = \cY \oplus \cZ$, \eqref{spec-h_w} follows, giving  (i). 
	
	Furthermore, by \eqref{h_w-diag} and \eqref{spec-Landau}, any eigenvector $\chi$ of  $h_{1}$ corresponding to the eigenvalue $-n$ must be of the form
	\begin{equation}
	\chi=U (0, \beta)  
	= \frac{1}{\sqrt{2}} (\beta,i\beta) , 
	\end{equation}
	where $\beta$ satisfies 
	\begin{equation} \label{eta-eqn}
	-\Delta_{ a^n}\beta = n\om.
	\end{equation}
This relation, together with the equation $\Null (-\Delta_{ a^n} - n) = \Null \bar\p_{ a^n}$ (see  \eqref{nullL'}), 
implies $\overline{\partial}_{ a^n}\beta = 0$. 	Since $\curl_{ a^n}\chi = i\overline{\partial}_{ a^n}\beta$,  this gives 	
	\begin{equation}\curl_{ a^n}\chi = i\overline{\partial}_{ a^n}\om = 0.\end{equation} 
	Furthermore, by Proposition \ref{prop:Landau-ham-spec}, the space of such functions is $n$-dimensional. Thus (after rescaling $\om$ by a factor of $\sqrt{2}$) $\chi$ is of the form \eqref{chi-def}. This 
	gives also the first part of (ii) completing the proof of the proposition. 
\end{proof}

We see that the operator $H_1(\mu)$ is non-negative for the magnetic fields satisfying $b < b_* :=  g^2 \varphi_0^2/2e = M_W^2/e$ and acquires a negative eigenvalue $ \mu-n = (b_*/b  - 1)\,n$ of multiplicity $n$ as the magnetic field increases to $b > b_*$. Theorem \ref{thm:normal-instab} follows by undoing the rescaling \eqref{rescaling} - \eqref{r}.

\section{Setup of the bifurcation problem}\label{sbp}

We substitute $a=\frac1e a^n+\alpha$ (with $\div(\alpha)=0$), $\phi=\xi+\psi$, $\nu= a^n+\tilde\nu$ and 
$\xi = \sqrt{2\mu}/g$
  into \eqref{WS-eq1-resc} - \eqref{WS-eq4-resc} 
  and  relabel the unknowns $w, \alpha,z,\psi$ as $u_1, u_2, u_3, u_4$ to obtain the system
\begin{align}\label{mu-u-syst}
 H_iu_i &= -J_i(\mu, u),\ i=1, \dots, 4, 
\end{align}
where $u= (u_1, u_2, u_3, u_4)\equiv  (w, \alpha,z,\psi)$, the operators $H_i$ on the left-hand side are defined in \eqref{Hw} - \eqref{Hvphi}, 
 and
\begin{align}
\label{J1-def} J_1(\mu,u) & :=   M w + \frac{g^2}{2}\psi^2 w + g\sqrt{2\mu}\psi w -i(\curl\tilde\nu)Jw  
 + g^2(\overline{w} \times w)Jw,  \\
  J_2(\mu, u) & :=   2e\Im[
(\curl_{\nu} w) J\overline{w} - \curl^*(\overline{w}_1 w_2)], \label{J2-def}\\
 J_3(\mu, u) & :=   2g\cos\theta\Im[
(\curl_{\nu} w)  J\overline{w} - \curl^*(\overline{w}_1 w_2)] 
 + 
\kappa \frac{2\sqrt{2\mu}}{g}      \psi z + \kappa\psi^2 z, \label{Jz-def} \\
  J_4(\mu, u) & :=   3\lambda \frac{\sqrt{2\mu}}{g}\psi^2 + \lambda\psi^3 + \frac{g^2}{2}|w|^2(\frac{\sqrt{2\mu}}{g}+\psi) 
   + 
       \frac12\kappa |z|^2(\frac{\sqrt{2\mu}}{g}+\psi), \label{Jphi-def}
 \end{align}
 with $ \tilde\nu  :=   g(\alpha\sin\theta+z\cos\theta)$, $\xi\times \eta := \xi_1 \eta_2 - \xi_2 \eta_1$, recall,  $\curl_{q} w = \n_1 w_2 - \n_2 w_1$, $\n_i :=  \partial_i - iq_i$ and, recalling that $w:\R^2\ra\C^2$, 
 \begin{align} \label{M}
 M  :=\curl_{\nu}^* \curl_{\nu}-\curl_{ a^n}^* \curl_{ a^n}=   &\begin{pmatrix}
  M_{22} & -M_{21} \\
  -M_{12} & M_{11}
 \end{pmatrix}, 
\end{align}
with $M_{ij}  :=   i \tilde\nu_i (\nabla_{a^n})_j + i \tilde\nu_j (\nabla_{ a^n})_i + i \partial_i\tilde\nu_j + \tilde\nu_i\nu_j$. 
\DETAILS{Thus we may rewrite \eqref{mu-u-syst} as 
\begin{align}\label{F-def}
 F(\mu,u)=0,\ \text{ where }\  F(\mu,u):=  L_{n, \mu}u + J(\mu,u),
\end{align}
with $L_{n, \mu}=\oplus H_i$ given in \eqref{Ln} and $J(\mu,u)$ given by 
\begin{align}
J(\mu,u)&\equiv  (J_1(\mu,u), \dots, J_4(\mu,u)). 
\label{Jmu}
\end{align}

 We consider $F(\mu,u)$ as a map from  the space $\R_{>0} \times\cX$, where  $\cX :=  \cH_n^2 \oplus \cH^2_0 \oplus \cH^2_0 \oplus \cH^2$, to the space 
 $\cY:=  L_n^2 \oplus L_0^2 \oplus L_0^2 \oplus L^2$,  and let $F = (F_1, \dots, F_4)$, where
 \begin{equation} \label{F-comps}
  F_i(\mu, u) = H_i u + J_i (\mu, u), \quad i=1, \dots, 4. 
 \end{equation}}

Note that system \eqref{mu-u-syst} can be also written as $G(m^{n, r}+u)|_{\xi = \sqrt{2\mu}/g}=0$, where $G$ is defined in \eqref{G-expl} and $m^{n, r} :=  (0,\frac1e a^n,0,\xi)$.

Applying $\divv$ to the second equation in \eqref{mu-u-syst}, we find that a solution $(\mu, u)$ should satisfy $\divv J_2 (\mu, u)=0$. To prove that a solution $(\mu, u)$ satisfies this constraint, we consider the following auxiliary problem %
\begin{align}\label{F-eq'}
 F(\mu,u)=0,\ \text{ where }\  F(\mu,u):=  L_{n, \mu}u + P' J(\mu,u),
\end{align}
where  $P'=\one\otimes P_0\otimes \one\otimes \one$, with $P_0$ the orthogonal projection onto the divergence-free  vector fields ($P_0=\frac{1}{-\Delta}\curl^*\curl$), and, recall, $L_{n, \mu}=\oplus H_i$ and $J(\mu,u)$ given in \eqref{Ln} and 
\begin{align}
J(\mu,u)& :=   (J_1(\mu,u), \dots, J_4(\mu,u)). 
\label{Jmu}
\end{align}

We consider $F(\mu,u)$ as a map from  the space $\R_{>0} \times\cX$, where  $\cX :=  \cH_n^2 \oplus \cH^2_0 \oplus \cH^2_0 \oplus \cH^2$, to the space 
 $\cY:=  L_n^2 \oplus L_0^2 \oplus L_0^2 \oplus L^2$,  and let $F = (F_1, \dots, F_4)$, where 
 \begin{equation} \label{Fi-def} F_i(\mu, u) = H_i u + \del_{i, 2}P_0J_i (\mu, u), \quad i=1, \dots, 4. \end{equation}
In what follows, we denote 
 the partial (real) G\^ateaux derivatives with respect to $\#$ by $\del_\#$.
\begin{proposition}\label{prop:reconstr}   Assume $(\mu,u)$ is a solution of the system \eqref{F-eq'} satisfying  the gauge - periodicity conditions \eqref{resc-gauge-1} - \eqref{resc-gauge-4}. 
Then $\divv J (\mu,u)=0$  and therefore   $(\mu,u)$ solves the original system \eqref{mu-u-syst}.    \end{proposition}
\begin{proof}  
  We follow \cite{TS1}. Assume 
  $\chi\in H^1_{\rm loc}$ and is $\cL-$periodic (we say, $\chi \in H^1_{\rm per} $). The gauge invariance implies that \begin{equation}E_{\Omega'}(e^{i s\chi}w, a+s\n\chi, z, \phi) =E_{\Omega'}(w, a, z, \phi),\end{equation} 
 where  $E_{\Omega'}(w, a, z, \phi)$ is given in \eqref{WS-energy-resc}. Differentiating this equation with respect to $s$ at $s=0$ gives $\del_w E_{\Omega'}(w, a, z, \phi)(i\chi w)+\del_a E_{\Omega'}(w, a, z, \phi)(\n\chi) =0$. Now, we use the fact that the partial G\^ateaux derivative with respect to $w$ vanishes, $\del_w E_{\Omega'}(w, a, z, \phi)=0$, and that $\curl\n \chi=0$, and integrate by parts, to obtain
\DETAILS{\begin{align}  
 \Re\lan [\curl_{\nu}^* \curl_{\nu} + \frac{g^2}{2}\phi^2 - i (\curl \nu)J&+ g^2(\overline{w} \times w)J]w,  i \chi\psi\ran \notag \\ 
\label{E-deriv} &+
 \lan J (\mu,u), \n \chi\ran=0.\end{align} 
 (Due to conditions \eqref{resc-gauge-1} - \eqref{resc-gauge-4} and the $\cL-$periodicity of $\chi$,  there are no boundary terms.) 
This, together with equation \eqref{WS-eq1-resc}, implies}
 \begin{align}\label{J-relat} \lan J (\mu,u), \n\chi\ran=0.\end{align} 
 (Due to conditions \eqref{resc-gauge-1} - \eqref{resc-gauge-4} and the $\cL-$periodicity of $\chi$,  there are no boundary terms.)  Since the last equation holds for any  $\chi \in H^1_{\rm per} $, we conclude that $\divv J (\mu,u)=0$.
\end{proof} 
In Sections \ref{sec:rf-dp} - \ref{sec:bif-res} 
 we solve equation \eqref{F-eq'}, subject to conditions \eqref{resc-gauge-1} - \eqref{resc-gauge-4}. 


In conclusion of this section, we investigate properties of the map $F(\mu,u)$. For $f=(f_1, f_2, f_3, f_4)$ and $\delta\in\R$, define the global transformation 
\begin{equation} \label{T-delta-def}
T_{\delta}f=(e^{i\delta}f_1, f_2, f_3, f_4).
\end{equation}

\begin{proposition}\label{F-prop}
 The map $F(\mu,u)$ defined in \eqref{F-eq'} has the following properties:
 \begin{enumerate}[(i)]
  \item $F:\R_{>0}\times\cX \to \cY$ is continuously G\^ateau differentiable of all orders;
  \item $F(\mu,0)=0$ for all $\mu\in\R_{>0}$;
  \item $\del_u F(\mu,0)=L_{n, \mu}$ for all $\mu\in\R_{>0}$;
  \item $F(\mu,T_{\delta}u)=T_{\delta} F(\mu,u)$ for all $\delta\in\R$;
  \item $\langle u,F(\mu,u)\rangle_{\cY} \in\R$ (respectively $\langle w,F_1(\mu,u)\rangle_{L_n^2}\in\R$) for all $u\in\cX$ (respectively $w\in\cH_n^2$).
 \end{enumerate}
\end{proposition}
\begin{proof}
 $(i)$ follows because $F$ is a polynomial in the components of $u$ and their first- and second-order (covariant) derivatives. $(ii)$, $(iii)$ and $(iv)$ follow from an easy calculation (in fact, $u$ and $L_{n, \mu}$ were defined so that $(ii)$ and $(iii)$ hold). For $(v)$, it suffices to show that $\langle w,F_1(\mu,u)\rangle_{L_n^2}\in\R$. To simplify notation we return to the coordinates $(w,a,z,\phi) = (w,\frac1e a^n+\alpha,z,\frac{\sqrt{2\mu}}{g} +\psi)$. 
 Then
 \begin{align}
  \langle w,F_1(\mu,u)\rangle_{L_n^2} &= \frac{1}{|\Omega'|}\int_{\Omega'} |\curl_{\nu}w|^2 + \frac{1}{|\Omega'|}\int_{\Omega'} \frac{g^2}{2}\phi^2|w^2|  \nonumber\\ 
  + &  \frac{1}{|\Omega'|}\int_{\Omega'} i (\curl \nu) (\overline{w}\times w) + \frac{1}{|\Omega'|}\int_{\Omega'} g^2 |\overline{w} \times w|^2.
 \end{align}
 The first, second and fourth terms are clearly real, while the third term is real because $\nu$ is real and $\overline{w}\times w$ is imaginary.
\end{proof}

\section{Reduction to a Finite-Dimensional Problem}\label{sec:rf-dp}

In this section we shall reduce solving equation \eqref{F-eq'}, i.e. $F(\mu, u)=0$, with  $u= (u_1, u_2, u_3, u_4)$ $ \equiv (w, \al, z, \psi)$ 
and $F: \R_{>0} \times\cX\ra \cY$ defined in \eqref{F-eq'} - \eqref{Jmu}), to a finite-dimensional problem. 
To this end, we use the Lyapunov-Schmidt reduction.

Recall that $L_{n,\mu}$ is defined in \eqref{Ln}.  Let $P$ be the orthogonal projection onto $\cK:= \Null(L_{n,\mu=n})$, which 
can be written explicitly as
\begin{align} \label{proj-def}
 &P = P_1 \oplus P_2 \oplus 0 \oplus 0, \\
 \label{P1-def} &P_1w  :=   -\frac{1}{2\pi i} \oint_{\g_n} (H_1(n)-z)^{-1}w\; dz, \\
 &P_2\alpha  :=   \langle \alpha \rangle,
\end{align}
 where $H_1(n)$ is defined in \eqref{Hw}, $\g_n$ is any simple closed curve in $\C$ containing the eigenvalue $0$ and no other eigenvalues of $H_1(n)$ (see Proposition \ref{H_w-prop}), and $\langle\alpha\rangle$ is the mean value of $\alpha$ in $\Omega'$, $\langle \alpha \rangle :=  \frac{1}{|\Omega'|} \int_{\Omega'} \alpha$. $P_1$ is a projection onto $\Null(H_1(n))$ (spanned by vectors of the form \eqref{chi-def}). Since $H_1(n)$ is self-adjoint, $P_1$ is an orthogonal projection (relative to the inner product of $L^2_n$). 
 By Theorem \ref{thm:WS-lin-op-lowestEV}, $\cK:= \Null(L_{n,\mu=n})$ is $(n+1)$-dimensional.

 Let $P^{\perp}=1-P$ be the projection onto the orthogonal complement of $\cK$. Applying $P$ and $P^{\perp}$ to the equation $F(\mu,u)=0$ (see \eqref{F-eq'}), we split it into two equations for two unknowns as
\begin{align}\label{projF1}
& P F(\mu, v+u') =0, \\ \label{projF2}
& P^{\perp}F(\mu,v+u') =0, 
\end{align}
where $v:= P u, \ u':= P^{\perp}u$.

Our next goal is to solve \eqref{projF2} for $u'$ in terms of $\mu$ and $v$.
For $n=1$, $\cK$ is two-dimensional and we write $v = (v_1,v_2,v_3,v_4)\equiv  (v_1,v_2, 0, 0)\in  \cK$. 
 Let $\cX^{\perp}:= P^{\perp}\cX=\cX\ominus\cK$ and $\cY^{\perp}:= P^{\perp}\cY=\cY\ominus\cK$, and let $\p_i \equiv  \p_{x_i}$.
\begin{proposition}\label{prop:sol-u'} 
	There is a neighbourhood $U \subset \R_{>0} \times \cK$ of $(n,0)$ such that for every $(\mu,v)\in U$, equation \eqref{projF2} for $u'$ has a unique solution $u' = u'(\mu,v)= (u'_1, u'_2, u'_3, u'_4)$. Furthermore, this 
	solution 
	 has the following properties:
	\begin{align}\label{u2-anal}
	&u': \R_{>0} \times \cK \to \cX^\perp \text{ is continuously 
	  differentiable of all orders;} \\
 \label{Ow2-eq}	&\|(\nabla_{a^n})_j^m u_1'\|_{\cH^2_n} = \cO(\|v\|_{\cX}^2),\\ 
\label{Op2-eq}	&\|\partial_j^m u_k'\|_{\cH^2_k} = \cO(\|v\|^2_{\cX}),\\ 
 \label{dv-u1'-est}&||\p_{v_i} (\nabla_{a^n})_j^m u'_1(\mu, v^{i})\|_{\cH^2_n} \lesssim \|v^{i}\|_{\cX},\\ 
 \label{dv-uk'-est}&\|\p_{v_i} \p_j^m u'_k(\mu, v^{i})\|_{\cH^2_0} \lesssim \|v^{i}\|_{\cX},\\ 
\label{dmu-u'-est}&\|\partial_{\mu} u'(\mu,v)||_{\cX} \lesssim \|v\|^2_{\cX};\end{align}
where $i=1,... 4,\; m=0,1,\; j=1,2, \; k=2,3,4$, $v = (v_1,v_2,v_3,v_4)$, $v^{i} \equiv 
v|_{v_i=0}$ 
and	$\cH^2_k=\cH^2_0$, $\cH^2_0$, $\cH^2$ for $k=2, 3, 4$. 
\end{proposition} \begin{proof}
	Define $F^{\perp}:\R_{>0}\times\cK\times \cX^{\perp} \to \cY^{\perp}$ by
	\begin{equation}
	F^{\perp}(\mu,v,u')  :=   P^{\perp} F(\mu,v+u').
	\end{equation}
	By Proposition \ref{F-prop} $(i)$ and $(ii)$, $F^{\perp}$ is continuously differentiable of all orders as a map between Banach spaces and $F^{\perp}(\mu,0,0)=0$ for all $\mu\in\R_{>0}$. Furthermore, 
\begin{align} \label{deluFperp}	\del_{u'} F^{\perp}(\mu,0,0) = P^{\perp} L_{n,\mu}P^{\perp}|_{\cX^\perp},\end{align} which is invertible for $\mu=n$ because $P^\perp$ is the projection onto the orthogonal complement of $\cK = \Null(L_{n,\mu=n})$. By the Implicit Function Theorem (see e.g. \cite{berger}), there exists a function $u'(\mu,v)$ with continuous derivatives of all orders such that for $(\mu,v)$ in a sufficiently small neighbourhood $U\subset\R_{>0}\times\cK$ of $(n,0)$, $(\mu,v,u')$ solves \eqref{projF2} if and only if $u' = u'(\mu,v)$. This proves the first statement and property \eqref{u2-anal}.

We define the operator 
\begin{align} \label{Ln-perp}L_{n, \mu}^{\perp}  :=   P^{\perp}L_{n, \mu}P^{\perp}|_{\cX^{\perp}}: \cX^{\perp}\ra \cY^{\perp}.\end{align}
Then by \eqref{F-eq'} and \eqref{deluFperp}, we can write equation \eqref{projF2} as $L_{n, \mu}^{\perp} u' = - P^{\perp} P' J(\mu,u)$.
By Theorem \ref{thm:WS-lin-op-lowestEV} and the relation 
$\cK:= \Null(L_{n,\mu=n})=\Null(L_{n,\mu}-\mu+n)$, for $\mu$ in a neighbourhood of $n$, the operator $L_{n, \mu}^{\perp}$ has a uniformly bounded inverse $(L_{n, \mu}^{\perp})^{-1}:\cY^\perp \to \cX^\perp$. Hence 
 equation $L_{n, \mu}^{\perp} u' = - P^{\perp} P' J(\mu,u)$, with $(\mu, v)\in U$ (replacing $U$ with a smaller neighbourhood if necessary), is equivalent to
	\begin{align} \label{u-perp-J}
	u' = -(L_{n, \mu}^{\perp})^{-1} P^{\perp} P' J(\mu,u);
	\end{align}
	hence
	\begin{align} \label{J-ineq}
	\|u'\|_{\cX} \lesssim \|J(\mu,u)\|_{\cY}, 
	\end{align}
    uniformly in $\mu$. Recall that $\cX = \cH^2_n\oplus \cH^2_0\oplus \cH^2_0\oplus \cH^2$ and $\cY = L^2_n\oplus L^2_0\oplus L^2_0\oplus L^2$. $J(\mu,u)$ is a polynomial in the components of $u$ and their first-order (covariant) derivatives consisting of terms of degree at least $2$, so the left-hand side of \eqref{J-ineq} can be bounded above by a sum of products of one $\cL^2$-norm and at least one $\cL^{\infty}$-norm of these terms. $\cH^1$ is trivially continously embedded in $\cL^2$, and by the Sobolev Embedding Theorem, $\cH^1$ is continuously embedded in $\cL^{\infty}$. Therefore,
	\begin{align} \label{J-bds}
	\|J(\mu,u)\|_{\cY} \lesssim \|u\|_{\cX}^2.
	\end{align}
	Recalling that $u=v+u'$, this proves \eqref{Ow2-eq} and \eqref{Op2-eq} when $m=0$. The other case is proven similarly.

 For $v = (v_1, \dots, v_4)$, we let $v_{\widehat i} \equiv  v |_{v_{i} = 0}$, $i=1, \dots, 4$. By the Taylor theorem for Banach spaces (see e.g. \cite{berger}), we have
    \begin{align}
   &  u'(\mu, v) = u'(\mu,v_{\widehat i}) + \del_{v_{ i}} u'(\mu,v_{\widehat i}) v_{i} + R_2(\mu,v_{\widehat i})(v_{i}), \\
    &  R_2(\mu,v_{\widehat i})(v_{i})  :=  \int_{0}^{1} (1-t) \del_{v_{ i}}^2 u'(\mu,v_{\widehat i} + tv_i)(v_{ i},v_{ i}) dt.
    \end{align}
    Let $(\mu,v) \in U$ with $\|v_{\widehat i}\|=\|v_{ i}\|=1$, and let $\epsilon>0$. Then
        \begin{align} \label{q-tay1}
     \|\del_{v_{ i}} u'(\mu,\epsilon v_{\widehat i})& \epsilon v_{ i} \|_{\cX} = \|u'(\mu, \epsilon v) - u'(\mu,\epsilon v_{\widehat i}) - R_2(\mu,\epsilon v_{\widehat i})(\epsilon v_{i})\|_{\cX} \nonumber \\
     &\leq \|u'(\mu, \epsilon v)\|_{\cX} + ||u'(\mu, \epsilon v_{\widehat i})\|_{\cX} \nonumber \\
      &+ \epsilon^2 \|v_{i}\|^2\sup\limits_{0\leq t\leq 1} (1-t) \|\del_{v_{i}}^2 u'(\mu,\epsilon v_{\widehat i} + t\epsilon v_{ i})\|^2_{\cX^* \otimes \cX^* \otimes \cX}  \nonumber \\
     &\lesssim \epsilon^2.
    \end{align}
    with the norm taken in the appropriate space for $v_{ i}$. Taking the supremum over all $v_{ i}$ with $||v_{ i}||=1$ gives
    \begin{align}
     \|\del_{v_{ i}} u'(\mu,\epsilon v_{\widehat i})\|_{\cX } \lesssim \epsilon, \ \|v_{\widehat i}\|_{\cX} = 1,
    \end{align}
    proving \eqref{dv-u1'-est} - \eqref{dv-uk'-est} for $m=0$. The other cases are proven in exactly the same way.
    
    Again by Taylor's Theorem,
	\begin{align}
	&\partial_\mu u'(\mu, v) = \partial_\mu u'(\mu,0) + \partial_\mu \del_{v} u'(\mu,0) v +\tilde R_2(\mu,0)(v), \\
	&\tilde R_2(\mu,0)(v) :=  \int_{0}^{1} (1-t) \partial_\mu \del_{v}^2 u'(\mu, tv)(v,v) dt.
	\end{align}
	By 
	Equations \eqref{Op2-eq} and \eqref{dv-u1'-est} - \eqref{dv-uk'-est} with $m=0$, we have $u'(\mu,0) = 0$ and $\del_{v} u'(\mu,0) = 0$, so
	\begin{align}
	 \|\partial_\mu u'(\mu, v)\|_{\cX} &= \|\tilde R_2(\mu,0)(v)\|_{\cX} \\
	 &\leq \|v\|_{\cX}^2\sup\limits_{0\leq t\leq 1} (1-t) \|\partial_{\mu}\del_{v}^2 u'(\mu,tv)\|^2_{\cX^* \otimes \cX^* \otimes \cX}  \\
	  &\lesssim \|v\|^2_{\cX}, 
	\end{align}
	proving \eqref{dmu-u'-est}.
\end{proof}

 We plug the solution $u' = u'(\mu,v)$ into equation \eqref{projF1} to get the \emph{bifurcation equation}
 \begin{equation}\label{bif-eq}
  \g(\mu,v):= P F(\mu,v+u'(\mu,v))=0.
 \end{equation}

\begin{corollary} \label{cor:red2v} 
 In a neighbourhood of $(n,0)$ in $\R_{>0}\times\cX$, the pair $(\mu,u)$ 
 solves 
 \eqref{F-eq'} if and only if $(\mu,v)$ 
 solves the finite-dimensional equation \eqref{bif-eq}. Moreover, a solution of \eqref{F-eq'} can be constructed from a solution $(\mu,v)$ of \eqref{bif-eq} by setting $u = v + u'(\mu,v)$, where $u'(\mu,v)$ is given by Proposition \ref{prop:sol-u'}.
\end{corollary}

Since $F:\R_{>0}\times\cX\to\cY$ and $u':\R_{>0}\times\cK\to\cY^\perp$ have been shown to be continuously differentiable of all orders, we conclude:
\begin{corollary} \label{g-anal}
	$\g:\R\times\cK\to\cK$ is continuously G\^ateau differentiable of all orders.
\end{corollary}

Furthermore, $\g(\mu,v)$ inherits the following symmetry of $F(\mu,u)$, which we will use to find a solution of \eqref{bif-eq}:

\begin{lemma} \label{g-prop}
 Let $T_{\delta}$ be given by \eqref{T-delta-def}. For every $\delta\in\R$ and $(\mu, v)$ in a neighbourhood of $(n,0)$, we have 
 \begin{align}\label{u2-sym}
  u'(\mu,T_{\delta}v) &= T_{\delta} u'(\mu,v), \\ \label{gamma-sym}
  \g(\mu,T_{\delta}v) &= T_{\delta}\g(\mu,v).
 \end{align}
\end{lemma}

\begin{proof}
 For equation \eqref{u2-sym}, we note that by Proposition \ref{F-prop} $(iv)$
 \begin{align}
  P^{\perp}F(\mu,T_{\delta}v + T_{\delta} u'(\mu,v)) &= P^{\perp}T_{\delta}F(\mu,v + u'(\mu,v))\\ \notag & = T_{\delta}P^{\perp}F(\mu,v + u'(\mu,v)) = 0.
 \end{align}
 (Here we used $P^{\perp}T_{\delta} = T_{\delta}P^{\perp}$, which follows because $T_{\delta} = e^{i\delta}\oplus 1\oplus 1\oplus 1$ and $P^{\perp} = 1-P$ where $P$ is defined in \eqref{proj-def}.) 
 Since $u' = u'(\mu,T_{\delta}v)$ is the unique 
 solution to $P^{\perp}F(\mu, T_{\delta}v + u')=0$ 
 for $(\mu,v)$ in a neighbourhood $U\subset \R\times\cK$ 
 of $(n,0)$, we conlcude that $u'(\mu,T_{\delta}v) = 
 T_{\delta} u'(\mu,v)$.
 
 For equation \eqref{gamma-sym}, we note that by \eqref{u2-sym} and Proposition \ref{F-prop} $(iv)$,
 \begin{align}\nonumber
  \g(\mu,T_{\delta}v) &= PF(\mu,T_{\delta}v + u'(\mu,T_{\delta}v)) = PF(\mu, T_{\delta}(v + u'(\mu,v))) \\ &= T_{\delta}PF(\mu, v + u'(\mu,v)) = T_{\delta}\g(\mu,v)
 \end{align}
 (where again we used $PT_{\delta}=T_{\delta}P$).
\end{proof}

\section{The bifurcation result when $n=1$}\label{sec:bif-res}

\begin{theorem}\label{bifurcation-result}
 Assume that $n=1$ and $|1-b_*/b|\ll 1$, $b_* :=   M_W^2/e$. 
 Then there exists $\epsilon>0$ and a branch $(\mu_s,u_s)  :=  (\mu_s, w_s, \al_s,z_s, \psi_s)$, 
 with $s\in [0,\sqrt{\epsilon})$, of non-trivial solutions of equation \eqref{mu-u-syst}, 
  unique modulo a gauge symmetry 
 in a sufficiently small neighbourhood of the rescaled vacuum solution \eqref{gs-resc} in $\R_{>0}\times\cX$, such that
\begin{align}\label{bif-soln}
 \begin{cases}
  w_s = s \chi + sg_1(s^2), \\
  \al_s = 
  g_2(s^2), \\
  z_s = g_3(s^2), \\
  \alpha\alpha_s = 
  g_4(s^2), \\
\mu_{s} = n + g_5(s^2),  
 \end{cases}
\end{align}
 where $\chi$ solves the eigenvalue problem $H_1(n) \chi=0$ (it is defined in \eqref{chi-def}, see Proposition \ref{H_w-prop}), $\mu:=  g^2 \xi^2/2= g^2 r^2\vphi_0^2/2$, 
 $g_1:[0,\epsilon)\to\cH_n^2$ and is orthogonal to $\Null(H_1(n))$, $g_2:[0,\epsilon)\to\cH^2_0$, $g_3:[0,\epsilon)\to\cH^2_0$, $g_4:[0,\epsilon)\to\cH^2$, $g_5:[0,\epsilon)\to\R_{>0}$, and $g_j$ for $j=1,\cdots, 5$ are functions, continuously differentiable of all orders in $s$, such that $g_j(0)=0$.
 
\end{theorem}

\DETAILS{
We shall need the following fact about continuously differentiable of all orders functions:

\begin{proposition}\label{Oprop}
	Let $X_1, ..., X_k$ and $Y$ be Banach spaces, and let $f: \R\times X_1 \times ... \times X_k \to Y$ be a function, continuously G\^ateau differentiable of all orders in $s$ and $x_i$ for $i\in\{1,...,k\}$, such that $f(s,x_1,...,x_k) = \cO(|s|^n)$ (i.e. $||f(s,x_1,...,x_k)|| \leq c(x_1,...,x_k) |s|^n$ with $s$ sufficiently close to $0$ for some $c(x_1,...,x_k) > 0$ not depending on $s$). Then $\partial_{x_i} f = \cO(|s|^n)$ for $i\in\{1,...,k\}$.
\end{proposition}

\begin{proof}
 This follows by differentiating each locally convergent power series of $f$ term-by-term within its radius of convergence.
\end{proof}
}

\begin{proof}[Proof of Theorem \ref{bifurcation-result}]
	For the proof below, recall that we denote the partial (real) G\^ateaux derivatives with respect to $\#$ by $\del_\#$, and let $\p_i \equiv \p_{x_i}$.
	
By Proposition \ref{prop:reconstr},   solving  equation \eqref{mu-u-syst} is equivalent to solving  \eqref{F-eq'}. By Corollary \ref{cor:red2v}, solving  \eqref{F-eq'} is equivalent to solving the bifurcation equation \eqref{bif-eq}. Hence, we address the latter equation. 
 
Recall that $P$ is the projection onto $\cK = \Null L_{n,\mu=n}=\Null(H_1(n)) 
  \times\{\text{constants}\}\times \{0\}\times \{0\}$. The projection onto constant vector fields in $\cH_0^2$ can be written as the mean value $\langle\alpha\rangle:= \frac{1}{|\Omega'|}\int_{\Omega'}\alpha$. Since $\text{dim} \Null (H_1(n)) =1$ for $n=1$, we may choose $\chi\in\Null (H_1(n))$ such that
  \begin{align} \label{Pu}
  &  P (w,  \alpha, z, \psi) =(s\chi, c,0,0),\\
  & s :=  \langle \chi,w\rangle_{L_n^2}\in \C,\   c:=\langle\alpha\rangle\in \R^2, 
 \end{align}
  and $\chi$ satisfies $\|\chi\|_{L_n^2}^2=\langle |\chi|^2 \rangle =1$ (see \eqref{Hns-ip}), where, recall, $\chi$ is described in \eqref{chi-def}. Hence we may write the $\g$ from the bifurcation equation \eqref{bif-eq} as $\g = (\tilde{\g}_1  \chi, \tilde{\g}_2,0,0)$, 
 where $\tilde{\g}_1,\tilde{\g}_2:\R_{>0} \times \C \times \R^2 \to \C$ are given by
 \begin{align}
 \tilde{\g}_1(\mu,s, c)& :=  \langle \chi,F_1(\mu,v(s,c)+u'(\mu,v(s,c))
 \rangle_{L_n^2}, \\
 \tilde{\g}_2(\mu,s,c)& :=  \langle F_2(\mu,v(s,c)+u'(\mu,v(s,c))\rangle, 
 \end{align}
 where, recall, $F_j$, $j=1, \dots, 4$ are defined by \eqref{Fi-def}, $s\in \C,\ c\in \R^2$ and (see \eqref{Pu}) 
 \begin{align} \label{v-expr} 
 v(s,c)& :=  (s\chi, c,0,0). \end{align} 
 Note that $\tilde\g_1$ and $\tilde\g_2$ are continuously differentiable of all orders in $\mu$, $s$ and $c$ by Corollary \ref{g-anal}. ($\tilde\g_2$ is independent of $\mu$.) The bifurcation equation \eqref{bif-eq} is then equivalent to the equations
 \begin{align}
  \tilde{\g}_1(\mu,s,c)&=0, \label{gam-w-eq}\\
  \tilde{\g}_2(\mu,s,c)&=0. \label{gam-a-eq}
 \end{align}
 
\begin{lemma} \label{gam-a-lem}
	There exists a neighbourhood $U\subset\R_{>0}\times\R_{>0}$ of $(n,0)$ and a unique function $c:U\to\R^2$ with continuous derivatives of all orders such that
  \begin{equation}\label{gam-a-soln}
 \tilde{\g}_2(\mu,s,c(\mu, s^2)) = 0
 \end{equation}
 and
   \begin{equation}\label{OC}
 \|\p^l_\mu c(\mu, s^2)\|_{\R^2} = \cO(|s|^2),\ l=0,1.
 \end{equation}
\end{lemma}

	\begin{proof}
		Recall that $F_2(\mu,u) = H_2(\mu) \alpha + P_0 J_2(\mu, u)$ (see Equation \eqref{F-eq'}), with $P_0$ the projection onto the divergence-free vector fields and 
		\begin{align} \label{u-deco}u=(w, \al, z, \psi)=v+u',\end{align}
		where $v=v(s,c)$ and $u'=u'(\mu, v)$ solves \eqref{projF2}.  By definition, $(\one - P_0)f=\Delta^{-1}\n\div f$ and therefore $\langle(\one - P_0)f \rangle=0$. Hence $\langle P_0 f \rangle=\langle f \rangle$. This and the relation $\langle H_2(\mu)\alpha \rangle = \frac{1}{|\Omega'|} \int_{\Omega'}\curl^*\curl\alpha = 0$ give 
		\begin{equation} \label{tildegam2'}
		\tilde{\g}_2(\mu,s,c) = \langle J_2(\mu, v(s,c) + u'(\mu, v(s,c))) \rangle.
		\end{equation}
		Using \eqref{J2-def}, $\nu=a^n+\tilde\nu$, 
		$\curl_{a^n}w=\curl_{a^n}w- i\tilde\nu \times w$ and that the final term in \eqref{J2-def} vanishes after taking the mean, we find 
		\begin{align} \label{avJ2-transf}
		& \langle J_2(\mu,u)\rangle = 2e\Im\langle (\curl_{a^n} w - i \tilde\nu \times w)J\overline{w}\rangle. 
		\end{align}
		Recall $u'=( w', \alpha',z',\psi')$. Then \eqref{v-expr} and \eqref{u-deco} give $w=s\chi+w'$ and (using that $e=g \sin\theta$) $ \tilde\nu  =e c+\nu'$. Using these relations and $\curl_{ a^n} \chi = 0$ (by \eqref{chi-def}) and \eqref{tildegam2'} and \eqref{avJ2-transf}, we find  for $\overline{\g}_2(\mu,s,c) := (2e)^{-1} |s|^{-2}\tilde{\g}_2(\mu,s,c)$ 
		\begin{align}
		\label{J_2-mean}  \overline{\g}_2(\mu,s,c) &:=  
		- e \langle\Re [(c\times \chi) J\overline{\chi}] \rangle 
		+ \Im s^{-1}\langle(\curl_{ a^n}  w')J \overline {\chi}\rangle\\  
		&\qquad  \qquad  \qquad  \qquad \qquad  \qquad  \qquad  +   \Im \langle\bar R_2(\mu,s,c)\rangle, \label{R_2} \\
		\bar R_2(\mu, s, c) & :=  |s|^{-2} [ -i  (ec\times s\chi)J\overline{ w'}  -i  (ec\times  w')J\overline{ w'}  \\
		& -i   (ec\times  w')J\overline{ s\chi}-i  (\nu'\times  w')J\overline{s\chi} -i  (\nu'\times s\chi)J\overline{ w'} \\ 
		& -i(\nu'\times s\chi)J\overline{s\chi}-i (\nu'\times  w')J\overline{ w'} + (\curl_{ a^n}  w') J\overline{  w'}].   \label{R_a-def} 
		\end{align}
		Note that we expect \eqref{J_2-mean} $=\cO(|s|^2)$ and \eqref{R_2} $=\cO(|s|^4)$.  We now simplify \eqref{J_2-mean}. For the first term on the right-hand side, 
		we  use \eqref{chi-def} 
		and the condition $\langle |\chi|^2 \rangle=1$ to compute 
		\begin{equation}
		\langle\Re [(c\times \chi) J\overline{\chi}] \rangle = -\frac{1}{2}c.
		\end{equation}

		For the second term on the right-hand side of \eqref{J_2-mean}, we use  $\langle f J \overline {\chi}\rangle =\langle f (i\overline {\eta}, \overline {\eta})\rangle= \langle f\bar\eta\rangle(i, 1)=\langle\eta, f\rangle (i,1)$ and integrate by parts to compute
		\begin{align}
		\langle(\curl_{ a^n}  w')J \overline {\chi}\rangle&=\langle\eta, \curl_{ a^n}  w'\rangle (i,1)=\langle \curl_{ a^n}^*\eta,  w'\rangle (i,1). 
		\end{align}

		Abusing notation, we write in what follows $w(\mu,s,c)\equiv w(\mu, v(s,c))$.  Then \eqref{J_2-mean} becomes
		\begin{align}
		\overline{\g}_2(\mu,s,c) = &\frac12 e  c +   \Im  s^{-1}\langle \curl_{ a^n}^*\eta,  w'(\mu,s,c) \rangle\rangle (i,1)   
		+ \Im \langle\tilde R_2(\mu,s,c)\rangle. \label{J_a-mean'}
		\end{align}
		Now, Equation \eqref{Ow2-eq}, with $m=0$, implies that 
		\begin{align} \label{im-est}  | \Im \langle \curl_{ a^n}^*\eta,  w'(\mu,s,c) \rangle\rangle | = \cO(|s|^2).
		\end{align}
		Furthermore, we show below the following estimate on the remainder:
		\begin{align}\label{R2-est}  ||
		\Im \langle \p_c^l \tilde R_2(\mu,s,c)\rangle||_{\R^2} = \cO(|s|^{2-l}), \ l=0,1.
		\end{align}
		Hence $\overline{\g}_2(\mu, 0, 0)=0$.  
		To apply the Implicit Function Theorem to solve for $c$ as a function of $\mu$ and $s$, 
		we have to estimate the derivative 
		\begin{align} \label{C-overline-g-a}
		\p_c\overline{\g}_2(\mu,s,c) = &\frac12 e \id +  \Im s^{-1}\langle \curl_{ a^n}^*\eta,   \p_c  w'(\mu,s,c) \rangle (i,1) \nonumber \\
		& \qquad  \qquad  \qquad  \qquad \qquad  \qquad  \quad + \Im \langle \p_c \tilde R_2(\mu,s,c)\rangle.
		\end{align} 
		at $(n,s,0)$. At the first step, we use the following 
		
		\begin{lemma} \label{lem:dw'dc} 
			Using Dirac's bra-ket notation, we have
			\begin{align} \label{dcw'} 
			(\p_c  w')(n,s,0) &= -n^{-1}e s |\curl^*_{ a^n} \eta \rangle\langle(1,   i)| + \cO(|s|^2).
			\end{align}
		\end{lemma}
		\begin{proof}[Proof of Lemma \ref{lem:dw'dc}] By definition \eqref{P1-def},  $ P^\perp_1$ projects onto the orthogonal complement of the eigenspace of $H_1(n)$ corresponding to the eigenvalue $0$ and therefore  the operator $H_1^\perp(n)$ is invertible  on $\Ran P^\perp_1$. 
			Hence \eqref{mu-u-syst} with $i=1$ can be rewritten as $  w' = - (H_1^\perp(n))^{-1}P^\perp_1 J_1(n, u)$ (which is the first component of \eqref{u-perp-J}), which gives
			\begin{align} \label{dcw'1} 
			\p_c  w' = - (H_1^\perp(n))^{-1}P^\perp_1\p_c J_1(n, u),
			\end{align}
			where $u\equiv u(s,c):=v(s,c) + u'(\mu,v(s,c))$.	  
			\DETAILS{ Since $F_1(\mu,u) = H_1(\mu)w + J_1(\mu,u)$ and the first term on the right-hand side is independent of $\alpha$, we have
				\begin{align} \label{dcw'2}
				\p_c  w' = - (H_1^\perp(n))^{-1}P^\perp_1\p_c  F_1(\mu,u). 
				\end{align}
				To simplify the computation, we compute the right-hand side only for the leading term $ F'(\mu,s,c):=F_1(\mu,v(s,c))$.
				
				Using that  $F_1(\mu,u)=G_1(m^{n, r}+u)|_{\xi = \sqrt{2\mu}/g}=0$, where $G_1$ is defined in \eqref{G1-expl} and $m^{n, r} :=  (0,\frac1e a^n,0,\xi)$  (explicitly, $w=s\chi+w', a=\frac1e a^n+\alpha$, $\phi=\xi+\psi$, $\nu= a^n+\tilde\nu$ and 
				$\xi = \sqrt{2\mu}/g$), and recalling that $ v(s,c) =  (s\chi, c, 0, 0)$ (see \eqref{v-expr}) so that, in particular, $\nu|_{u=v(s,c)}=a^n +e c$ and $\phi|_{u=v(s,c)}=\xi= \sqrt{2\mu}/g$, we find  
				\begin{align} \label{F'-expl} & F'(\mu,s,c) = s [\curl_{a^n +e c}^* \curl_{a^n +e c} + \frac{1}{2}2\mu - i n J + g^2s^2(\overline{\chi} \times \chi)J]\chi. 
				\end{align}}	
			By \eqref{J1-def}  and \eqref{M}, we have 
			\begin{align} \label{dcJ1} & \p_c J_1(n, u) = \p_c [\curl_{\nu}^* \curl_{\nu}w]. 
			\end{align}
			Using $w=s\chi+w',$ 
			$\nu= a^n+ec+\nu'$ and  $\curl_{\nu}=\curl_{a^n} + i J (ec+\nu') \cdot$, $\curl_{\nu}^*=\curl_{a^n }^* - i J (ec+\nu')$ and that $\nu'=\cO(|s|^2)$, we compute 
			\begin{align} \label{dcF'1}  
			\p_c J_1(n, u) c' &= s\p_c  [\curl_{\nu}^* \curl_{\nu}]\chi c' + \cO(|s|^2)\notag\\
			&= s ie [-Jc' \curl_{\nu}+\curl_{\nu}^*  Jc' \cdot]\chi + \cO(|s|^2)\\
			\label{dcF'2} &= s ie [-Jc' \curl_{a^n+ec}+\curl_{a^n+ec}^*  Jc' \cdot]\chi + \cO(|s|^2). 
			\end{align}
			Since $\curl_{a^n}\chi=\n_1i\beta - \n_2\beta=i\bar\p_{a^n}\beta = 0$ and $ Jc' \cdot \chi= (-c_2', c_1') \cdot (\beta, i \beta)=-c_2'\beta +  c_1' i \beta=i (c'_1+ic_2')\beta$ and therefore $\curl_{a^n}^*  Jc' \cdot \chi=i\curl_{a^n}^* \beta (c'_1+ic_2')$, this yields
			\begin{align} \label{dcF'3}  \p_c  J_1(n, u) c'\big|_{c=0} &= - s e \curl_{a^n}^* \beta (c'_1+ic_2') + \cO(|s|^2). 
			\end{align}
			By Proposition \ref{H_w-prop}(ii), $\Null{(H_1(\mu)-\mu+n)}=\{\chi=(\beta, i\beta): \curl_{ a^n} \chi =i\bar\p_{ a^n} \eta= 0\}$. The relation  $\curl_{a^n}\chi= 0$ implies also $\langle\chi,   \curl_{ a^n}^* \chi \rangle=\langle \curl_{ a^n} \chi,   \chi \rangle=0$, which, for $n=1$, gives that $P^\perp_1   \p_c J_1(n, u) c' = \p_c J_1(n, u) c'  $ and therefore 
			\begin{align} \label{dcF'} P^\perp_1  \p_c J_1(n, u) c' &=  -s e \curl_{a^n}^* \beta (c'_1+ic_2') + \cO(|s|^2). 
			\end{align}
			By \eqref{chi-def}, we have $\curl_{a^n}^*\eta = i\n_{a^n}\beta$, and by \eqref{H1-grad}, we have 
			$H_1(n)\nabla_{ a^n}\beta = n\nabla_{ a^n}\beta$; hence $(H_1^\perp(n))^{-1}\curl_{ a^n}^* \beta=n^{-1}\curl_{ a^n}^* \beta$. This relation, together with \eqref{dcw'1} and \eqref{dcF'}, yields
			\begin{align} \label{dcw'3}
			\p_c  w' c'&=  s e n^{-1} \curl_{a^n}^* \beta (c'_1+ic_2') + \cO(|s|^2), 
			\end{align}  
			which gives \eqref{dcw'}.   	
		\end{proof}
		\DETAILS{ Note that we can rewrite \eqref{dcw'} as 
			\begin{align} \label{dcw'4} 
			(\p_c  w')(n,s,0) &= -n^{-1}s\left(\begin{array}{cc}
			(\nabla_{a^n})_2 \eta & i(\nabla_{ a^n})_2\eta \\ -(\nabla_{ a^n})_1 \eta & -i(\nabla_{ a^n})_1 \eta \end{array} \right) + \cO(|s|^2)\\
			\label{dcw'5}   &:= -n^{-1}s (\curl^*_{ a^n} \eta,  - i \curl^*_{ a^n}\eta) + \cO(|s|^2).
			\end{align} 
			
			Recall from Equation \eqref{chi-def} that $\chi$ is of the form $\chi = (\eta,i\eta)^T$, where $\chi$ satisfies $\curl_{a^n}\chi = 0$. Hence, in particular, 
			\begin{equation}
			\curl^*_{a^n}\chi = \left(\begin{array}{cc}
			(\nabla_{a^n})_1 \eta & i(\nabla_{ a^n})_1\eta \\ (\nabla_{a^n})_2 \eta & i(\nabla_{ a^n})_2 \eta \end{array} \right)= (\nabla_{ a^n}) \eta,  i \nabla_{ a^n}\eta).
			\end{equation}
		}
		

		Using Equation \eqref{dcw'}, 
		we calculate the second term on the right-hand side of \eqref{C-overline-g-a} at $(n,s,0)$: 
		\begin{align}
		&\Im s^{-1}\langle \curl_{ a^n}^*\beta,   \p_c  w'(\mu,s,c)c' \rangle (i,1) \nonumber \\
		&  =  en^{-1}\Im \langle \curl_{ a^n}^*\beta,  \curl_{a^n}^* \beta \rangle (c'_1+ic_2') (i,1).
		\end{align}
		The  inner product term is real. Integrating it by parts and using that, by  Equation \eqref{eta-eqn}, $\beta$ satisfies $\curl_{ a^n} \curl_{a^n}^* \beta=-\Delta_{ a^n}\beta = n\beta$ and using that $\|\beta\|_{L^2_n}^2 = \frac12 \|\chi\|_{L^2_n}^2 = \frac12$, gives 
		\begin{align}
		\langle \curl_{ a^n}^*\beta,  \curl_{a^n}^* \beta \rangle 
		=& \langle \beta,-\Delta_{a^n}\beta \rangle_{L^2_n}=\frac12 n. 
		\end{align}
		The last two equations and the relation $ \im(c'_1+ic_2') (i,1)=\im\bigg(\begin{array}{cc} i & - 1 \\ 1 & i \end{array} \bigg)c'=\one c'$ imply
		\begin{align}\label{C-overline-g-a-lim}  &\Im s^{-1}\langle \curl_{ a^n}^*\beta,   \p_c  w'(\mu,s,c) \rangle (i,1)   =  \frac12 e \id.
		\end{align}
		This, together with \eqref{C-overline-g-a}, gives
		\begin{align} \label{C-overline-g-a'}
		\p_c\overline{\g}_2(n,s,0) = & \frac12 e\id +  \frac12 e\id+ \Im \langle \p_c \tilde R_2(n,s,0)\rangle.
		\end{align}

		Therefore,  
		\eqref{C-overline-g-a'} and \eqref{R2-est} (with $l=1$) imply 
		\begin{align}
		\p_c\overline{\g}_2(n,0,0) = e\id,
		\end{align}
		\DETAILS{
			and so 
			\begin{align}
			\det\p_c\overline{\g}_2(n,0,0) = e^4+\frac{e^4}{\pi^2}\langle \curl^*_{a^n}\eta, (H_1^\perp(n))^{-1}\curl^*_{a^n}\eta \rangle_{L^2_n}^2 > 0,
			\end{align}
		}
		proving that $\p_c\overline{\g}_2(n,0,0)$ is invertible, as required.
		
		Recall that, by \eqref{J_a-mean'}, \eqref{im-est} and \eqref{R2-est} (with $l=0$), we have
		\begin{align}
		\overline{\g}_2(n,0,0) = 0.
		\end{align}
		
		Since $\p_c\overline{\g}_2(n,0,0)$ is invertible, by the Implicit Function Theorem there exists a unique function $\tilde c:\R_{>0}\times\C\to\R^2$ with continuous derivatives of all orders such that $\overline{\g}_2(\mu,s,\tilde c(\mu,s))=0$ for $(\mu,s)$ in a sufficiently small neighbourhood of $(n,0)$. 
		Furthermore, the symmetry \eqref{gamma-sym} implies that $\overline{\g}_2(\mu,|s|,\tilde c(\mu, s)) = \overline{\g}_2(\mu,e^{i\arg s}|s|,\tilde c(\mu, s)) = \overline{\g}_2(\mu,s,\tilde c(\mu, s))=0$, so by the uniqueness of the branch $\tilde c(\mu, s)$ we have 
		\begin{equation}\label{C-sym}
		\tilde c(\mu, s) = \tilde c(\mu, |s|).
		\end{equation}
		In particular, $\partial_{\mu}^l\tilde c(\mu, s)$, $l=0,1$, restricted to $s\in\R$ are even functions with continuous derivatives of all orders; thus $\partial_s \partial_{\mu}^l\tilde c(\mu, 0) = 0$ and hence $\partial_{\mu}^l \tilde c(\mu, s) = \cO(|s|^2)$, since the first two terms of the Taylor expansion are $0$. We define $c:\R_{>0}\times \R_{>0}\to\R^2$ by $c(\mu, s):= \tilde c(\mu, \sqrt{s})$, which is a function with continuous derivatives of all orders satisfying $||\partial_{\mu}^l c(\mu, s^2)||_{\R^2} = \cO(|s|^2)$, $l=0,1$, and $\tilde{\g}_2(\mu,s,c(\mu, s^2)) = |s|^2\overline{\g}_2(\mu, s, c(\mu, s^2)) = 0$, as required.
	\end{proof}

 \begin{lemma} \label{gam-w-lem}
 	For $\epsilon>0$ sufficiently small, there exists a unique function $\mu:[0,\epsilon)\to\R_{>0}$ with continuous derivatives of all orders such that
 	\begin{equation}\label{gam-w-soln}
 	\tilde{\g}_1(\mu(s^2),s,c(\mu(s^2), s^2)) = 0.
 	\end{equation}
 \end{lemma}
 
 \begin{proof}
	
	To simplify notation for this lemma, we set $u=v_s + u'_s$, with $v\equiv  v_s \equiv  (s\chi, c(\mu,s^2),0,0), \ u'\equiv  u'_s \equiv  u'(\mu, v_s),\ c \equiv  c(\mu,s^2)$.
	
We first show that $\tilde{\g}_1(\mu,s,c)\in\R$ for $s\in\R$. Since $u'$ by definition solves $P_1^\perp F_1(\mu, v+ u') = 0$, where $P_1^\perp  w' =  w'$ and $P_1^\perp$ is self-adjoint, we have
\begin{align}
&\langle  w', F_1(\mu,v+u')\rangle_{L^2_n} = \langle  w', P_1^\perp F_1(\mu,v+u')\rangle_{L^2_n} = 0.
\end{align}
Therefore, for $s\neq 0$, we find
\begin{align}
\tilde{\g}_1(\mu,s,c) &= s^{-1}\langle s\chi,F_1(\mu,v+u'\rangle_{L^2_n} \nonumber \\
&= s^{-1}\langle s\chi +  w', F_1(\mu,v + u')\rangle_{L^2_n}, 
\end{align}
which is real by Proposition \ref{F-prop} $(v)$. Furthermore, by equations 
 \eqref{gamma-sym} and \eqref{C-sym}, we have 
 $\tilde{\g}_1(\mu,s,c(\mu, s^2)) = e^{i\text{arg}(s)}\tilde{\g}_1(\mu,|s|,c(\mu,|s|^2))$, so we may restrict $s$ to be real.
 
 Next, we show that 
 \begin{align} \label{gam-w-ineq}
 \tilde{\g}_1(n,s,c(n,s^2)) = \cO(|s|^2)
 \end{align}  
  Indeed,
 \begin{align}
  |\tilde{\g}_1(n,s,c(n,s^2))| &\leq \|\chi\|_{L^2_n} \|F_1(n,v+u')\|_{L^2_n} \nonumber \\
  &\leq \|\chi|\|_{L^2_n} \big[\|H_1(n)(s\chi +  w')\|_{L^2_n} \nonumber \\ 
  &\quad + \|J_1(n, v+u') \|_{L^2_n})\big]. 
  \end{align}
  Recall that $H_1(n)\chi =0$, so that
  \begin{align}
  |\tilde{\g}_1(n,s,c(n,s^2))| &\leq \|\chi\|_{L^2_n} \big[\|H_1(n)\|_{L^2_n \otimes (L^2_n)^*}\| w'\|_{L^2_n} \nonumber \\
  &\quad + \|J_1(n, v+u')\|
 \end{align}
By the definition $v\equiv  v_s \equiv  (s\chi, c(\mu,s^2),0,0)$ and equation \eqref{OC}, $\|v\|_{\cX} = \cO(|s|)$; hence by Proposition \ref{prop:sol-u'}, 
\begin{align}
\|w'\|_{L^2_n} \leq \|w'\|_{\cH^2_n} = \cO(|s|^2).
\end{align} 
Furthermore, by equation \eqref{J-bds} and recalling that $H_1(n) \chi = 0$, 
\begin{align}
\|J_1(n, v+u')\|_{L^2_n}
\leq \|J_1(n, v+u')\|_{\cH^2_n}\lesssim \|v+u' \|_{\cX}^2 = \cO(|s|^2).
\end{align}
This proves that $ \tilde{\g}_1(n,s,c(n,s^2))$ 
 is $\cO(|s|^2)$, as required.

In light of equation \eqref{gam-w-ineq}, we can define
a function $\overline{\g}_1 : \R_{>0}\times\R_{>0} \to \R$ with continuous derivatives of all orders by
\begin{equation}
 \overline{\g}_1(\mu,s)\equiv 
 		\begin{cases}
 			s^{-1}\tilde{\g}_1(\mu,s,c(\mu,s^2)), \quad s\neq 0, \\
 			0, \quad s=0.
 		\end{cases}
\end{equation}
We now find a non-trivial branch of solutions $(\mu,s) = (\tilde{\mu}(s),s)$ by applying the Implicit Function Theorem to $\overline{\g}_1$. 
 First, we prove the following proposition to bound the polynomials of functions appearing below:

\begin{lemma} \label{lem:banach-poly}
	Let $X$ be one of the spaces $\cH^2_n$, $\cH_0$ or $\cH^2$ defined before equation \eqref{L2n-space}. Let
	$p(x_1,...,x_n)$ be a polynomial with positive coefficients and let $f_1,...,f_n\in X$. Then 
	$\|p(f_1,...,f_n)\|_X \lesssim p(\|f_1\|_X,...,\|f_n\|_X)$.
\end{lemma}
\begin{proof}
 Write $p(x_1,...,x_n) = \sum_{|\alpha|\leq N} p_{\alpha} x^{\alpha}$, where $\alpha = (\alpha_1,...,\alpha_n)$ is a multi-index, $x^\alpha = \prod_{i=1}^n x_i^{\alpha_i}$ and $p_{\alpha} \ge 0$. Since by the Sobolev Embedding Theorem (see e.g. \cite{AF}), $X$ is a Banach algebra, we have
 \begin{align}\notag
  \|p(f_1,...,f_n)\|_X &\leq \sum_{|\alpha|\leq N} p_{\alpha} \|f^{\alpha}\|_X \\
\notag&\lesssim \sum_{|\alpha|\leq N} p_{\alpha} \prod_{i=1}^n \|f_i\|_X^{\alpha_i} \\
\notag&= p(\|f_1\|_X,...,\|f_n\|_X),
 \end{align}
which implies the desired result. 
\end{proof}

\begin{lemma} \label{tilde-mu} There exists $\epsilon>0$ and a unique function $\tilde{\mu}:(-\sqrt{\epsilon},\sqrt{\epsilon})\to\R_{>0}$ with continuous derivatives of all orders such that $\tilde{\mu}(0)=n$ and $\mu=\tilde{\mu}(s)$ solves $\overline{\g}_1(\mu,s)=0$ for $s\in(-\sqrt{\epsilon},\sqrt{\epsilon})$.
Moreover,  $\tilde{\mu}$ is an even function: $\tilde{\mu}(s) = \tilde{\mu}(-s)$.
\end{lemma}
\begin{proof}
Recall that $F_1(\mu, u) = H_1(\mu) w + J_1(\mu, u)$ (where $H_1(\mu)$ and $J_1(\mu,u)$ are defined in \eqref{Hw} and \eqref{J1-def}). Using that $\partial_{\mu}F_1(\mu,u) = (1+\frac{g}{2\sqrt{2\mu}}\psi)w$ and setting $u=v_s + u'_s$, with $v\equiv  v_s  \equiv   (s\chi, c(\mu,s^2),0,0), \ u'\equiv  u'_s  \equiv   u'(\mu, v_s),\ c =  c(\mu,s^2)$, we compute
\begin{align}
 \partial_{\mu}&[s^{-1} F_1(\mu, v + u')] = s^{-1}(1+\frac{g}{2\sqrt{2\mu}}\psi') (s\chi +  w') \nonumber\\ 
 &+s^{-1}\sum\limits_{i=1}^4 \del_{u_i} F_1w(\mu, v + u') (\partial_{\mu}v_i + \partial_{\mu}u'_i) \nonumber \\
 \label{mu-F}&= s^{-1}(1+\frac{g}{2\sqrt{2\mu}}\psi') (s\chi +  w') + s^{-1}\del_{\alpha}F_1(\mu, v + u')\partial_{\mu}c \nonumber\\ 
 &+s^{-1}\sum\limits_{i=1}^5 \del_{u_i} F_1(\mu, v + u') \partial_{\mu}u'_i. 
\end{align}
By Lemma \ref{gam-a-lem}, $\|\partial_{\mu}^l c\|_{\R^2} = \cO(|s|^2)$, $l=0,1$. Since $||v||_{\cX}$ is $\cO(|s|)$, by Proposition \ref{prop:sol-u'} 
 the terms $\|\partial_\mu^l u'_i\|$ ($l=0,1$, $i=1, \dots, 4$, with the norms taken in the appropriate spaces), are $\cO(|s|^2)$.  By Lemma \ref{lem:banach-poly}, this implies that all terms in \eqref{mu-F} containing $c, w',\alpha',z',\psi'$ or their $\mu$-derivatives vanish at $(\mu,s) = (n,0)$. Therefore
\begin{equation}
 \partial_{\mu}[s^{-1}F_1(\mu,v + u')]|_{(\mu,s)=(n,0)} = \chi
\end{equation}
and hence
\begin{equation}
 \partial_{\mu}\overline{\g}_1(n,0) = \langle \chi, \partial_{\mu}[s^{-1}F_1(\mu,s)]\,|_{(\mu,s)=(n,0)}\rangle_{L^2_n} = \|\chi\|_{L^2_n}^2 \neq 0.
\end{equation}
Since $\overline{\g}_1(\mu,s)$ is continuously differentiable of all orders in $\mu$ and $s$, by the Implicit Function Theorem, we obtain the first statement of the lemma. 

By the symmetry $\overline{\g}_1(\mu,-s)=-\overline{\g}_1(\mu,s)$ of $\overline{\g}_1$ and the uniqueness of the branch $\tilde{\mu}(s)$, we have $\tilde{\mu}(s) = \tilde{\mu}(-s)$, which gives the second statement. 
\end{proof}
We define $\mu(s) \equiv  \tilde{\mu}(\sqrt{s})$, which is a function with continuous derivatives of all orders for $s\in [0,\epsilon)$ for the same reasons that $c(\mu,s)  :=   \tilde c(\mu, \sqrt s)$ was shown to be continuously differentiable of all orders in Lemma \ref{gam-a-lem}. Furthermore, $\mu$ satisfies $\tilde{\g}_1(\mu(s^2),s,c(\mu(s^2), s^2)) = s\overline{\g}_1(\mu(s^2),s,c(\mu(s^2), s^2)) = 0$, as required.
\end{proof}

We will now use the branch of solutions to \eqref{gam-w-eq} - \eqref{gam-a-eq}, provided by Lemmas \ref{gam-a-lem} and \ref{gam-w-lem}, and Corollary \ref{cor:red2v}  to obtain the corresponding unique branch, $(\mu_s, u_s)$, of solutions to \eqref{F-eq'}, with 
\begin{align}
\mu_s & \equiv   \mu(s^2), \quad u_s \equiv   v_s+u'_s, \label{mu-fin} \\
v_s & \equiv   (s\chi,c_s,0,0), \quad c_s  \equiv   c(\mu_s,s^2),  \label{vs'}\\ 
u'_s & \equiv   u'(\mu,v_s). \label{u-perp-fin}
\end{align}
\DETAILS{
	Let $u^\perp = (w^\perp, \alpha^\perp, z^\perp, \psi^\perp)$ and define functions $\tilde g_w:[0,\epsilon)\to\overrightarrow{\cH}_n^2$, $\tilde g_a:[0,\epsilon)\to\overrightarrow{\cH}^2_0$, $\tilde g_z:[0,\epsilon)\to\overrightarrow{\cH}^2_0$, $g_{\phi}:[0,\epsilon)\to\cH^2$ and $\tilde g_{\xi}:[0,\epsilon)\to\R_{>0}$ by
	\begin{align}
	\tilde g_w(s) &:=  w^\perp, \label{tilde-gw} \\
	\tilde{g}_a(s) &:=  \alpha^\perp, \label{tilde-ga} \\
	\tilde g_z(s) &:=  z^\perp, \label{tilde-gz} \\
	\tilde{g}_{\phi}(s) &:=  \psi^\perp. \label{tilde-gphi}
	\end{align}
}
\eqref{mu-fin} - \eqref{u-perp-fin} have continuous $s$-derivatives of all orders because each component function has continuous derivatives of all orders. Symmetry \eqref{u2-sym} with $\delta = \pi$ and the relation $T_{\pi}(f_1,f_2,f_3,f_4) = (-f_1,f_2,f_3,f_4)$ imply that $(u_s')_1$ is an odd function of $s$ and $(u_s')_2$, $(u_s')_3$ and $(u_s')_4$ are even functions of $s$. 
Arguing as in the case of Lemma \ref{gam-a-lem} above shows that the 
 functions:
\begin{align}
g_1(s) & :=   \begin{cases}
\frac{1}{\sqrt{s}}(u'_{\sqrt s})_1, \quad &s\neq 0,\\
0, \quad &s=0, 
\end{cases}\quad  
g_2(s)  :=   c_{\sqrt s} + (u'_{\sqrt s})_2, \\
 g_3(s)& :=   (u'_{\sqrt s})_3,\quad  g_4(s)  :=   (u'_{\sqrt s})_4,\quad  g_5(s)  :=   \mu_{\sqrt s}-n, 
\end{align}
 are well-defined for $s\geq 0$ and have continuous derivatives of all orders. By  Proposition \ref{prop:sol-u'}, these functions have the properties listed in Theorem \ref{bifurcation-result}. The above definitions and equations \eqref{mu-fin} - \eqref{u-perp-fin} imply $u_s = (s\chi,\frac1e a^n,0,0)+(g_1(s), \dots, g_4(s))$. Hence, this solution is of the form \eqref{bif-soln}. Now, by Proposition \ref{prop:reconstr}, this also solves system \eqref{WS-eq1-resc} - \eqref{WS-eq4-resc}, 
\DETAILS{More explicitly,
\begin{align} \label{s-asymp}
 \begin{cases}
  w_s = s\chi + sg_1(s^2) \\
  a_s = \frac1e a^n + g_2(s^2) \\
  z_s = g_3(s^2) \\
  \phi_s = \sqrt{2n}/g + g_4(s^2) \\
  \xi_s = \sqrt{2n}/g + g_5(s^2)
 \end{cases}
\end{align}}
  completing the proof.
\end{proof}


\section{Proof 
 of  Theorem \ref{thm:AL-exist'}(a), (b)} \label{sec:main-thm-abc}

Recall that $M_W$, $M_Z$, $M_H$ are the masses of the $W$, $Z$ and Higgs bosons, respectively, and that $\tau$ is the shape parameter of the lattice $\cL$ (see the paragraph before Theorem \ref{thm:lattice-shape} of Section \ref{sec:probl-res}).
We introduce the notation
\begin{align} \label{average}
 \langle f \rangle:= \frac{1 }{|\Om'|}  \int_{\Om'} f,  
\end{align}
the average of $f$ over fundamental domain 
 $\Om'=\sqrt{\frac{2\pi}{|\Omega|}}\Omega$. 
 \DETAILS{Furthermore, we will use the space (cf. \eqref{L2-space})
\begin{align}
\label{L2-space'}& L^2_{\cL}  :=   \{\psi\in L^2_{loc}(\R^2,\R) : \psi(x+s)=\psi(x)\ \forall s\in\cL\}.
\end{align}}
Furthermore, we introduce the function (cf. \cite{MT}) 
\begin{equation}\label{eta'}
\eta\equiv \eta_{m_z,m_h}(\tau):=[m_w^2 \alpha_{m_z,m_h}(\tau)+\sin^2\theta]^{-1},
\end{equation} with, recall, $m_w :=  \sqrt n$,  $m_z :=  \frac{\sqrt n }{\cos\theta}$ and $m_h :=  \frac{\sqrt{4\lambda n}}{g}$ the masses of the rescaled W, Z and Higgs boson fields, $w$, $z$ and $\phi$, respectively, and 
\begin{align} 
\label{alpha'}
	&\alpha_{m_z,m_h}(\tau)  :=   \langle |\chi|^2 G_{m_z,m_h}(|\chi|^2)\rangle / (\langle|\chi|^2\rangle)^2.
	\end{align}	
Here  $\chi$ is defined in \eqref{chi-def}  and  $G_{m,m'}$ is the operator-family on the space \eqref{L2-space} 
 given by 
\begin{align} \label{U-def}
G_{m, m'}:= G_{m'}-G_{m},\ \quad \text{ where }\ \quad G_{m}:= (-\Delta+ m^2)^{-1}.   
\end{align}
\DETAILS{, and 
\begin{align}
X_r(x)  :=   r^{-1} \chi(r^{-1}x),
\end{align}
with $\chi$ given in \eqref{chi-def}.}
Note that $G_{m, m'}>0$ for $m'<m$. Recall  $M_W :=   \frac1{\sqrt{2}} g\varphi_0$, $M_Z :=   \frac1{\sqrt{2}\cos\theta} g\varphi_0$ and $M_H :=  \sqrt{2}\lambda\varphi_0$. 

\begin{proposition} \label{prop:b-asym}
	If $M_Z < M_H$, the parameter $s$ of the branch \eqref{bif-soln} is related to the magnetic field strength by
	\begin{align}\label{s2-asym}
	s^2  &= 
	\frac{n}{g^2 \langle|\chi|^2\rangle}\eta_{m_z,m_h}(\tau)\om + R_s(\om),\ \quad \om:=1-\frac{M_W^2}{eb},
	\end{align}
	where $R_s(\om)$ is a real, smooth function of $\om$ 
	 satisfying 
	\begin{align}\label{s2-asym-rem}
	R_s(\om) = \cO(|\om|^2). 
	\end{align}
\end{proposition}

Before proving Proposition \ref{prop:b-asym}, we shall see how it implies statements (a) and (b) of  Theorem \ref{thm:AL-exist'}. 
\begin{proof}[Proof of  Theorem \ref{thm:AL-exist'}(a), (b)] Since the operator $G_{m_z,m_h}$ is positivity preserving,  the function  $G_{m_z,m_h}(|\chi|^2)$ is positive for $M_Z < M_H$, and hence  $\alpha_{m_z,m_h}(\tau)$ and $\eta_{m_z,m_h}(\tau)$ are positive. 
Furthermore, when the right-hand side of \eqref{s2-asym} is positive, we 
solve \eqref{s2-asym} for $s$ as a function of $b$, $s=s(b)$, having continuous derivatives of all orders. When $|1-\frac{M_W^2}{eb}| \ll 1$, the right-hand side of \eqref{s2-asym} is positive if and only if $1-\frac{M_W^2}{eb} > 0$.\footnote{
The condition $0 < 1-\frac{M_W^2}{eb} \ll 1$ is equivalent to the condition $0 < 1-\frac{M_W^2}{2\pi}|\cL | \ll 1$ of Theorem \ref{thm:AL-exist'}.} Plugging $s = s(b)$ into \eqref{bif-soln} (i.e. passing from the bifurcation parameter $s$ to the physical parameter $b$), undoing the rescaling \eqref{rescaling}, and recalling that $b_* = \frac{M_W^2}{e}$, we arrive at the branch, $U_{\cL} \equiv (W_b, A_b, Z_b, \varphi_b)$, of solutions of \eqref{WS-eq1'} - \eqref{WS-eq4'}, which has the properties listed in statements (a) and (b) of Theorem \ref{thm:AL-exist'}. \end{proof} 


The following statement follows from the proof above:

\begin{lemma} \label{lem:U-differ} $U_{\cL}$ is continuously 
  differentiable of all orders in $b$ for $b$ in an open right half-interval of $b_*$ 
\end{lemma}

\begin{proof}[Proof of Proposition \ref{prop:b-asym}]
	Consider the solution branch $(\mu_s,w_s,a_s, z_s)$ given in equation \eqref{bif-soln}  and described in Theorem \ref{bifurcation-result}. Using Taylor's theorem for Banach space-valued finctions (see e.g. \cite{berger}) and recalling the relation  $\xi = \sqrt{2\mu}/g$, we may expand this branch in $s$ as follows:
	\begin{align}\label{s-asym}
	&\begin{cases}
	w_s = s\chi + s^3  w' + 
	 \cO(|s|^5), \\
	a_s = \frac1e a^n + s^2 a' + s^4a'' + 
	\cO(|s|^6), \\
	z_s = s^2 z' + 
	\cO(|s|^4), \\
	\psi_s :=  \phi_s-\xi_{s} = s^2 \psi' + 
	\cO(|s|^4), \\
	\xi_{s} :=   \sqrt{2\mu_s}/g= \sqrt{2n}/g +s^2 \xi' + 
	\cO(|s|^4), 
	\end{cases} 
	\end{align}
	where $w',a',z',\psi',\xi'$ and $a''$ 
	 are the coefficients of $s^2$ and $s^4$, respectively, in the Taylor expansion of $g_j(s^2)$, $j=0,...,5$, in \eqref{bif-soln}, and   $\chi$ is defined in \eqref{chi-def}. 
	 Here $\cO(|s|^p)$ stand for various error terms 
	  which, together with their (covariant) derivatives, have norms of order $\cO(|s|^p)$ 
	   when taken in the appropriate spaces. 
	
	
	
	To rewrite the asymptotics in terms of the parameter  $b$, we analyze how $s$ depends on $b$. 
For this, we use the definitions $\xi_{s} =  \sqrt{2\mu_s}/g$ and $\mu:= \frac12 (g \xi)^2=\frac12 (g r \vphi_0)^2$, with $r := \sqrt{\frac{n}{eb}}$ (see \eqref{r}) to find the following equation  for $s^2$:
\begin{align}\label{s-eq}\xi_{s} = \sqrt{\frac{n}{e b}}\varphi_0.\end{align} 
To solve this equation for $s^2$, we use the Implicit Function Theorem. By \eqref{s-asym}, we can write $\xi_{s}=\sqrt{2n}/g+ g_{\xi}(s^2)$, where recall, $g_\xi(0) = 0$ and $g'_{\xi}(0) = \xi'$. Hence, we have to show that  $\xi' \neq 0$. 
\begin{lemma} \label{tildephi0neq0} We have $\xi' \neq 0$.\end{lemma}\begin{proof}
We find relations between $\psi'$, $a'$ and $z'$ entering \eqref{s-asym}. Plugging \eqref{s-asym} into Equations \eqref{WS-eq2-resc} - \eqref{WS-eq4-resc}, we obtain at order $s^4$
	\begin{align}\label{os4-eqns}
	\begin{cases}
	-\Delta a' - e\curl^* |\chi|^2 = 0\\
	(-\Delta+ \frac{n}{\cos^2\theta}) z' - g\cos\theta\curl^* |\chi|^2 = 0\\
	(-\Delta + \frac{4\lambda n}{g^2})\psi' + \frac g2 \sqrt{2n} |\chi|^2 = 0.  \\
	\end{cases}
	\end{align}
We solve these equations, using that $\curl^* |\chi|^2=\curl^* (|\chi|^2-\langle |\chi|^2\rangle')$ and $|\chi|^2-\langle |\chi|^2\rangle'\in \Ran (\Delta)$, for the first one, to find 
\footnote{To check the 
solutions, one may use  that $\curl\curl^* =-\Delta $.} 
	\begin{align}\label{s-soln}
	\begin{cases}
	a' = e \curl^* G_0(|\chi|^2-\langle |\chi|^2\rangle')\\
	z' = g\cos\theta\curl^*  G_{m_z}(|\chi|^2) \\
	\psi' = -\frac g2\sqrt{2n}\; G_{m_h}(|\chi|^2),  \\
	\end{cases}
	\end{align}
	where 
 $G_m:=(-\Delta+ m^2)^{-1}$ acting on the space \eqref{L2-space} 
  (cf. \eqref{U-def}), and	$m_z :=  \frac{\sqrt n }{\cos\theta}$ and $m_h :=  \frac{\sqrt{4\lambda n}}{g}$ are the masses of the rescaled Z and Higgs boson ($\Phi$) fields, $z$ and $\phi$, respectively. 
	\DETAILS{, so that, for $f\in\cL^2_{loc}$,
	\begin{align}\label{Um}
	(U_m f)(x)  :=   \begin{cases}
  \frac{1 }{|\Om'|} \int_{\Om'} K_0(m|x-x'|) f(x') d^2x', \:\:\: m>0\\ 
  \frac{1 }{|\Om'|} \int_{\Om'} -\text{ln}(|x-x'|) f(x') d^2x', \quad\quad m=0,
	\end{cases}
	\end{align}
	with {\bf $K_0$ a modified Bessel function of the third kind}. (Recall that $|\Om'|=2\pi$.) Note that $U_m$ satisfies
	\begin{align}
	(-\Delta + m^2) U_m f = f\ \text{ and }\ 
	m^2\langle U_m f \rangle' = \langle f \rangle' \quad (m\neq 0). \label{U-prop}
	\end{align}}
Next, we use the following relation proven in Appendix \ref{Sec:Lemma-tilde-phi0}:
	\begin{align} \label{phi0-eqn'}
	\lan g\sqrt{2n} \xi'  |\chi|^2\ran  =\lan -g\sqrt{2n} \psi' |\chi|^2 + \curl\nu'|\chi|^2 - g^2 |\chi|^4\ran ,
	\end{align}
where, recall, $\nu' := g(a'\sin\theta + z'\cos\theta)$. First, we evaluate  $\curl\nu'$. The   relations $(-\Delta + m^2) G_m  = \one$ and   $\curl\curl^* =-\Delta $ 
  imply  $\curl a'= e (|\chi|^2-\langle |\chi|^2\rangle)$. Next, the second relation in   \eqref{s-soln} and the relation  $\curl\curl^* =-\Delta $ yield $\curl z'=g\cos\theta (-\Delta)(-\Delta+ m_z^2)^{-1} |\chi|^2$, which, together with $m_z :=  \frac{\sqrt n }{\cos\theta}$, gives $\curl z' = g \cos\theta |\chi|^2- g\frac{ n }{\cos\theta}G_{m_z}|\chi|^2$. 
  Finally, using that  $e:=g\sin \theta$, we conclude that
 \begin{align} \label{v3-simp'}
 	\curl \nu' = g^2|\chi|^2 - e^2\langle|\chi|^2\rangle - 
	g^2n G_{m_z}(|\chi|^2).
	\end{align}
	 Plugging the last relation and equation \eqref{s-soln} into 
 the relation \eqref{phi0-eqn'}, gives 
 		\begin{align}\label{phi0-eqn-s}
		g\sqrt{2n} \xi' \langle & |\chi|^2 \rangle  = -g^2[m_w^2\langle G_{m_z,m_h}(|\chi|^2)|\chi|^2\rangle+\sin^2\theta (\langle |\chi|^2\rangle)^2],
		\end{align}
		where $m_w :=  \sqrt n$ is the mass of the rescaled $W$ boson field $w$ and the operator-family $G_{m,m'}$ is defined by \eqref{U-def}. 
		\DETAILS{, for $f\in\cL^2_{loc}$,
		\begin{align}
		(U_{m_1, m_2}f)(x)  :=  &  
  \frac{1 }{|\Om'|} \int_{m_2}^{m_1} \int_{\Om'} |x-x'| K_1(m|x-x'|) f(x') d^2x' dm \nonumber \\
		=& (U_{m_1}f)(x)-U_{m_2}(f)(x),
 \end{align}
		with {\bf $K_1$ a modified Bessel function of the third kind}.  {\bf(relation between $K_1$ and $K_0$, see \eqref{Um}? what happened to $\curl^*$?)}}
		We solve for $\xi'$ and write the solution as  
	\begin{equation} \label{xi'} 
	\xi' = -\frac{g}{\sqrt{2n}} \langle |\chi|^2\rangle\eta^{-1},
	\end{equation}
where  $\eta\equiv \eta_{m_z,m_h}(\tau)$ is defined in \eqref{eta'}-\eqref{alpha'}.
The operator $G_{m_z,m_h}$ in \eqref{alpha'} is positivity preserving and therefore the function $\alpha_{m_z,m_h}(\tau)$ (and hence $\eta_{m_z,m_h}(\tau)$) is positive, if and only if $m_z < m_h$ (equivalently, $M_Z < M_H$), in which case $\xi' <0$. \end{proof}

 
We now derive the estimate \eqref{s2-asym} - \eqref{s2-asym-rem} for $s^2$. Equations \eqref{s-asym} and \eqref{s-eq} give $\xi_s$ as a function of $s$ and $b$ respectively, yielding
\begin{align}
 \xi_s^2 = [\frac{\sqrt{2n}}{g} +g_\xi(s^2)]^2 = \frac{n}{eb}\varphi_0^2,
\end{align}
which can be rearranged to give
\begin{align}\label{s2-b-rel}
 \frac{2\sqrt{2n}}{g} g_\xi(s^2) + g_\xi(s^2)^2 = \frac{2n}{g^2}\om,
\end{align}
where, recall, $\om=1-\frac{M_W^2}{eb}$, with $M_W = \frac{1}{\sqrt 2} g\vphi_0$. Recall that $g_\xi(0) = 0$ and $g_\xi'(0) = \xi'$. We have
\begin{align}\label{s2-b-rel-deriv}
\frac{d}{ds^2}|_{s^2=0}[\frac{2\sqrt{2n}}{g} g_\xi(s^2) + g_\xi(s^2)^2] = \frac{2\sqrt{2n}}{g}\xi'.
\end{align}
 Since $\xi'\neq 0$ and $g_{\xi}(s^2)$ is continuously differentiable of all orders (see Theorem \ref{bifurcation-result}), by the Implicit Function Theorem, we may solve \eqref{s2-b-rel} for $s^2$, with the solution, $s^2 = s^2(\om)$, 
 with continuous derivatives 
of all orders in $\om$. Explicitly, \eqref{s2-b-rel} - \eqref{s2-b-rel-deriv} give 
	\begin{align}\label{s2-asym'}
	s^2  = \frac {g}{2\sqrt{2n}}&\xi'^{-1}\frac{2n}{g^2}\om+ \cO(|\om|^2). 
	\end{align}
Plugging \eqref{xi'} into \eqref{s2-asym'} gives
	\begin{align}
	&s^2=\frac{n}{g^2}\frac{\om}{ 
	\langle|\chi|^2\rangle}\eta  + 
	 R_s(\om), \label{s2-soln} 
	\end{align}
which is \eqref{s2-asym}, with   $R_s(\om)$ satisfying $R_s(\om) = \cO(|\om|^2)$. Furthermore, since the solution $s^2 = s^2(\om)$ is continuously differentiable of all orders in $\om$, so is the remainder term $R_s(\om)$.
%
\end{proof}

\section{Asymptotics of the Weinberg-Salam energy near $b=M_W^2/e$} \label{sec:asymp}
 Recall $\om=1-\frac{M_W^2}{eb}$, with $M_W = \frac{1}{\sqrt 2} g\vphi_0$, and $\eta_{m_z, m_h}( \tau)$ is defined in \eqref{eta'}. The main result of this section is the following:
	
\begin{theorem} \label{thm:Eb-asym}
	If $M_Z < M_H$, then the WS energy \eqref{WS-energy'} of the branch of solutions \eqref{bif-soln} has the following expansion:
	\begin{align} \label{WS-energy-s-4}
	\frac{1}{|\Omega|}E_{\Omega} (W_b, A_b, Z_b, \varphi_b) 
	= \frac12 b^2 -\frac12 b^2\sin^2\theta\ &\eta_{m_z, m_h}( \tau) \om^2 
	 + R_E(\om),
	\end{align}
	where $R_E(\om)$ is a real function with continuous derivatives of all orders satisfying 
	\begin{align}
	R_E(\om) = \cO(|\om|^3). 
	\end{align}
\end{theorem}

Before proving Theorem \ref{thm:Eb-asym}, we derive from it Theorem \ref{thm:AL-exist'} (c). \begin{proof}[Proof of  Theorem \ref{thm:AL-exist'} (c)] Since $\eta_{m_z, m_h}(\tau)$ is positive,\footnote{See the discussion following Proposition \ref{prop:b-asym} for details.} the second term in Equation \eqref{WS-energy-s-4} is negative, and so for $ 0< 1-\frac{M_W^2}{eb} \ll 1$, $E_{\Omega}^{WS}$ is less than the vacuum energy $\frac{1}{2}b^2|\Omega|$. This proves Theorem \ref{thm:AL-exist'} (c). \end{proof}

\begin{proof}[Proof of  Theorem \ref{thm:Eb-asym}]	
	
Let $\cE' (w_s, a_s, z_s, \psi_s+\xi_{s}; r):= \frac{1}{|\Omega'|}\E_{\Omega'} (w_s, a_s, z_s, \psi_s+\xi_{s}; r)$, where $\E_{\Omega'}$ is  the rescaled WS energy given in \eqref{WS-energy-resc}.	In Appendix \ref{Sec:Lemma-Eos}, we 
	 derive the following expansion (to order $s^4$) of $\cE'$ evaluated at family \eqref{s-asym} of solutions:
	
		\begin{align} \label{WS-energy-s-1}
		\cE' (w_s, a_s, z_s, \psi_s+\xi_{s}; r) &= \frac12\frac{n^2}{e^2} + s^4\lan    \frac12 |\curl z'|^2 + \frac12 |\curl a'|^2\nonumber \\
		 +& \ g\sqrt{2n}(\psi' + \xi') |\chi|^2 + \frac{n}{2\cos^2\theta}|z'|^2  + |\nabla\psi'|^2\nonumber \\
		+&  \frac{4\lambda n}{g^2} \psi'^2 - |\chi|^2\curl \nu'
		+ \frac{g^2}{2}|\chi|^4\ran  \nonumber \\
		+& R_\varepsilon(s),
		\end{align}
		where $R_\varepsilon(s) = \cO(|s|^6)$ and has continuous derivatives of all orders, 
		 $\nu' := g(a'\sin\theta + z'\cos\theta)$ and, recall,  $\xi_{ s}= \sqrt{2\mu_s}/g$. 
	
	To simplify notation, in what follows, we shall suppress the argument $(w_s, a_s, z_s, \psi_s+\xi_{ s}; r)$ of $\cE'$.	
	We claim the following relation:
		\begin{align} \label{WS-energy-s-3}
		\cE' &= \frac12 \frac{n^2}{e^2} - s^4 \frac{g^2}{2}  \langle |\chi|^2\rangle^2\eta^{-1} 
		 + R_\varepsilon(s). 
		\end{align}
\DETAILS{where, recall, $\eta=\eta_{m_z,m_h}(\tau, r)$ 
		 and $\alpha_{m_z,m_h}(\tau,r)$ are given in \eqref{eta'} and \eqref{alpha'}.} 
\begin{proof}[Proof of \eqref{WS-energy-s-3}]
	We simplify the integral at order $s^4$ in \eqref{WS-energy-s-1} by applying equations \eqref{os4-eqns} for $a'$, $z'$ and $\psi'$ to convenient groupings of terms. 
	
	First, we address the $z'$ terms in \eqref{WS-energy-s-1}. Integrating by parts and factoring out $z'$ gives
	\begin{align}
	 \frac12 \lan  |\curl z'|^2 + \frac{n}{\cos^2\theta}|z'|^2\ran   &= \frac12 \lan  z' \cdot (-\Delta + \frac{n}{\cos^2\theta}) z'\ran .
	\end{align}
	Applying \eqref{os4-eqns} for $z'$ gives
	\begin{align}
	 \frac12 \lan  |\curl z'|^2 + \frac{n}{\cos^2\theta}|z'|^2\ran   &= \frac12 \lan   z'\cdot g\cos\theta\curl^*|\chi|^2\ran  .
	\end{align}
	Integrating by parts again gives
	\begin{align}
	 \frac12\lan   |\curl z'|^2 + \frac{n}{\cos^2\theta}|z'|^2\ran   &=\frac12 \lan  g\cos\theta(\curl z') |\chi|^2\ran. \label{int-s3}
	\end{align}
	
	Next, we address the $a'$ term in \eqref{WS-energy-s-1}. Integrating by parts gives
	\begin{align}
	\lan \frac12 |\curl a'|^2\ran  &=\lan \frac12 a' \cdot (-\Delta)a'\ran .
	\end{align}
	Inserting into this expression \eqref{os4-eqns} for $a'$ gives
	\begin{align}
	\lan \frac12 |\curl a'|^2\ran  &= \lan  \frac12a' \cdot e\curl^* |\chi|^2\ran.
	\end{align}
	Integrating by parts again gives
	\begin{align}
	\lan \frac12 |\curl a'|^2\ran  &=\lan \frac12 g \sin\theta(\curl a') |\chi|^2\ran. \label{int-s4}
	\end{align}
	
	Next, we address the $\psi'$ terms. Integrating by parts and factoring out $\psi'$ gives
	\begin{align}
	\lan  |\nabla\psi'|^2 + \frac{4\lambda n}{g^2} &\psi'^2 + g\sqrt{2n}\psi' |\chi|^2\ran  \notag \\ &= \lan \psi' (-\Delta + \frac{4\lambda n}{g^2} + g\sqrt{2n}|\chi|^2) \psi'\ran .
	\end{align}
	Inserting into this expression  \eqref{os4-eqns} for $\psi'$ gives
	\begin{align}
	\lan  |\nabla\psi'|^2 + \frac{4\lambda n}{g^2} &\psi'^2 + g\sqrt{2n}\psi' |\chi|^2\ran  = \lan   \frac{g}{2}\sqrt{2n}\psi' |\chi|^2\ran . \label{int-s5}
	\end{align}
	
	For the $\xi'$ term in \eqref{WS-energy-s-1}, we have by \eqref{phi0-eqn'} and \eqref{xi'},
	\begin{align} \label{xi'-term}
		\lan g\sqrt{2n} \xi'  |\chi|^2\ran & =\frac12\lan g\sqrt{2n} \xi'  |\chi|^2\ran  + \frac12\lan g\sqrt{2n} \xi'  |\chi|^2\ran\notag\\
		&  =\frac12\lan -g\sqrt{2n} \psi' |\chi|^2 + \curl\nu'|\chi|^2 - g^2 |\chi|^4\ran\notag\\
		&   - \frac12g^2 \langle |\chi|^2\rangle^2\eta^{-1},
	\end{align}
where, recall, $\nu' := g(a'\sin\theta + z'\cos\theta)$. 	
	Finally, there are two remaining terms of the integral at order $s^4$ in \eqref{WS-energy-s-1},
	\begin{align}
	\lan  - |\chi|^2\curl\nu' + 
	\frac{g^2}{2}|\chi|^4\ran , \label{int-s}	
	\end{align}
 which we will not presently simplify.
 	
	Adding equations 
	\eqref{int-s3}, \eqref{int-s4}, \eqref{int-s5}, \eqref{xi'-term} and \eqref{int-s} and remembering  \eqref{WS-energy-s-1} gives 
\DETAILS{	\begin{align} \label{int-soln'}
	\cE' = \frac12 \frac{n^2}{e^2} + s^4 \lan \frac 12 g\sqrt{2n} \psi' |\chi|^2 - \frac12 \curl\nu'|\chi|^2 + \frac 12 g^2 |\chi|^4\ran ,
	\end{align}
\DETAILS{	Equations 
	\eqref{s-soln} and \eqref{U-prop} imply that
	\begin{align}
	\curl \nu' = g^2|\chi|^2 - e g\sin\theta\langle|\chi|^2\rangle - g^2 m_z^2\cos^2\theta\ U_{m_z}(|\chi|^2). \label{v3-simp}
	\end{align}}
	where, recall $e = g\sin\theta$. 
	 Plugging Equation \eqref{v3-simp'} for $\curl\nu'$ and \eqref{s-soln} for $\psi'$ into \eqref{int-soln'} gives} 
	Eq. \eqref{WS-energy-s-3}, as required.
\end{proof}

 	Plugging \eqref{s2-soln} into \eqref{WS-energy-s-3} gives
	\begin{align} \label{WS-energy-s-3'}
	\cE' = \frac12 \frac{n^2}{e^2} -& \frac12 \frac{n^2}{g^2} \om^2\eta_{m_z,m_h}(\tau, r)
	+ \tilde R_\varepsilon(\om),
	\end{align}
	where $\tilde R_\varepsilon(\om)$ has continuous derivatives of all orders and satisfies $\tilde  R_\varepsilon(\om) = \cO(|\om|^3)$.
	
 For the WS energy \eqref{WS-energy'}, evaluated at $(W_b, A_b, Z_b, \varphi_b)$, 
we recall that $E_{\Omega}^{WS} = \frac1{r^2} \E_{\Omega'}= \frac{1}{r^2}|\Omega'| \cE'$, which implies 
	\begin{align} \label{E-resc}
	\frac{1}{|\Omega|}E_{\Omega} = \frac{|\Omega'|}{r^2|\Omega|} \cE' = \frac{e^2 b^2}{n^2} \cE',  \quad r=\sqrt{\frac{|\Omega|}{|\Omega'|}}=\sqrt{\frac{n}{eb}}.
	\end{align}
	Eq. \eqref{WS-energy-s-4} follows by plugging \eqref{WS-energy-s-3'} 
	 into \eqref{E-resc}. Since the remainder term 
	 $\tilde R_\varepsilon$ of 
	   \eqref{WS-energy-s-3'} has continuous derivatives of all orders, so does the remainder term $R_E$ of \eqref{WS-energy-s-4}. \end{proof}

\DETAILS{
	where $\beta$ is the \emph{Abrikosov function} \cite{Abr} (``constant" in the physics literature)
	\begin{align} \label{beta-def}
	\beta(\tau) :=  \langle|X_r|^4\rangle/\langle|X_r|^2\rangle^2.
	\end{align}
	We note that this agrees with the results of \cite{MT} except for the extra $\beta(\tau)$ term.}


\section{Shape of lattice solutions}\label{sec:shape}
In this section we shall prove Theorem \ref{thm:lattice-shape}. Recall the shape parameter $\tau$ described in the paragraph preceding \eqref{fund-domSL2Z}. 
We return briefly to working with the rescaled fields to prove that $\E_{\Omega'}(u; r),\ u=(w, \al, z, \psi),$ 
given in \eqref{WS-energy-resc}, (and hence $E_{\Om} (U)$) is continuously G\^ateau differentiable of all orders in the shape parameter $\tau$ (restricted to domain \eqref{fund-domSL2Z}), which enters through $\Omega'$ (and $\Om$), as well as the spaces containing $u$  (and $U$). Below, we write 
\begin{align}
& u_{\tau,b}(x) \equiv (w_{\tau,b}(x), a_{\tau,b}(x), z_{\tau,b}(x), \phi_{\tau,b}(x)), \\
&\label{WS-energy-resc'} \E(\tau, b,u) \equiv  \E_{\Omega'}(u; r)=:\int_{\Omega'} e (u; r),\\
& U_{\tau,b}(x) \equiv (W_{\tau,b}(x), A_{\tau,b}(x), Z_{\tau,b}(x), \varphi_{\tau,b}(x)), \\
& E(\tau,b,U)\equiv   E_{\Om} (U), \\
& \cX_{\tau} \equiv \cX,
\end{align}
 to emphasize the dependence of the family of solutions \eqref{s-asym}, the corresponding energy $\eqref{WS-energy-resc}$ (respectively 
\eqref{WS-energy'}) and the space \eqref{space} containing these solutions
 on the shape parameter $\tau$, the magnetic field strength $b$ and the position in space $x\in\mathbb{R}^2$. 
Also, recall the notation $r:=\sqrt{n/eb}$.

To get rid of the dependency of the space $\cX_{\tau}$ containing $u_{\tau,b}$, on the shape parameter $\tau$, we make the change of coordinates 
\begin{align}
M_{\tau}:\cX_{\tau}&\to\cX_1, 
(M_{\tau}u)(x) = u(m_{\tau}x),\ m_{\tau}=\frac1{\sqrt{\Im(\tau)}}\left(\begin{array}{cc}1 & \text{Re}(\tau) \\ 0 & \text{Im}(\tau) \end{array} \right),\end{align}
mapping $\Omega'$ into a square of area $2\pi$. This allows us to define the  functions $G':\C\times\R\times\cX_1\to \C\times\R\times\cY_1$ and $\varepsilon' : \C\times\R\times\cX_1\to \C\times\R\times\cY_1$ on the fixed space $\cX_1$:
\begin{align}
G'(\tau,b,v) &= M_{\tau}G(b, M_{\tau}^{-1}v),\\
\varepsilon'(\tau,b,v) &= M_{\tau}\varepsilon (b,M_{\tau}^{-1}v),
\end{align}
where, recall, $G(b, v)$ is the map given by the left-hand side of \eqref{WS-eq1-resc} - \eqref{WS-eq4-resc}, given explicitly in \eqref{G-expl}, 
and $\varepsilon_{WS} (b,u):=  \varepsilon (u; r)$ is the rescaled energy density given by the integrand in \eqref{WS-energy-resc}, 
see  \eqref{WS-energy-resc'} ($\varepsilon $ depends on the magnetic field strength $b$ but does not directly depend on the shape parameter $\tau$).

\begin{lemma} \label{G-lem}
	$G'(\tau,b,v)$ and $\varepsilon'(\tau,b,v)$ are continuously G\^ateau differentiable of all orders in $\Re(\tau)$, $\Im(\tau)$, $b$ and $v$. 
\end{lemma}

\begin{proof}
	Since $G(b,u)$ and $\varepsilon (b,u)$ have continuous $b$ and $u$ derivatives of all orders, and $M_{\tau}$ is a linear map independent of $b$ and $v$, it follows that $G'(\tau,b,v)$ and $\varepsilon'(\tau,b,v)$ have continuous $b$- and $v$-derivatives of all orders.
	
	For the $\tau$-derivatives, note that 
	\begin{align}
&	M_{\tau}\circ\partial_{x_1}\circ M_{\tau}^{-1}(v_j)(x) = \frac1{\sqrt{\Im(\tau)}}\partial_{x_1}v_j(x),\quad j=1,...,4, \label{M-tau-partial-1}\\
&	M_{\tau}\circ\partial_{x_2}\circ M_{\tau}^{-1}(v_j)(x) \nonumber \\ 
	&\qquad = \frac1{\sqrt{\Im(\tau)}}(\Re(\tau)\partial_{x_1}v_j(x) + \Im(\tau)\partial_{x_2}v_j(x)), \quad j=1,...,4, \label{M-tau-partial-2}
	\end{align}
	are continuously differentiable of all orders in $\Re(\tau)$ and $\Im(\tau)$. Since $G(b,u)$ and $\varepsilon_{WS}(b,u)$ are polynomials in the components of $u$ and their (covariant) derivatives, $G'$ and $\varSigma$ are simply $G$ and $\varepsilon_{WS}$ with the coefficients of the derivative-containing terms multiplied by smooth functions of $\Re(\tau)$ and $\Im(\tau)$. Therefore $G'(\tau,b,v)$ and $\varSigma(\tau,b,v)$ have continuous $\Re(\tau)$- and $\Im(\tau)$-derivatives of all orders.
\end{proof}

\begin{lemma} \label{v-tau-lem}
	$v_{\tau,b} :=  M_{\tau}u_{\tau,b}$ is continuously differentiable of all orders in $\text{Re}(\tau)$ and $\text{Im}(\tau)$.
\end{lemma}

\begin{proof}
Let $\tau_0$ be an arbitrary shape parameter, and recall that $\del_\#$ denotes the partial (real) G\^ateaux derivative with respect to $\#$. Then $G'(\tau_0,b,v_{\tau_0,b}) = \\ M_{\tau_0}G(b,u_{\tau_0,b}) =0$, $\del_v G(\tau_0,b,v_{\tau_0,b}) = M_{\tau_0} \circ \del_u G(b,u_{\tau_0,b}) \circ M_{\tau_0}^{-1}$ is invertible, and by Lemma \ref{G-lem}, $G'$ is continuously G\^ateau differentiable of all orders in $\tau$, $b$ and $v$. Therefore, by the Implicit Function Theorem, the unique solution $v_{\tau,b}$ to the equation $G(\tau,b,v) = 0$ is continuously differentiable of all orders in $\text{Re}(\tau)$ and $\text{Im}(\tau)$ near $(\text{Re}(\tau), \text{Im}(\tau)) = (\text{Re}(\tau_0), \text{Im}(\tau_0))$. Since $\tau_0$ was arbitrary, this proves the result.
\end{proof}

\DETAILS{
	Recall that $\varepsilon_{WS}(b,u)$ is the rescaled energy density \eqref{e-dns-resc} (where we emphasize the dependence on the magnetic field strength $b$, and note that it is independent of the shape parameter $\tau$).
	
	\begin{lemma} \label{ep-lem}
		$\partial_{\tau}^k[\varepsilon_{WS}(\tau,b,v_{\tau,b})]\in\cY_{\tau}$ for all $k\in\Z_{\geq 0}$.
	\end{lemma}
}

\begin{proposition}
	$E (\tau,b,U_{\tau,b})$ is continuously differentiable of all orders in $\text{Re}(\tau)$ and $\text{Im}(\tau)$.
\end{proposition}

\begin{proof}
	To get rid of the dependency of $\E (\tau,b,u_{\tau,b})$ on the domain of integration $\Omega'$, we again make the change of coordinates $y = m_{\tau}^{-1}x$. \DETAILS{with
		$m_{\tau}=\sqrt{\frac{2\pi}{\text{Im}(\tau)}}\left(\begin{array}{cc}
		1 & \text{Re}(\tau) \\ 0 & \text{Im}(\tau) \end{array} \right)$
		to turn the domain of integration $\Omega'$ into the unit square.} Then
	\begin{align}
	\E(\tau,b,u_{\tau,b}) = &\int_{\Omega'} \varepsilon (b,u_{\tau,b})(x)\; d^2x \nonumber\\
	= &\int_0^{\sqrt{2\pi}}\int_0^{\sqrt{2\pi}} \varepsilon'(\tau,b,v_{\tau,b})(y) d^2y. \label{EWS-integrand}
	\end{align}
	By Lemma \ref{v-tau-lem}, $v_{\tau,b}$ has continuous $\Re(\tau)$- and $\Im(\tau)$-derivatives of all orders, and by Lemma \ref{G-lem}, $\varSigma$ has continuous derivatives of all orders mapping $\C\times\R\times\cX_1$ to $\C\times\R\times\cY_1$. In particular, the $\Re(\tau)$- and $\Im(\tau)$-derivatives of 
	$\varepsilon'(\tau,b,v_{\tau,b})$ remain integrable, so we conclude that $\E(\tau,b,u_{\tau,b})$ (and hence $E (\tau,b,U_{\tau,b})$) is continuously differentiable of all orders in $\text{Re}(\tau)$ and $\text{Im}(\tau)$.
\end{proof}

\DETAILS{
	Define 
	\begin{align} \label{delta-def1}
	\delta(M_Z,M_H,\tau) := & \frac{M_W^2\alpha(M_Z,M_H;\tau)+\sin^2\theta+\beta(\tau)}{[M_W^2\alpha(M_Z,M_H;\tau)+\sin^2\theta]^2} 
	.
	\end{align}
	(This is the coefficient of $(1-M_W^2/eb)^2$ in \eqref{WS-energy-s-4} up to a factor of $-\frac12 b^2\sin^2\theta$.)
}

\begin{theorem} \label{thm-tau*}
	When $M_Z < M_H$, the minimizers $\tau_b$ of $E (\tau, b, U_{\tau,b})$ are related to the maximizers $\tau_*$ of $\eta_{m_z, m_h}(\tau)$ as $\tau_b-\tau_* = \cO(|1-\frac{M_W^2}{eb}|^{\frac12})$. In particular, $\tau_b \to \tau_*$ as $b \to b_* = M_W^2/e$.
\end{theorem}

\begin{proof}
	The minimizers ($\tau_b$) of $E (\tau,b,U_{\tau,b})$ are equivalent to the minimizers of the energy functional $\tilde{E} (\tau,U_{\tau,b}) :=  \om^{-2}(E (\tau,b,U_{\tau,b})-\frac12 b^2)$. By Theorem \ref{thm:Eb-asym}, we have
	\begin{align}
\tilde{E} (\tau,U_{\tau,b}) 
 \nonumber= - \frac12 b^2 \sin^2\theta\ \eta_{m_z, m_h}(\tau) + \cO(|\om|),
		\end{align} 
where, recall, $\om=1-\frac{M_W^2}{eb}$.		 Since $\partial_{\tau}\tilde{E} (\tau,U_{\tau,b})|_{\tau=\tau_b}=0$,	 we have the expansion 	
	\begin{align}
\tilde{E} (\tau_*,U_{\tau_*,b}) - \tilde{E} (\tau_b,U_{\tau_b,b})&= \frac12 \partial_{\tau}^2\tilde{E} (\tau_b,U_{\tau_b, b}) 
 [\tau_*-\tau_b]^2 + \cO([\tau_*-\tau_b]^3)\nonumber\\
 \nonumber	= &  - \frac14 b^2 \sin^2\theta\ \partial^2_{\tau}\eta_{m_z, m_h}(\tau_b)  [\tau_*-\tau_b]^2\nonumber\\& + \cO([\tau_*-\tau_b]^3) + \cO(|\om|).
		\end{align} 
		For both expansions to hold, we must have $\tau_b-\tau_* = \cO(|\om|^{\frac12})$, as required. 
\end{proof}

The maximizer of $\eta_{m_z, m_h}(\tau)$, defined in \eqref{eta'}, was found numerically in \cite{MT} with some analytical results in \cite{MacDow}:

\begin{theorem}[\cite{MT}] \label{thm-eta-max}
	For $M_Z < M_H$, the function $\eta_{m_z, m_h}(\tau)$ has a maximum at $\tau_* =e^{i\pi/3}$.
\end{theorem}

 Theorem \ref{thm:lattice-shape} follows from Theorems \ref{thm-tau*} and \ref{thm-eta-max}.  
 
\begin{remark} Using symmetries of $\eta_{m_z, m_h}(\tau)$, one might be able to prove that its only critical points 
  are $\tau =e^{i\pi/3}$ and $\tau =e^{i\pi/2}$, i.e. the hexagonal and square lattices, cf. \cite{NV, ABN}.
\end{remark}

\appendix

\section{Covariant derivatives and curvature}\label{sec:Cov-Deriv-Curv} 

In this appendix, we briefly review some basic definitions from gauge theory. 
Recall that we use the Einstein convention of {\it summing over repeated indices}.

\medskip


Let $V$ be an inner product vector space,  $G$  a Lie group acting transitively  on $V$ via a unitary representation $\rho: g \mapsto \rho_g$,  and let $\fg$ be the Lie algebra of $G$ acting on $V$ via the representation $\tilde\rho: A \mapsto \tilde\rho_A$ induced by $\rho$.

To simplify notation below, we take  $V=\C^m$ and  $G$  a matrix group,  acting  on $V$ by matrix rules (and similarly for $\fg$) and write $\rho_g\Psi = g \Psi$ and $\tilde\rho_A\Psi = A \Psi$. Moreover, we assume that  $G$ is either  $U(m)$ or a Lie subgroup of $U(m)$. 

 Let $M$ be an open subset in a finite-dimensional vector space, with a metric $h$ and local coordinates $\{x^i\}$, and let $\p_i \equiv \p_{x^i}$.

For a ${\fg}$-valued connection (one-form) 
 $A \equiv A_i dx^i$ on $M$, we define the covariant derivatives:

-  $\n_{A}$, mapping functions (sections), $\Psi: M\ra V$, into $\fg$-valued one-forms,   as \begin{equation}\n_{A}\Psi :=d\Psi + A\Psi \equiv (\p_i\Psi + A_i\Psi)dx^i;\end{equation} 

-  $d_{A}$, mapping  $\fg$-valued functions ($0$-forms) $f$ into $\fg$-valued one-forms 
\begin{equation}d_A f:=  d f +  [A, f] \equiv (\p_i f +  [A_i,f])dx^i;\end{equation}

-  $d_{A}$, mapping $\fg$-valued one-forms into $\fg$-valued two-forms 
\begin{equation}d_A B :=  d B +  [A,B],\end{equation} with $[A, B]$ defined in local coordinates $\{x^i\}$ as
\begin{align}\label{commut-def'} [A, B]:=  [A_i,  B_j] dx^i\wedge dx^j \equiv [B,A],\end{align}
for $A= A_i dx^i$ and $B=B_j dx^j$.\footnote{More generally, if $A$ is a $\fg-$valued $p-$form and $B$ is a $\fg-$valued $q-$form, written as  
	$A= A^a \otimes \g_a$ and $B= B^b \otimes \g_b$, where $A^a$ and $B^b$ are  $p-$ and $q-$forms and $\{\g_a\}$ is a 
	basis in $\fg$, then \begin{equation}[A, B]:= 
	(A^a\wedge  B^b)  \otimes [\g_a, \g_b] =(-1)^{pq+1}[B, A].\end{equation}} 

The curvature form of the connection $A$ is  the $\fg$-valued two-form given by the formula
\begin{align}\label{FA-expr} F_A=d A + \frac12 [A, A].\end{align}  
 It is related to the curvature operator (denoted by the same symbol) $F_{A} :=  d_A \circ d_A$. As a simple computation shows, this operator is a matrix-multiplication operator given by the matrix-valued 2-form \eqref{FA-expr}. 

Let $U$ be a vector space ($V$ or $\fg$ in our case) and let $\Omega^p_U\equiv  U\otimes\Omega^p$ denote the space of $U$-valued $p$-forms. On $\Omega^p_U$, one defined the inner product, $\lan \cdot,  \cdot \ran_{\Omega^p_U}\equiv \lan \cdot,  \cdot \ran_{\Omega^p_U}^h$ as 
\begin{equation} \label{inner-prod-ext}
\langle A,B \rangle_{\Omega^p_U} \equiv \langle A,B \rangle_{\Omega^p_U}^h:=  \langle A_{\alpha}, B^{\alpha} \rangle_U, 
\end{equation}
where $A=A_{\alpha} dx^{\alpha}$ and $B=B_{\alpha} dx^{\alpha}$ are $U$
-valued $p$-forms, $\alpha$ is a $p$-form index and $\langle \cdot, \cdot \rangle_U$ 
is the inner product on $U$. 
Here the indices are raised and lowered with help of the inner product $h$ on $M$.


 Above, we 
  did not display the coupling constants. Doing so would change the covariant derivative to  $d_{A} \Psi=(d + g A) \Psi$, if $G$ is simple. If  $G$ is not simple, then each simple 
   component of $G$ gets its own coupling constant, as was done in the main text for $G=SU(2)\times U(1)$ (see also \eqref{nabQ-spec}-\eqref{Xij} below).
\DETAILS{This results in a slight discrepancy in the definition of the covariant derivative 
for the gauge group $G=U(2)=SU(2)\times U(1)$.}

\section{The time-dependent Yang-Mills-Higgs system}\label{sec:YM}

In this appendix, we briefly review the 
Yang-Mills-Higgs theory, including the derivation of the energy functional \eqref{WS-energy0}. In what follows, 
we use the convention of raising or lowering an index by contracting a tensor $T$ with the metric tensor:
\begin{equation}
T^{\alpha}_{i,\beta} = \eta_{ij} T^{j,\alpha}_{\beta}
\end{equation}
where $\eta$ is the Minkowski metric of signature $(+,-,...,-)$ on $M\subset\R^{d+1}$ and $\alpha, \beta$ are multi-indices.
The same equations could be {\it reinterpreted as stationary} equations by taking the {\it Euclidean metric $\del_{i j}$, instead of $ \eta_{ij}$,  and letting the indices range over $1, \dots, d$, rather than $1, \dots, d+1$.}
In this case, $T^{\alpha}_{i,\beta} =  T^{i,\alpha}_{\beta}$.

\medskip

\paragraph{\bf Lagrangian.}

Let $\Omega$ be a bounded domain in $\R^{d}$ and $M=\Omega\times [0, T]\subset\R^{d+1}$ be spacetime equipped with the Minkowski metric $\eta$ of signature $(+,-,...,-)$ and $V$ and $G$ be as in Appendix \ref{sec:Cov-Deriv-Curv}.
The theory involves a Higgs field $\Psi: M\ra V$ interacting with the gauge field $A$, a connection (one-form) on $M$ with values in the algebra ${\fg}$. 
The dynamics are given by the Lagrangian 
\begin{equation} \label{YMH-Lagr}
\cL (\Psi, A) :=  \int_{\Omega}  \big(  \lan \n_{A}\Psi,  \n_{A}\Psi \ran_{\Omega^1_V}^\eta - 
U(\Psi) + \lan F_A,  F_A \ran_{\Omega^2_\fg}^\eta\big), 
\end{equation}
with corresponding action $\cS := \int_0^T \cL(\Psi, A) dt$, $T>0$, given explicitly by
\begin{equation} \label{YMH-act}
\cS (\Psi, A) =  \int_{M} \big( \lan \n_{A}\Psi,  \n_{A}\Psi \ran_{\Omega^1_V}^\eta - 
U(\Psi) + \lan F_A,  F_A \ran_{\Omega^2_\fg}^\eta\big), 
\end{equation}
where  
$U: V\ra \R^+$ is a self-interaction potential, which is assumed to be gauge invariant: $ U(\rho_g \Psi)= U(\Psi)$. Typical examples of $G, V$ and $U(\Psi)$ are $U(m), \C^m$ and $U(\Psi)=\frac{1}{2} \lam (1 - \|\Psi \|^2_V)^2$. 
 \DETAILS{Furthermore,  $\n_{A} \Psi=(d - \tilde\rho_A) \Psi$ and  $F_A:= d A - \frac12 [A, A]$ 
with $\tilde\rho_A$ and $[A, B]$ defined in local coordinates $\{x^i\}$ as 
\begin{equation}\label{1-form-action} \tilde\rho_A \Psi :=  \tilde\rho_{A_i}\Psi \ dx^i\end{equation}
and
\begin{equation}\label{commut-def'} [A, B]:=  [A_i,  B_j] dx^i\wedge dx^j,\end{equation}
for $\Psi:\R^{d+1}\to V$, $A= A_i dx^i$ and $B=B_j dx^j$.\footnote{More generally, if $A$ and $B$ are  $\fg-$valued $p-$forms, written locally as  
	$A= A^a \otimes \g_a$ and $B= B^b \otimes \g_b$, where $A^a$ and $B^b$ are $p-$forms and $\{\g_a\}$ is a basis in $\fg$, then $[A, B]:= \sum_{a b} (A^a\wedge  B^b)  \otimes [\g_a, \g_b]$.}  
Finally,}

\medskip

\paragraph{\bf Euler-Lagrange equations.}The Euler-Lagrange equations (called Yang-Mills-Higgs equations) for the fields $\Psi$ and $A$ are 
\begin{align} \label{YMH-eqs-psi}
&\n_{A}^{*_\eta} \n_{A}\Psi  =
U'(\Psi)      , \\ \label{YMH-eqs-A}
&d_{A}^{*_\eta} F_A = J(\Psi, A), 
\end{align}
where 
$\n_{A}^{*_\eta}$ and $d_{A}^{*_\eta}$ are the adjoints of $\n_{A}$ and $d_{A}$ in the appropriate inner products {\it involving the metric $\eta$} and $J(\Psi, A)$ is the YMH current given by
\begin{align} \label{YMH-eqs-curr}
J(\Psi,A) &:= \Re\lan {\g_a} \Psi, \nabla_A\Psi \ran_V \g_a\\
&= \Re\lan {\g_a} \Psi, \nabla_i\Psi \ran_V \g_a\otimes dx^i,
\end{align}
where $\g_a$ is an orthonormal basis of $\mathfrak g$ and $\nabla_i :=  \partial_i + {A_i}$, with $\p_i\equiv \p_{x^i}$, so that  $\n_{A} \Psi =\n_{i} \Psi dx^i $. \eqref{YMH-eqs-A} is the Yang-Mills equation.

\begin{proof}[Proof of \eqref{YMH-eqs-psi} - \eqref{YMH-eqs-A}]
For convenience, we {\it assume periodic or Dirichlet boundary conditions and that  $\Psi$ and $A$ are $T$-periodic in $t$} and  calculate the G\^ateaux derivatives {\it formally}.  

Recall that $\del_\#$ denotes the partial (real) G\^ateaux derivative with respect to $\#$. First we calculate the (complex) G\^ateaux derivative of \eqref{YMH-act} in the $\Psi$-direction. Define $\partial_z \equiv   \frac12 (\partial_{\Re z} - i\partial_{\Im z})$ and $\del_{\Psi} \equiv  \frac12 (\del_{\Re\Psi} - i \del_{\Im\Psi})$. Then $\del_{\Psi}\cS(\Psi,A)\Psi' = \p_z\cS(\Psi_z,A)|_{z=0}$, where $\Psi_z = \Psi + z\Psi'$, $z\in\C$. Using this,  we find
	\begin{align}
	\del_{\Psi}\cS(\Psi,A)\Psi'  
	&= \int_{M}  \big(\lan \n_{A}\Psi, \n_{A}\Psi' \ran_{\Omega^1_V} - \lan U'(\Psi), \Psi' \ran_V\big).
	\end{align}
	Integrating the first term by parts and factoring out $\Psi'$ gives
	\begin{align}
	\del_{\Psi}\cS(\Psi,A)\Psi' &= \int_{M}  \lan \n^*_A \n_{A}\Psi - U'(\Psi), \Psi' \ran_V.
	\end{align}
	For this derivative to be zero for every variation $\Psi'$, \eqref{YMH-eqs-psi} must hold.
	
	Next we calculate the G\^ateaux derivative of \eqref{YMH-act} in the $A$-direction. Using the definition $\del_A f(A)B = \p_s f(A_s)|_{s=0}$, where $A_s = A + s A'$, $s\in\R$, we find 
	\begin{align}\label{A-var}
	\del_A \cS(\Psi,A)B 
	&= \int_{M} \big(\lan B\Psi, \n_A\Psi \ran_{\Omega^1_V} + c.c. + 2\lan d_A B, F_A \ran_{\Omega^2_\fg}\big)\\
	&= I+II.  \end{align}
Writing  $B = B^a\g_a = B_i^a dx^i \otimes \g_a$ (with $B_i^a$ real) and $\n_{A} \Psi =\n_{i} \Psi dx^i $, so that 
\begin{equation}\lan B\Psi, \n_A\Psi \ran_{\Omega^1_V}=\lan B^a, \lan \g_a\Psi, \n_{i}\Psi \ran_{V} dx^i\ran_{\Omega^1},\end{equation} 
and using that $B^a C^a=-\Tr [(B^c  \g_c) (C^a \g_a)]$ (since $\Tr (\g_c^*\g_a)=-\Tr (\g_c\g_a)=\del_{c a}$), gives 
	\begin{align}
	I = -\int_{M} &\lan B, \lan\g_a \Psi, \nabla_i\Psi \ran_V \g_a\otimes dx^i\ran_{\Omega^1_\fg} + c.c.. 
	\end{align}
	which gives $I= \int_{M} \lan B, J(\Psi,A) \ran_{\Omega^1_\fg}$. For the second term on the r.h.s. of \eqref{A-var}, integrating by parts yields $II= \int_{M} \lan B, d_A^* F_A \ran_{\Omega^1_\fg}$. Collecting the last two equations gives
	\begin{align}
	 \del_A \cS(\Psi,A)B = 2\int_{M} \lan B, - J(\Psi,A) + d_A^* F_A \ran_{\Omega^1_\fg}.
	 \end{align}
	For this derivative to be zero for every variation $B$, \eqref{YMH-eqs-A} must hold.
\end{proof}

\medskip

\paragraph{\bf Conserved energy.} Again, the G\^ateaux derivative calculations in the following subsection are formal. Recall that $M:=\Omega\times [0, T]\subset \R^{d+1}$.

To find the expression for the energy, we use, as in classical mechanics, the (partial, i.e. without passing to the momentum fields) Legendre transform of \eqref{YMH-Lagr} is given by
 \begin{align}
  E(\Psi, A) &= \p_{\n_0\Psi}\cL(\Psi,A)\n_0\Psi + \p_{\overline{\n_0\Psi}}\cL(\Psi,A)\overline{\n_0\Psi} \nonumber \\ \label{Leg-transf} &+ \sum\limits_{i=1}^d \p_{F_{0i}}\cL(\Psi,A) F_{0i} - \cL(\Psi,A).
 \end{align}
\begin{proposition}
The (partial) Legendre transform \eqref{Leg-transf} of Lagrangian \eqref{YMH-Lagr} 
 yields the conserved energy
\begin{align} \label{YMH-energy}
 E(\Psi, A) := \int_{\Omega} \big(\| \n_{A}\Psi \|_{\Omega^1_V}^2 + 
 U(\Psi) + \| F_A \|_{\Omega^2_\fg}^2\big),
\end{align}
where the norms are taken using the Euclidean metric on $\R^{d+1}$ (rather than the Minkowski metric). 
\end{proposition}

Note that for static (time-indepent) fields, $E(\Psi, A) = -\cL(\Psi, A)$.

\begin{proof}
 Let $\p_\#$ denote the partial derivative with respect to the symbol $\#$, and recall that $\del_\#$ denotes the partial (real) G\^ateaux derivative with respect to $\#$.  We calculate
 \begin{align}
 \p_{\n_0\Psi}\cL(\Psi,A)\n_0\Psi = \int_{\Omega} \|\n_0\Psi\|_V^2 = \p_{\overline{\n_0\Psi}}\cL(\Psi,A)\overline{\n_0\Psi}
 \end{align}
 and
 \begin{align}
  \sum\limits_{i=1}^d \p_{F_{0i}}\cL(\Psi,A) F_{0i} = \int_{\Omega} 2\sum\limits_{i=1}^d |F_{0i}|^2.
 \end{align}
 \eqref{YMH-energy} results.
 
 It remains to show that \eqref{YMH-energy} is conserved by the YMH equations \eqref{YMH-eqs-psi} - \eqref{YMH-eqs-A}. This can be done by using  the (partial) Legendre transform \eqref{Leg-transf} as in classical mechanics, or by a direct computation. We proceed in the second way. Applying the chain rule gives
 \begin{align}
  \frac{d}{dt} E(\Psi,A) = \del_\Psi E(\Psi,A)\p_0\Psi + \del_{\overline\Psi} E(\Psi,A)\p_0\overline\Psi + \del_A E(\Psi,A)\p_0 A, \label{E-t-deriv}
 \end{align}
 where, recall,  $\p_i\equiv \p_{x^i}$. We now calculate the first term using \eqref{YMH-eqs-psi}.
  \begin{align}
  \del_\Psi E(\Psi,A)\p_0\Psi = \int_{\Omega} \big(&\lan\n_0\Psi,\n_0\p_0\Psi\ran_V + \sum_{k=1}^d \lan\n_k\Psi,\n_k\p_0\Psi\ran_V \nonumber \\ + &\lan U'(\Psi), \p_0\Psi\ran_V\big).
  \end{align}
  Integrating the second term by parts gives
  \begin{align} \label{E-psi-2}
  \del_\Psi E(\Psi,A)\p_0\Psi = \int_{\Omega} \big(&\lan\n_0\Psi,\n_0\p_0\Psi\ran_V + \sum_{k=1}^d \lan\n_k^*\n_k\Psi,\p_0\Psi\ran_V \nonumber \\ + &\lan U'(\Psi), \p_0\Psi\ran_V\big).
  \end{align}
  By \eqref{YMH-eqs-psi}, we have
  \begin{align}
   \n_0^*\n_0\Psi - \sum_{k=1}^d \n_k^*\n_k\Psi = U'(\Psi),
  \end{align}
  so \eqref{E-psi-2} becomes
  \begin{align}
   \del_\Psi E(\Psi,A)\p_0\Psi = \int_{\Omega} \big(\lan\n_0\Psi,\n_0\p_0\Psi\ran_V + \lan\n_0^*\n_0\Psi,\p_0\Psi\ran_V\big).
  \end{align}
  Here $\n_0^* = -\p_0 + A_0^\dagger = -\p_0 - A_0$, where the second equality follows because the representation of $\fg$ is unitary. Therefore,
  \begin{align}
   \del_\Psi E(\Psi,A)\p_0\Psi &= \int_{\Omega} \big(\lan(\p_0 + A_0)\Psi,(\p_0 + A_0)\p_0\Psi\ran_V \nonumber \\ 
   & \quad + \lan(-\p_0 - A_0)(\p_0 + A_0)\Psi,\p_0\Psi\ran_V\big) \nonumber \\ 
   &= \int_{\Omega} \p_0\lan\Psi, A_0\p_0\Psi\ran_V.
  \end{align}
  Similarly,
  \begin{align}
   \del_{\overline\Psi} E(\Psi,A)\p_0\overline\Psi = \int_{\Omega} \p_0\lan A_0\p_0\Psi, \Psi\ran_V,
  \end{align}
  and so
  \begin{align}
   \del_\Psi E(\Psi,A)\p_0\Psi + \del_{\overline\Psi} E(\Psi,A)\p_0\overline\Psi = \int_{\Omega} \p_0 J_0(\Psi,A),
  \end{align}
  where $J_0(\Psi,A)$ is the time component of the YMH current \eqref{YMH-eqs-curr}.
  
  One may show using \eqref{YMH-eqs-A} that
  \begin{align}
   \del_A E(\Psi,A)\p_0 A = -\int_{\Omega} \p_0 J_0(\Psi,A).
  \end{align}
  Hence, by \eqref{E-t-deriv} we have $\frac{d}{dt} E(\Psi,A) = 0$, as required.
\end{proof}

\medskip

\paragraph{\bf Gauge symmetries.} 

 We define the local action, $ \rho_{g} A$\footnote{Compared with the notation of Appendix \ref{sec:Cov-Deriv-Curv}, to simplify the notation we omit the tilde over $\rho_{g}$ in action of the Lie algebra $\fg$  on $V$. 
 },  of the group $G$ on $ A$,  by the equation $d_{ \rho_{g} A}=g  d_{A}g^{-1}$, for all $g \in C^1(N, G)$, where $N$ is either $M$ or $\Om$. We compute
\begin{align} \label{local-act}
 \rho_{g}   A = g  Ag^{-1}  + gdg^{-1}.
\end{align}

\begin{proposition}
The Lagrangian \eqref{YMH-Lagr} is invariant  under the Poincar\'e group and  the gauge transformations
\begin{align}\label{gauge-transf-YMH}
T^{gauge}_g &:   
(\Psi, A) \mapsto 
( g \Psi, \rho_{g}   A) ,  \qquad 
 \forall g \in C^1(M, G).
\end{align}
\end{proposition}

\begin{proof}
 The invariance under the Poincar\'e group follows from the definition of this group and the choice of the Minkowski metric on $M\subset \R^{d+1}$.
	
		For the gauge invariance, recall that $U(\Psi)$ is $\fg$-invariant, and that the representations $g \mapsto \rho_g$ (on $V$) and the adjoint representation $g\mapsto \text{ad}_g$ (on $\fg$) are unitary. Therefore, to prove invariance under the gauge transformation \eqref{gauge-transf-YMH}, it suffices to show that
	\begin{align}
 \label{action-1} \n_{ \rho_g A}  g \Psi &= g \n_A \Psi,\\	
	F_{ \rho_g A} &= g F_A g^{-1}. \label{action-2}
	\end{align}
	We shall use the equation
	\begin{equation}
	h dh^{-1} = -dh h^{-1}, \quad \forall h\in G
	\end{equation}
	which follows from $d(hh^{-1}) = 0$.	
	For \eqref{action-1} we compute
	\begin{align}
	\n_{ \rho_g A} g\Psi &= d( g\Psi) +  ({gAg^{-1} + gdg^{-1}})( g\Psi) \\
	&= (d g) \Psi + g d\Psi +  g A \Psi +  g  {dg^{-1}g} \Psi. 
	\end{align}
Since $ g  {dg^{-1}g}=- g {g^{-1}dg} = -dg$, this gives $\n_{ \rho_g A} g\Psi=  g \n_A \Psi$.	

	For \eqref{action-2}, computing in coordinates $\{x^i\}$ and writting $F_{\rho_g A} :=  (F_{\rho_g A})_{ij} dx^i\wedge dx^j$ and $F_A := (F_A)_{ij} dx^i\wedge dx^j$, we find
	
	\begin{align}
	(F_{ \rho_g A})_{ij} &= \frac12 [\partial_i(gA_j g^{-1} + g\partial_jg^{-1}) - \partial_j(gA_i g^{-1} + g\partial_i g^{-1})] \nonumber \\ 
	&+\frac12[gA_i g^{-1} + g\partial_i g^{-1}, gA_j g^{-1} + g\partial_jg^{-1}],
	\end{align}
	where, recall,  $\p_i\equiv \p_{x^i}$. 	Expanding the partial derivative and commutators gives
	\begin{align}
	(F_{ \rho_g A})_{ij} =&\frac12[\partial_i g A_j g^{-1} + g\partial_i Ag^{-1} + gA_j\partial_i g^{-1} + \partial_i g\partial_j g^{-1} + g\partial_i\partial_j g^{-1} \nonumber \\
	&+ (gA_i g^{-1} + \partial_i g g^{-1})(gA_j g^{-1} + g\partial_j g^{-1}) \nonumber \\
	&- (i\leftrightarrow j)].
	\end{align}
	Expanding the product on the second line gives
	\begin{align}
	(F_{ \rho_g A})_{ij}=& \frac12[\partial_i g A_j g^{-1} + g\partial_i Ag^{-1} + gA_j\partial_i g^{-1} + \partial_i g\partial_j g^{-1} + g\partial_i\partial_j g^{-1} \nonumber \\
	&+ gA_i A_j g^{-1} + \partial_i g A_j g^{-1} + gA_i\partial_j g^{-1} + \partial_i g\partial_j g^{-1} \nonumber \\
	&- (i\leftrightarrow j) ].
	\end{align}
	Cancelling terms symmetrical in $i$ and $j$ and simplifying gives
	\begin{align}
	(F_{ \rho_g A})_{ij}&= g(\frac12[\partial_i A_j - \partial_j A_i] + \frac12 [A_i A_j - A_j A_i])g^{-1} \\
	&= g (F_A)_{ij} g^{-1},
	\end{align}
	as required.
\end{proof}

Specifying \eqref{YMH-energy} to the WS model gives \eqref{WS-energy0}.

\DETAILS{\noindent the translations
	\begin{equation}\label{translations}
	T^{transl}_s :    (\Psi(x), A(x)) \mapsto (\Psi(x + s), A(x + s)),\qquad \forall t \in \R^2;
	\end{equation}
	\noindent the rotations and reflections,
	\begin{align}\label{rotations-reflections}
	T^{rot}_R :    (\Psi(x), A(x)) \mapsto (\Psi(R^{-1}x),  RA(R^{-1}x)),\qquad \forall R \in  O(n) .
	\end{align}}


\DETAILS{
	\paragraph{\emph{YM solutions.}} For every solution $A_{\rm YM}$ to the vacuum Yang-Mills equations \newline
	$d_{A}^* F_A = 0$ (the constant curvature connection), there corresponds the solution  $(0, A_{\rm YM})$ 
	of the Yang-Mills-Higgs equations. 
}




\paragraph{\bf The YMH equations in coordinate form.} 

\DETAILS{In what follows, we use the Einstein summation convention of summing over repeated indices. Furthermore, we use the convention of raising or lowering an index by contracting a tensor $T$ with the metric tensor:
\begin{equation}
T^{\alpha}_{i,\beta} = \eta_{ij} T^{j,\alpha}_{\beta}
\end{equation}
where $\eta$ is the Minkowski metric of signature $(+,-,...,-)$ on $\R^{d+1}$ and $\alpha, \beta$ are multi-indices. The same equations could be {\it reinterpreted as stationary} equations by taking the {\it Euclidean metric $\del_{i j}$, instead of $ \eta_{ij}$,  and letting the indices range over $1, \dots, d$, rather than $1, \dots, d+1$.}
In this case, $T^{\alpha}_{i,\beta} =  T^{i,\alpha}_{\beta}$.} 

In coordinate form, the differential form (gauge field) entering the YMH Lagrangian \eqref{YMH-Lagr} is written as $A =A_{i} dx^i $. The local coordinate expression for the curvature is $F_A = F_{ij} dx^i \wedge dx^j$, where $F_{i j}:= \frac12(\p_i A_j-  \p_j A_i) +\frac12[A_i,  A_j]$. Furthermore, for the covariant derivatives $\n_A$ and $d_A$, we have $\n_{A} \Psi =\n_{i} \Psi dx^i $  and $d_{A}^* F_A=-\n^{i} F_{i j} dx^j $, where $\n_{i} \Psi  :=  (\p_i +   A_{i})  \Psi$ and $\n^{i} F_{i j} :=  \p^i   F_{i j} + [A^i,  F_{i j}]$.

For an arbitrary $\fg$-valued one-form $B = B_i dx^i$, we have $d_{A} B =\n_{i} B_j dx^i \wedge dx^j$ and $d_{A}^* B =-\n^{i} B_i$, where
\begin{equation} 
\n^{i} B_j :=  \p^i B_j + [A^i,  B_j].
\end{equation}
We write $F_{ij} = F_{ij}^a \g_a$ for an orthonormal basis $\g_a$ of $\fg$ and the lower case  roman indices run over the spatial components $1, 2, \dots, d$.  Note that $F_{i j}=[\n_{i}, \n_{j}]$, but $F_{ij} \neq \frac12 (\n_i A_j - \n_j A_i)$.

Let  $\Omega$ be either a bounded domain in  $\R^d$ or $\R^{d+1}$. In the former case, we assume either periodic or Dirichlet boundary conditions.

\begin{proposition} \label{coord-prop-1}
	The Lagrangian and energy for the YMH model are given in coordinates by 
	\begin{align} \label{YMH-Lagr-coord}
	&\cL(\Psi, A)=   \int_{\Omega}  \lan\n_{k}\Psi,  \n^{k}\Psi \ran_{V} - U(\Psi) + 
	\frac{1}{2} F^a_{i j} F^{a,i j},\\
 \label{YMH-energy-coord}	&E_{\Omega} (\Psi, A) =   \int_{\Omega}  \lan\n_{k}\Psi,  \n_{k}\Psi \ran_{V} + U(\Psi) + \frac{1}{2} F^a_{i j} F^{a}_{i j}
	\end{align}
(with different ranges of indices as mentioned above). 	The  YMH equations 
	are given in coordinates by
	\begin{align} \label{YMH-eqs-psi-coord}
	&-  \n^{i} \n_{i}\Psi  =  U'(\Psi),\\  \label{YMH-eqs-A-coord}
	&-\n^{i} F_{i j} = \Re\langle {\g_a}\Psi, \n_j\Psi\rangle_V \g_a.\ 
	\end{align} 
\end{proposition}

\begin{proof}
	Equations \eqref{YMH-Lagr-coord} and \eqref{YMH-energy-coord} follow from the coordinate expressions $d_{A} \Psi =\n_k \Psi dx^k$ and $F_A = F_{ij}^a \g_a\otimes dx^i \wedge dx^j$, together with the fact that $dx^k$ and $\g_a\otimes dx^i\wedge dx^j$ form orthonormal bases for $\Omega^1$ and $\Omega^2_\fg$, respectively.
	
	Equations \eqref{YMH-eqs-psi-coord} - \eqref{YMH-eqs-A-coord} follow from equations \eqref{YMH-eqs-psi} - \eqref{YMH-eqs-curr} and the coordinate expressions for $d_A$ and $d_A^*$ above.
\end{proof}

\section{The WS equations in coordinate form} \label{sec:WS-coord}

For the gauge group $G=U(2)=SU(2)\times U(1)$, we choose the standard inner product \begin{equation}\langle \g, \delta \rangle_{\mathfrak{u}(2)}  := 2\Tr\g^*\del =  -2 \Tr\g\delta\end{equation} on $\mathfrak{u}(2)$, for which $-\frac i2\tau_a$, $a=0,1,2,3$, (where $\tau_a$, $a=1,2,3$, are the the Pauli matrices together with $\tau_0  :=   \one$) form an orthonormal basis. It is customary to factor out the coefficient of $-\frac i2$. 
In coordinates, we write 
 \begin{equation}\label{nabQ-spec}\n_{Q} \Phi =\n_{i} \Phi dx^i, \quad Q =-\frac{i}{2} Q_i dx^i\ \text{ and }\ F_Q =-\frac{i}{2} Q_{ij} dx^i\wedge dx^j,\end{equation} 
with $Q_i(x)$, $ Q_{ij}(x)\in i\mathfrak{u}(2)$. Using equation \eqref{FQ}, we compute $Q_{ij} = \frac12 (\partial_i Q_j - \partial_j Q_i) - \frac i{4} [Q_i,Q_j]$. 
Furthermore, we write $Q=V+X$ and 
\begin{equation}V = -\frac{i}{2}V_i dx^i\ \text{ and }\ X = -\frac{i}{2}X_i dx^i,\end{equation} with  $V_i(x)\in i\mathfrak{su}(2)$ and $X_i(x)\in i\mathfrak{u}(1)$. Then $Q_{i j}=V_{i j} + X_{i j} $ and 
\begin{align}\n_{i} \Phi  &:=  (\p_i -   \frac{ig}{2}  V_{i} -    \frac{ig'}{2} X_{i})  \Phi, \\
\label{Vij}V_{i j} &:=  \frac12(\p_i V_j-  \p_j V_i) - \frac{ig}{4}[V_i,  V_j], \\
\label{Xij}X_{i j} &:=    \frac12(\p_i X_j-  \p_j X_i).\end{align} 
\DETAILS{For a $\mathfrak{u}(2)$-valued function $f$ and a $\mathfrak{u}(2)$-valued $1-$form $B=B_{i} dx^{i}$, where $B_{i}$ are  $\mathfrak{u}(2)$-valued functions, we define $d_Q f = d_V f:= d f - \frac{ig}2 [V,f] $ and $d_Q B = d_V B:=  (d_V B_{i})\wedge dx^{i}$.

  $F_{Q} :=  F_{V} + F_{X}$, where $F_V :=  dV - \frac{ig}{2} [V \wedge V]$ and $F_X :=  dX$. 
 Here $[A\wedge B]$ is defined as 
 \begin{equation}\label{commut-def} [A\wedge B]:= 
  A^a\wedge B^b [\tau_a,  \tau_b] =[B\wedge A], \end{equation}
 where $A= A^a \tau_a$, $B= B^a \tau_a$ and  $\tau_a$, $a=1,2,3$ are the Pauli matrices
 \begin{align}
  \tau_1:= \left(\begin{array}{cc}
  0 & 1 \\ 1 & 0 \end{array} \right),\
  \tau_2:= \left(\begin{array}{cc}
  0 & -i \\ i & 0 \end{array} \right),\
  \tau_3:= \left(\begin{array}{cc}
  1 & 0 \\ 0 & -1 \end{array} \right).
 \end{align}}

We specify equation \eqref{YMH-Lagr-coord} - \eqref{YMH-eqs-A-coord} for to the Weinberg-Salam (WS) model, which has the gauge group $G=U(2)=SU(2)\times U(1)$. As was mentioned in Appendix \ref{sec:Cov-Deriv-Curv}, in this case, there is a slight discrepancy in the definition of the covariant derivative due to the fact that $U(2)$ is not simple, but a (semi-)direct product of the simple group $SU(2)$ and $U(1)$, with each component having a coupling constant, see \eqref{nabQ-spec}-\eqref{Xij}. 
\DETAILS{We choose the standard inner product \begin{equation}\langle \g, \delta \rangle_{\mathfrak{u}(2)}  := 2\Tr\g^*\del =  -2 \Tr\g\delta\end{equation} on $\mathfrak{u}(2)$, for which $-\frac i2\tau_a$, $a=0,1,2,3$, (where $\tau_a$, $a=1,2,3$, are the the Pauli matrices together with $\tau_0  :=   \one$) form an orthonormal basis. It is customary to factor out the coefficient of $-\frac i2$. 
In coordinates, we write 
 \begin{equation}\n_{Q} \Phi =\n_{i} \Phi dx^i, \quad Q =-\frac{i}{2} Q_i dx^i\ \text{ and }\ F_Q =-\frac{i}{2} Q_{ij} dx^i\wedge dx^j,\end{equation} 
with $Q_i(x)$, $ Q_{ij}(x)\in i\mathfrak{u}(2)$. Using equation \eqref{FQ}, we compute $Q_{ij} = \frac12 (\partial_i Q_j - \partial_j Q_i) - \frac i{4} [Q_i,Q_j]$. 
Furthermore, we write $Q=V+X$ and 
\begin{equation}V = -\frac{i}{2}V_i dx^i\ \text{ and }\ X = -\frac{i}{2}X_i dx^i,\end{equation} with  $V_i(x)\in i\mathfrak{su}(2)$ and $X_i(x)\in i\mathfrak{u}(1)$. Then $Q_{i j}=V_{i j} + X_{i j} $ and 
\begin{align}\n_{i} \Phi  &:=  (\p_i -   \frac{ig}{2}  V_{i} -    \frac{ig'}{2} X_{i})  \Phi, \\
V_{i j} &:=  \frac12(\p_i V_j-  \p_j V_i) - \frac{ig}{4}[V_i,  V_j], \\
X_{i j} &:=    \frac12(\p_i X_j-  \p_j X_i).\end{align}} 

Using Eqs \eqref{nabQ-spec}-\eqref{Xij}, we express the Lagrangian and the energy 
  in coordinates as
\begin{align} \label{EW-Lagr-coord}
&\cL (\Phi, Q) :=  \int_{\Omega}  \lan\n_{i}\Phi,  \n^{i}\Phi \ran_{\C^2} - U(\Phi) + \frac{1}{2} \Tr Q_{i j} Q^{i j},\\ 
 \label{EW-energy-coord}&E (\Phi, Q) :=  \int_{\Omega}  \lan\n_{i}\Phi,  \n_{i}\Phi \ran_{\C^2} + U(\Phi) + \frac{1}{2} \Tr Q_{i j} Q_{i j},\end{align}
\DETAILS{where $\n_{i}$ is  the covariant derivative given by $\n_{i}:= \p_i -   g  i\frac{1}{2}V_{i} -    g'  i\frac{1}{2} X_{i}$,  
	$G_{i j}:= 
	\p_i G_j-  \p_j G_i +[G_i,  G_j]$ and $g,\ g'$  and $\lam$ are positive constant (coupling constants).}
(with indices ranging from $0$ to $d$ and $1$ to $d$, respectively, as mentioned above), and the Euler-Lagrange equations are written in coordinates as
\begin{align} \label{WS-eqs-orig-1} 
-&\n^i\n_i\Phi = U'(\Phi), \\
&\n^i Q_{ij} = \frac12 g \Im \lan\tau_a\Phi, \n_j\Phi\ran_{\C^2} \tau_a + \frac12 g' \Im \lan\tau_0\Phi, \n_j\Phi\ran_{\C^2} \tau_0. \label{WS-eqs-orig-2} 
\end{align}

Eqs. \eqref{EW-energy-coord} - \eqref{WS-eqs-orig-2} can be expressed in terms of the $W$, $Z$, Higgs and electromagnetic fields resulting  in 2D equations 
 \eqref{WS-energy'} - \eqref{WS-eq4'}, see Appendix \ref{sec:WSeqs2D}. 


\section{The Weinberg-Salam energy 
  in terms of the fields $W$, $A$, $Z$ and $\vphi$}\label{sec:en-expl}

\subsection{Dimension $3$} 
\label{sec:en-expl}

We work in a fixed coordinate system, $\{x^i\}_{i=1}^3$ and write the fields as $W = W_{i} dx^i,$
$ Z = -\frac i2 Z_{i} dx^i$  and $A = -\frac i2 A_{i} dx^i$. We show 
\begin{proposition}\label{prop:WS-energy} Energy \eqref{WS-energy0}, written in terms of the fields $W,A,Z$ and $\vphi$ and coordinates $\{x^i\}_{i=1}^3$, is given by (see also \cite{SY}): 
 \begin{align} \label{WS-energy}  E_{\Omega} &(W, A, Z, \vphi) :=  \int_{\Omega } \big[   \sum_{ij} (\frac{1}{2} |W_{i j}|^2 +  \frac{1}{4}  |Z_{i j}|^2 +\frac{1}{4} |A_{i j}|^2)  \notag \\ & \qquad  +\frac{1}{2} g^2 \vphi^2|W|^2+\frac{1}{4\cos^2\theta} g^2 \vphi^2|Z|^2
 +T(W,A,Z) \notag \\ & \qquad   \qquad  \qquad  \qquad  \qquad +  |\n \vphi|^2  + \frac{1}{2} \lam (\vphi^2-\vphi_0^2)^2\big], \end{align} 
 where 
$W_{i j}:= \n_{i} W_{j} -\n_{j}W_{i}$, with $\n_{k}:= \p_k-i g V^3_k, \partial_k\equiv \partial_{x^k}$, $Z_{i j}:= \p_{i} Z_{j} -\p_{j}Z_{i}$, $A_{i j}:= \p_{i} A_{j} -\p_{j}A_{i}$ 
and $T(W, A, Z)$ is the sum of super-quadratic terms,
\begin{align} \label{T-def}
 T(W,A,Z) :=  &\frac{g^2}{2} \sum_{ij} (|W_i W_j|^2 - W_i^2 \overline{W}_j^2)- ig \sum_{ij}V_{ij}^3 W_i \overline{W}_j,
\end{align}
where $V^3 :=  Z\cos\theta + A\sin\theta$ and $V^3_{ij} :=  \partial_i V_j - \partial_j V_i$,  with the important property that $T(W,A,Z)$ is invariant under the gauge transformation \eqref{gauge-transf'}. 
 \end{proposition}

\begin{proof}[Proof of \eqref{WS-energy}]
	We proceed by rewriting the terms in the coordinate expression of the WS energy \eqref{EW-energy-coord},
	in terms of the fields  $W = W_{i} dx^i$, $Z = -\frac i2 Z_{i} dx^i$, $A = -\frac i2 A_{i} dx^i$ and $\vphi$. 
	
	For the first term, first we calculate $\n_i\Phi$. Recall the definition $\n_i\Phi := (\partial_i - \frac{ig}{2}V_i - \frac{ig'}{2} X_i)\Phi$. We simplify the matrix representing the connection's action on $\Phi$:
	\begin{align}
	- \frac{ig}{2}V_i &- \frac{ig'}{2} X_i = -\frac{ig}2 V_i^a\tau_a -\frac{ig'}2 X_i\tau_0 \nonumber \\
	&= -\frac{ig}2\left(\begin{array}{cc}
	0 & V^1_i \\ 
	V^1_i & 0 
	\end{array} \right) 
	-\frac{ig}2\left(\begin{array}{cc}
	0 & -iV^2_i \\ 
	iV^2_i & 0 
	\end{array} \right) \nonumber \\
	&-\frac{ig}2\left(\begin{array}{cc}
	V^3_i & 0 \\ 
	0 & -V^3_i 
	\end{array} \right)
	-\frac{ig}2\tan\theta\left(\begin{array}{cc}
	X_i & 0\\ 
	0 & X_i 
	\end{array} \right) \nonumber \\
	&= -\frac{ig}{2\cos\theta}\left(\begin{array}{cc}
	V^3_i\cos\theta + X_i\sin\theta & V^1_i\cos\theta - iV^2_i\cos\theta \\
	V^1_i\cos\theta + iV^2_i\cos\theta  & -V^3_i\cos\theta + X_i\sin\theta 
	\end{array} \right). \label{conn-matr}
	\end{align}
	In terms of the fields $Z$, $A$ and $W$ (see equations \eqref{Z,A-fields} - \eqref{W-field} for the definitions of these fields), \eqref{conn-matr} becomes
	\begin{align}
	- \frac{ig}{2}V_i - \frac{ig'}{2} X_i &= -\frac{ig}{2\cos\theta}\left(\begin{array}{cc}
	Z_i\cos2\theta + A_i\sin2\theta & \sqrt{2}\ W_i\cos\theta \\ 
	\sqrt{2}\ \overline W_i\cos\theta & -Z_i
	\end{array} \right).
	\end{align}
	Hence, for $\Phi = (0, \vphi)$,
	\begin{align}
	\n_i\Phi = \left(\begin{array}{cc}
	-\frac{ig}{\sqrt 2}W_i\vphi \\ 
	\partial_i\vphi + \frac{ig}{2\cos\theta}Z_i\vphi
	\end{array} \right).
	\end{align}
	Therefore, the first term of \eqref{EW-energy-coord}, written in terms of the fields $W,A,Z$ and $\vphi$, becomes
	\begin{align}
	\langle\n_i\Phi, \n^i\Phi\rangle_{\C^2} &= \overline{\frac{ig}{\sqrt2}W_i} \frac{ig}{\sqrt2}W^i \nonumber \\
	&+ \overline{(\partial_i\vphi + \frac{ig}{2\cos\theta}Z_i\vphi)}(\partial^i\vphi + \frac{ig}{2\cos\theta}Z^i\vphi) \nonumber \\
	&= \frac{g^2}{2}\vphi^2|W|^2 + |\n\vphi|^2 + \frac{g^2}{4\cos^2\theta}\vphi^2|Z|^2. 
	\label{WS-energy-term1}
	\end{align}
	
	The second term of \eqref{EW-energy-coord} becomes
	\begin{align}
	U(\Phi) &= \frac12\lambda (\|\Phi\|^2 - \vphi_0^2)^2 
	= \frac12\lambda (\vphi^2 - \vphi_0^2)^2.
	\label{WS-energy-term2}
	\end{align}
	
	For the third term of \eqref{EW-energy-coord}, we will use the fact that $\Tr Q_{i j} Q^{i j} = \Tr V_{i j} V^{i j} + \Tr X_{i j} X^{i j}$, where $V_{i j}$ and $X_{i j}$ are defined in \eqref{Vij} and \eqref{Xij}. Furthermore, we have
	\begin{align}
	 V_i  :=  V_i^a\tau_a = \left(\begin{array}{cc}
	 V_i^3 & \sqrt{2}\ W_i \\ 
	 \sqrt{2}\ \overline W_i & -V_i^3 
	 \end{array} \right). \label{vi}
	\end{align}
	We recall $V_{ij}^3 = \p_i V_j^3 - \p_j V_i^3$ and $W_{ij}^0 = \partial_i W_j - \partial_j W_i$ and calculate	
	\begin{align}
	\frac12  (\partial_i V_j - \partial_j V_i) 
	 = \frac12\left(\begin{array}{cc}
	V_{ij}^3 & \sqrt{2}\ W_{ij} \\ 
	\sqrt{2}\ \overline W_{ji}^0 & -V_{ij}^3
	\end{array} \right), \label{vij-deriv}
	\end{align}
	and, with $K_{ij}:= V^3_i W_i - V^3_j  W_i$,
	\begin{align}
	-\frac{ig}{4}[V_i, V_j] &= -\frac{ig}{4}\left(\begin{array}{cc}
	V_i^3 & \sqrt{2}\ W_i \\ 
	\sqrt{2}\ \overline W_i & -V^3_i 
	\end{array} \right)
	\left(\begin{array}{cc}
	V_j^3 & \sqrt{2}\ W_j \\ 
	\sqrt{2}\ \overline W_j & -V^3_j 
	\end{array} \right) - (i\leftrightarrow j) \nonumber \\
	&=  -\frac{ig}{4}\left(\begin{array}{cc}
	V_i^3V_j^3 + 2W_i\overline W_j &\sqrt{2}\ K_{ij} \\ 
	\sqrt{2}\ \overline K_{ij}  & -V_i^3V_j^3 - 2W_i\overline W_j
	\end{array} \right) - (i\leftrightarrow j) \nonumber \\
	&=-\frac{ig}2\left(\begin{array}{cc}
	W_i\overline W_j - \overline W_i W_j & \sqrt{2}\ K_{ij} \\ 
	\sqrt{2}\ \overline K_{ji} & -W_i\overline W_j + \overline W_i W_j
	\end{array} \right). \label{vij-comm}
	\end{align}
 Adding \eqref{vij-deriv} and \eqref{vij-comm}, using that $W_{ij}=W_{ij}^0+K_{ij}$ and denoting $L_{ij}:=V^3_{ij} - ig (W_i\overline W_j - \overline W_i W_j)$ gives
	\begin{align}
	&V_{ij} = \frac12\left(\begin{array}{cc}
	L_{ij} & \sqrt{2}\ W_{ij} \\
	-\sqrt{2}\ \overline W_{ij}  & -L_{ij}
	\end{array} \right).
	\end{align}
	Since $V_{ij}$ and $X_{ij}$ are Hermitian, $\Tr V_{ij}V^{ij}$ and $\Tr X_{ij}X^{ij}$ are the sum of the squared absolute values of the matrix coefficients of $V_{ij}$ and $X_{ij}$, respectively. Thus
	\begin{align}
	\frac{1}{2}& \Tr Q_{i j} Q^{i j} = \frac{1}{2} \Tr V_{i j} V^{i j} + \frac{1}{2} \Tr X_{i j} X^{i j} \nonumber \\
	 &=\frac18 \sum\limits_{ij} 2|L_{ij}|^2
	+ 4|W_{ij}|^2 + 2|X_{ij}|^2.
	\end{align}
	Using $L_{ij} = V^3_{ij} - ig (W_i\overline W_j - \overline W_i W_j)$ and expanding the first term gives
	\begin{align}
	\frac{1}{2} &\Tr Q_{i j} Q^{i j} = \sum\limits_{ij} \frac12|W_{ij}|^2 + \frac14|V^3_{ij}|^2 + \frac14|X_{ij}|^2 \nonumber \\
		&+\frac{g^2}4\sum\limits_{ij} |W_i\overline W_j - \overline W_i W_j|^2  
	-\frac{ig}4\sum\limits_{ij} 2V^3_{ij}(W_i\overline W_j - \overline W_i W_j). \label{trQ-1}
	\end{align}
	Recall that $A_{ij} = V^3_{ij}\sin\theta + X_{ij}\cos\theta$ and $Z_{ij} = V^3_{ij}\cos\theta - X_{ij}\sin\theta$. Writing the first line of \eqref{trQ-1} in terms of these fields gives
	\begin{align}
	\frac{1}{2} &\Tr Q_{i j} Q^{i j} =\sum\limits_{ij} \frac12|W_{ij}|^2 + \frac14|Z_{ij}|^2 + \frac14|A_{ij}|^2 \nonumber \\
			&+\frac{g^2}4\sum\limits_{ij} |W_i\overline W_j - \overline W_i W_j|^2  
	-\frac{ig}2\sum\limits_{ij} V^3_{ij}(W_i\overline W_j - \overline W_i W_j). \label{trQ-2}
	\end{align}
	Expanding the first term of the second line, and using $V_{ij}^3 = -V^3_{ij}$ in the second term, \eqref{trQ-2} becomes
	\begin{align}
	\frac{1}{2} \Tr Q_{i j} Q^{i j} &=\sum\limits_{ij} \frac12|W_{ij}|^2 + \frac14|Z_{ij}|^2 + \frac14|A_{ij}|^2 \nonumber \\
	&+\frac{g^2}4 \sum\limits_{ij}(|W_i|^2 |\overline W_j|^2 - W_i^2 \overline W_j^2 + (i\leftrightarrow j)) \nonumber \\ &-\frac{ig}2\sum\limits_{ij}( V^3_{ij} W_i\overline W_j + (i\leftrightarrow j) ).
	\end{align}
	Recalling the definition \eqref{T-def} of $T(W,A,Z)$ gives
	\begin{align}
\frac{1}{2} \Tr Q_{i j} Q^{i j}	&=\sum\limits_{ij} \frac12|W_{ij}|^2 + \frac14|A_{ij}|^2 + \frac14|Z_{ij}|^2 + T(W,A,Z).
	\label{WS-energy-term3}
	\end{align}
	
	Adding \eqref{WS-energy-term1}, \eqref{WS-energy-term2} and \eqref{WS-energy-term3} gives \eqref{WS-energy}.
\end{proof}


\subsection{Dimension $2$: Proof of \eqref{WS-energy'}} 
\label{sec:WSeqs2D}

\begin{proof}[Proof of \eqref{WS-energy'}]
Now, we consider the Weinberg-Salam (WS) model in $\R^2$ with fields independent of the third dimension $x_3$, and correspondingly choose the gauge with $V_3 = X_3=0$ (and hence $W_3 = A_3=Z_3=0$). In this case the summation in \eqref{WS-energy} contains only two terms, $(ij)=(12)$ and $(ij)=(21)$, and we use this to simplify \eqref{WS-energy}.

	We proceed by simplifying the terms of \eqref{T-def} and the first line of \eqref{WS-energy}; the remaining terms are unchanged.
\begin{align}
 \sum_{ij} (\frac{1}{2} |W_{i j}|^2 +  \frac{1}{4}  |Z_{i j}|^2 +&\frac{1}{4} |A_{i j}|^2) =\sum_{i<j} (|W_{i j}|^2 +  \frac12 |Z_{i j}|^2 +\frac12 |A_{i j}|^2) \nonumber \\
 &= |\curl_{gV^3}W|^2 + \frac12|\curl Z|^2 + \frac12|\curl A|^2; \label{WS-energy'1}
\end{align}	
\begin{align}
 \sum_{ij} (|W_i W_j|^2 &- W_i^2 \overline W_j^2) \notag \\ 
 &=  W_1 W_2 \overline W_1\overline W_2 - W_1^2\overline W_2^2 + W_2 W_1 \overline W_2\overline W_1 - W_2^2\overline W_1^2 \nonumber \\
 &= \overline{(\overline W_1 W_2 - W_1\overline W_2)}(\overline W_1 W_2 - W_1\overline W_2) \nonumber \\
 &=  |\overline W\times W|^2; \label{WS-energy'2}
\end{align}
\begin{align}
 - \sum_{ij}V^3_{ij}W_i\overline W_j &= \sum_{i<j} V^3_{ij}(-W_i\overline W_j + W_j\overline W_i) \nonumber \\
 &= (\curl V^3) \overline W\times W. \label{WS-energy'3}
\end{align}
Replacing corresponding terms in \eqref{WS-energy} - \eqref{T-def} with \eqref{WS-energy'1} - \eqref{WS-energy'3} proves \eqref{WS-energy'}.
\end{proof}

\begin{proof}[Proof of \eqref{WS-eq1'} - \eqref{WS-eq4'}]
	We proceed by calculating the (complex) G\^ateaux derivatives of \eqref{WS-energy'}. 
	
 Let $\del_\#$ denote the partial (real) G\^ateaux derivative with respect to $\#$. Let $W_z = W + z W'$, $z\in\C$, and define $\p_{\overline z}\equiv \frac12(\partial_{\Re z} + i\partial_{\Im z})$ and $\del_{\overline W}\equiv \frac12(\del_{\Re W} + i\del_{\Im W})$. Then
	\begin{align}
	\del_{\overline W} E_{\Omega} (W, A, Z, \vphi)\overline{W'} &= \partial_{\overline z} E_{\Omega}^{WS}(W_z, A, Z, \vphi)|_{z=0} \nonumber \\
	= \int_{\Omega } &\curl_{gV^3}W \cdot \overline{\curl_{gV^3}W'} + \frac12 g^2\vphi^2 W\cdot\overline{W'} \nonumber \\ &- ig(\curl V^3)JW\cdot\overline{W'} + g^2(\overline W\times W) JW\cdot\overline{W'}.
	\end{align}
	Integrating the first term by parts and factoring out $W$ and $\overline W'$ gives
	\begin{align}
	\del_{\overline W} E_{\Omega} (W, A, Z, \vphi)\overline{W'} = \int_{\Omega }& [\curl_{gV^3}^* \curl_{gV^3} + \frac{g^2}{2}\varphi^2 - ig(\curl V^3)J \nonumber \\ 
	&+ g^2(\overline{W} \times W)J]W \cdot \overline{W'}.
	\end{align}
	For the derivative to be zero for every variation $W'$, \eqref{WS-eq1'} must hold.
	
	Let $A_s = A + s A'$, $s\in\R$. Then
	\begin{align}
	\del_A E_{\Omega} (W, A, Z, \vphi)A' &= \partial_s E_{\Omega}^{WS}(W, A_s, Z, \vphi)|_{s=0} \nonumber \\
	= \int_{\Omega } &\curl_{gV^3} W\overline{(-ieA'\times W)} + \overline{\curl_{gV^3} W}(-ieA'\times W)
	\nonumber \\ &+(\curl A) (\curl A') + ie(\curl A')\overline W\times W.
	\end{align}
	Using $A'\times W = -JW\cdot A'$ in the first two terms, and integrating the last two terms by parts, gives
	\begin{align}
	\del_A E_{\Omega} (W, A, Z, \vphi)A' &= \int_{\Omega } [-ie(\curl_{gV^3}W)J\overline W + ie\overline{(\curl_{gV^3} W)J\overline W} \nonumber \\
	& + \curl^*\curl A + ie\curl^*(\overline W\times W)] \cdot A',
	\end{align}
	which simplifies to
	\begin{align}
	\del_A E_{\Omega} (W, A, Z, \vphi)A' = \int_{\Omega } &[\curl^* \curl A + 2e\Im[(\curl_{gV^3}W)J\overline{W} \nonumber \\ &- \curl^*(\overline{W}_1 W_2)]]\cdot A'.
	\end{align}
	For the derivative to be zero for every variation $A'$, \eqref{WS-eq2'} must hold.
	
	The proof of \eqref{WS-eq3'} is essentially the same as the proof of \eqref{WS-eq2'}, so we omit it.
	
	Let $\vphi_s = \vphi + s \vphi'$, $s\in\R$. Then
	\begin{align}
	\del_{\vphi} E_{\Omega} (W, A, Z, \vphi)\vphi' &= \partial_s E_{\Omega}^{WS}(W, A, Z, \vphi_s)|_{s=0} \nonumber \\ \int_{\Omega} & g^2\vphi\vphi'|W|^2 + \frac{g^2}{2\cos^2\theta}\vphi\vphi'|Z|^2 \nonumber \\
	&+ 2\n\vphi'\cdot\n\vphi + 2\lambda(\vphi^2 - \vphi_0^2)\vphi\vphi' \nonumber \\
	\end{align}
	Integrating the third term by parts and factoring out $2\vphi'$ gives
	\begin{align}
	= \int_{\Omega} &[\frac{g^2}2|W|^2 + \frac12\kappa|Z|^2 \nonumber \\
	&-\Delta + \lambda (\vphi^2-\vphi_0^2)]\vphi \cdot 2\vphi'.
	\end{align}
	For the derivative to be zero for every variation $\vphi'$, \eqref{WS-eq4'} must hold.
\end{proof}

 \section{Proof of \eqref{phi0-eqn'}} \label{Sec:Lemma-tilde-phi0}

In the proof below, we will use the following result: 
\begin{lemma} \label{os-lem}
	Let $L^2_{per}$ denote any of the spaces \eqref{L2n-space} - \eqref{L2-space}, and let $\cH^2_{per}$ denote the corresponding Sobolev space. Suppose that $f_s,g_s:\R\to\cH^2_{per}$ satisfy  $||f_s||_{\cH^2_{per}} = \cO(|s|^k)$ and $||g_s||_{\cH^2_{per}} = \cO(|s|^l)$ for some $k,l\in\Z$. Then for $i,j=1,2$ and $p,q=0,1$,
	\begin{align}\label{os-rem}
	|\int_{\Omega'} \partial_i^p f_s \partial_j^q g_s| = \cO(|s|^{k+l}).
	\end{align}
	Furthermore, if $f_s$ and $g_s$ have continuous derivatives of all orders in $s$, then so does the above integral.	
\end{lemma}

\begin{proof}
	Equation \eqref{os-rem} follows from the following chain of inequalities:
   \begin{align}
	|\int_{\Omega'} \partial_i^p f_s \partial_j^q g_s| 
	&\lesssim || \partial_i^p f_s||_{\cL^2_{per}} ||\partial_j^q g_s||_{\cL^2_{per}} \notag \\
	&\lesssim ||f_s||_{\cH^2_{per}} ||g_s||_{\cH^2_{per}} = \cO(|s|^{k+l}).
	\end{align}
	
	If $f_s$ and $g_s$ have continuous derivatives of all orders in $s$, then their $s$-derivatives of all orders are in $\cH^2_{per}$. In particular, this means that $\p_s^k (f_s g_s)$, $k\in\Z_{\geq 0}$, remains integrable, so the $s-$derivatives of the above integral (obtained by differentiation under the integral sign) are well-defined. 
\end{proof}

\begin{proof}[Proof of \eqref{phi0-eqn'}]
To prove \eqref{phi0-eqn'}, 
 we use the $w$-field Equation \eqref{WS-eq1-resc}, and $\nu_s := g(a_s\sin\theta + z_s\cos\theta)$, 
 to get
	\begin{align}
	\int_{\Omega'} \overline{\chi} \cdot [\curl_{\nu_s}^* \curl_{\nu_s}& + \frac{g^2}{2}(\psi_s + \xi_{s})^2 \notag \\ 
	&- i(\curl \nu_s)J + g^2(\overline{w_s} \times w_s)J]w_s = 0. \label{phi0-preqn}
	\end{align}
We shall calculate each term of the integral \eqref{phi0-preqn} up to order $s^3$ using Lemma \ref{os-lem} and the Taylor expansions \eqref{s-asym}.
\DETAILS{		\begin{align}\label{phi0-eqn-s}
		g\sqrt{2n} \xi' \langle & |\chi|^2 \rangle  \notag \\ 
	&= -g^2[m_w^2\langle U_{m_z,m_h}(|\chi|^2)|\chi|^2\rangle+\sin^2\theta \langle |\chi|^2\rangle^2],
		\end{align}
		where $m_w :=  \sqrt n$ is the mass of the rescaled $W$ boson field $w$, and for $f\in\cL^2_{loc}$,
		\begin{align}
		U_{M_1,M_2}(f)(\rho) := &  \frac{1}{2\pi}\int_{M_2}^{M_1} \int |\rho-\rho'| K_1(M|\rho-\rho'|) f(\rho') d^2\rho' dM \nonumber \\
		=& U_{M_1}(f)(\rho)-U_{M_2}(f)(\rho),
		\end{align}
		with $K_1$ the modified Bessel function of the third kind.} 
	
	Integrating the first term of \eqref{phi0-preqn} by parts gives
	\begin{align}
	\int_{\Omega'} \overline{\chi}\cdot \curl^*_{\nu_s}&\curl_{\nu_s}w_s =  \int_{\Omega'} \overline{\curl_{\nu_s}\chi}\cdot \curl_{\nu_s}w.
	\end{align}
	Plugging in the Taylor expansions \eqref{s-asym} gives
	\begin{align}
	\int_{\Omega'} \overline{\chi}\cdot \curl^*_{\nu_s}\curl_{\nu_s}w_s = \int_{\Omega'} &[\overline{\curl_{a^n} \chi} + \cO(|s|^2)] \nonumber \\ \cdot &[s\curl_{a^n} \chi - s^3 i\nu' w' + \cO(|s|^5)],
	\end{align}
	where, recall, $\nu' := g(a'\sin\theta + z'\cos\theta)$. Recall from Equation \eqref{chi-def} that $\curl_{a^n}\chi = 0$. Therefore, applying Lemma \ref{os-lem} gives
	\begin{align}
	\int_{\Omega'} \overline{\chi}\cdot \curl^*_{\nu_s}&\curl_{\nu_s}w_s 
	= \cO(|s|^5). \label{phi0-eqn-s1}
	\end{align}
	
	Plugging the Taylor expansions \eqref{s-asym} into the second term of \eqref{phi0-preqn} gives
\begin{align}
	\int_{\Omega'} \overline{\chi}\cdot \frac{g^2}{2} (\psi_s + \xi_{ s})^2 w_s &= \int_{\Omega'} \overline{\chi}\cdot \frac{g^2}{2} (\frac{\sqrt{2n}}{g} + s^2 (\psi' + \xi') + \cO(|s|^4))^2 \notag \\ 
	&\quad\quad\quad\quad \times (s\chi + \cO(|s|^5)).
	\end{align}
	Expanding this product and applying Lemma \ref{os-lem} gives
	\begin{align}
	\int_{\Omega'} \overline{\chi}\cdot \frac{g^2}{2} (\psi_s + \xi_{ s})^2 w_s &= s \int_{\Omega'} n|\chi|^2 + s^3\int_{\Omega'} g\sqrt{2n}(\psi' + \xi') |\chi|^2 \notag \\ 
	&\quad\quad + s^3\int_{\Omega'} n\overline\chi\cdot w' + \cO(|s|^5).
	\end{align}
	Recall that $\chi\in\Null(H_1(n))$ and that $w'$ is orthogonal to $\text{Null}\ H_1(n)$. Therefore the third term vanishes:
	\begin{align}
	\int_{\Omega'} \overline{\chi}\cdot \frac{g^2}{2} (\psi_s + \xi_{ s})^2 w_s &= s \int_{\Omega'} n|\chi|^2 + s^3\int_{\Omega'} g\sqrt{2n}(\psi' + \xi') |\chi|^2 \notag \\ 
	&\quad\quad + \cO(|s|^5). \label{phi0-eqn-s2}
	\end{align}
	
	Plugging the Taylor expansions \eqref{s-asym} into the third term of \eqref{phi0-preqn} gives
	\begin{align}
	\int_{\Omega'} \overline{\chi}\cdot (- i(\curl\nu_s)Jw_s)
	&= \int_{\Omega'} \overline{\chi}\cdot(-in - s^2 i(\curl\nu') + \cO(|s|^4))\notag \\ &\quad\quad\quad\quad \times (sJ\chi + s^3J w' + \cO(|s|^5)). \label{3rd-term}
	\end{align}
		Recall from Equation \eqref{chi-def} that $\chi$ is of the form $\chi = (\om, i\om)^T$, so\\ $\overline\chi\cdot J\chi = -i|\chi|^2$ and $\overline\chi\cdot J w' = -i\overline\chi\cdot  w'$. Therefore \eqref{3rd-term} simplifies to
		\begin{align}
		\int_{\Omega'} \overline{\chi}\cdot (- i(\curl\nu_s)Jw_s)
		&= \int_{\Omega'} (-in - s^2 i(\curl\nu') + \cO(|s|^4))\notag \\ &\quad\quad\quad\quad \times (-si|\chi|^2 - s^3i\overline\chi\cdot w' + \cO(|s|^5)).
		\end{align}
		 Expanding this product and applying Lemma \ref{os-lem} gives
	\begin{align}
	\int_{\Omega'} \overline{\chi}\cdot (- i(\curl\nu_s)Jw_s)
	&= -s\int_{\Omega'} n|\chi|^2 - s^3\int_{\Omega'} (\curl\nu') |\chi|^2 \notag \\ &\quad\quad - s^3\int_{\Omega'} n\overline\chi\cdot w' + \cO(|s|^5).
	\end{align}
	Recall that $\chi\in\Null(H_1(n))$ and that $w'$ is orthogonal to $\text{Null}\ H_1(n)$. Therefore the third term vanishes:
	\begin{align}
	\int_{\Omega'} \overline{\chi}\cdot (- i(\curl\nu_s)Jw_s)
	&= -s\int_{\Omega'} n|\chi|^2 - s^3\int_{\Omega'} (\curl\nu') |\chi|^2 \notag \\ &\quad\quad + \cO(|s|^5).\label{phi0-eqn-s3}
	\end{align}

	Using $\overline\chi\cdot Jw_s = -\overline\chi\times w_s$, the fourth term of \eqref{phi0-preqn} becomes
	\begin{align}
	\int_{\Omega'} \overline\chi \cdot (g^2 \overline w_s \times w_s) Jw_s& = \int_{\Omega'} -g^2 (\overline\chi \times w_s) \times (\overline w_s \times w_s).
	\end{align}
	Plugging in the Taylor expansions \eqref{s-asym} gives
	\begin{align}
	\int_{\Omega'} \overline\chi \cdot (g^2 \overline w_s \times w_s) Jw_s& =  \int_{\Omega'} -g^2 (s \overline\chi \times \chi + \cO(|s|^3)) \notag \\ 
		&\quad\quad\quad\quad\times (s^2 \overline\chi \times \chi + \cO(|s|^4)).
	\end{align}
	Recall from Equation \eqref{chi-def} that $\chi$ is of the form $\chi = (\om, i\om)$, so $\overline\chi\times\chi = i|\chi|^2$. This fact and Lemma \ref{os-lem} gives
	\begin{align}
	\int_{\Omega'} \overline\chi \cdot (g^2 \overline w_s \times w_s) Jw_s& = s^3\int_{\Omega'} g^2 |\chi|^4 + \cO(|s|^5). \label{phi0-eqn-s4}
	\end{align}
	
	The $s^3$ terms of \eqref{phi0-eqn-s1}, \eqref{phi0-eqn-s2}, \eqref{phi0-eqn-s3} and \eqref{phi0-eqn-s4} must sum to $0$, and so \eqref{phi0-eqn'} results.
\end{proof}

\section{Proof of \eqref{WS-energy-s-1}} \label{Sec:Lemma-Eos}
	\begin{proof}[Proof of \eqref{WS-energy-s-1}]
	We shall calculate each term in the integral \eqref{WS-energy-resc} up to order $s^6$ using Lemma \ref{os-lem} and the Taylor expansions \eqref{s-asym}.
	
	Plugging the Taylor expansions \eqref{s-asym} into the first term of \eqref{WS-energy-resc} gives
		\begin{align}
	\int_{\Omega'}  |\curl_{\nu} w_s|^2 &= \int_{\Omega'} |s\curl_{a^n} \chi + \cO(|s|^3)|^2.
	\end{align}
	Recall from Equation \eqref{chi-def} that $\curl_{a^n}\chi = 0$. Therefore, applying Lemma \ref{os-lem} gives
	\begin{align}
	\int_{\Omega'}  |\curl_{\nu} w_s|^2 &= \cO(|s|^6). \label{EWS-asymp-1}
	\end{align}
	
	Plugging the Taylor expansions \eqref{s-asym} into the second term of \eqref{WS-energy-resc} gives
	\begin{align}
	\int_{\Omega'} \frac12 |\curl z_s|^2& = \int_{\Omega'} \frac12 |s^2\curl z' + \cO(|s|^4)|^2.
	\end{align}
	Expanding the square and applying Lemma \ref{os-lem} gives
	\begin{align}
	\int_{\Omega'} \frac12 |\curl z_s|^2 &= s^4\int_{\Omega'}\frac12 |\curl z'|^2 + \cO(|s|^6). \label{EWS-asymp-2}
	\end{align}	
	
	Plugging the Taylor expansions \eqref{s-asym} into the third term of \eqref{WS-energy-resc} gives
	\begin{align}
	\int_{\Omega'}\frac12 |\curl a_s|^2 = \int_{\Omega'}&\frac12 |\curl\frac1e a^n + s^2\curl a' + s^4\curl a'' + \cO(|s|^6)|^2.
	\end{align}
	Recall that $\curl a^n = n$. Expanding the square gives 
	\begin{align}
	\int_{\Omega'}\frac12 |\curl a_s|^2 = \int_{\Omega'} [\frac12\frac{n^2}{e^2} + s^2\frac ne\curl a' &+ s^4\frac ne\curl a'' \nonumber \\ 
	&+ s^4\frac12 |\curl a'|^2 + \cO(|s|^6)].
	\end{align}
	The second and third terms vanish because $a'$ and $a''$ are $\cL'$-periodic. Therefore, applying Lemma \ref{os-lem} gives
	\begin{align}
	\int_{\Omega'}\frac12 |\curl a_s|^2 &= \frac12\frac{n^2}{e^2}|\Omega'| + s^4\int_{\Omega'} \frac12|\curl a'|^2 + \cO(|s|^6). \label{EWS-asymp-3}
	\end{align}
	
	Plugging the Taylor expansions \eqref{s-asym} into the fourth term of \eqref{WS-energy-resc} gives 
	\begin{align}
	\int_{\Omega'}\frac12 g^2\phi_s^2|w_s|^2 = \int_{\Omega'} &\frac12 g^2[\frac{\sqrt{2n}}{g}  + s^2(\xi' + \psi') + \cO(|s|^4)]^2 \notag \\ 
	&\qquad \qquad \qquad \times |s\chi + s^3 w' + \cO(|s|^6)|^2.
	\end{align}
	Expanding the square terms gives 
	\begin{align}
	\int_{\Omega'}\frac12 g^2\phi_s^2|w_s|^2=  \int_{\Omega'}&\frac12 g^2 [\frac{2n}{g^2} + s^2 2\frac{\sqrt{2n}}{g}(\xi' + \psi') + \cO(|s|^4)] \notag \\
	&\qquad \qquad   \times [s^2 |\chi|^2 + s^4 2\Re(\overline{\chi}\cdot w') + \cO(|s|^6)].
	\end{align}
	Expanding this product and applying Lemma \ref{os-lem} gives
	\begin{align}
	\int_{\Omega'}\frac12 &g^2\phi_s^2|w_s|^2= s^2\int_{\Omega'} n |\chi|^2 \nonumber \\
	&\qquad \qquad  + s^4\int_{\Omega'} [g\sqrt{2n}(\xi'+\psi')|\chi|^2 + 2n\Re(\overline{\chi}\cdot w')] +\cO(|s|^6). 
	\end{align}
	Recall that $\chi\in\Null(H_1(n))$ and that $ w'$ is orthogonal to $\Null(H_1(n))$. Therefore the third term vanishes:
	\begin{align}
	\int_{\Omega'}\frac12 g^2\phi_s^2|w_s|^2= s^2&\int_{\Omega'} n |\chi|^2 + s^4\int_{\Omega'} g\sqrt{2n}(\xi'+\psi)|\chi|^2 +\cO(|s|^6). \label{EWS-asymp-4}
	\end{align}
	
	Plugging the Taylor expansions \eqref{s-asym} into the fifth term of \eqref{WS-energy-resc} and expanding the square terms gives 
	\begin{align}
	\int_{\Omega'} \frac{1}{4\cos^2\theta}& g^2 \phi_s^2|z_s|^2 = \int_{\Omega'} \frac{1}{4\cos^2\theta} g^2\notag \\ &\times [\frac{2n}{g^2}  + s^2 2\frac{\sqrt{2n}}{g} (\xi' + \psi') + \cO(|s|^4)]  [s^4 |z'|^2 + \cO(|s|^6)].
	\end{align}
	Expanding this product and applying Lemma \ref{os-lem} gives
	\begin{align}
	\int_{\Omega'} \frac{1}{4\cos^2\theta} g^2 \phi_s^2|z_s|^2 &= s^4 \int_{\Omega'}\frac{n}{2\cos^2\theta} |z'|^2 + \cO(|s|^6). \label{EWS-asymp-5}
	\end{align}
	
	Plugging the Taylor expansions \eqref{s-asym} into the sixth term of \eqref{WS-energy-resc} gives
	\begin{align}
	\int_{\Omega'}& |\overline w_s\times w_s|^2 = \int_{\Omega'} |s^2\overline{\chi}\times \chi + \cO(|s|^4)|^2,
	\end{align}
	Recall from Equation \eqref{chi-def} that $\chi$ is of the form $\chi = (\om,i\om)$, so $\overline\chi\times\chi = i|\chi|^2$. Therefore, applying Lemma \ref{os-lem} gives
	\begin{align}
	\int_{\Omega'}& |\overline w_s\times w_s|^2 = s^4\int_{\Omega'} |\chi|^4 + \cO(|s|^6). \label{EWS-asymp-6}
	\end{align}
	
		Plugging the Taylor expansions \eqref{s-asym} into the seventh term of \eqref{WS-energy-resc} gives
	\begin{align}
	\int_{\Omega'} i(\curl\nu_s) \overline w_s\times & w_s= \int_{\Omega'} i[g\sin\theta\curl\frac1e a^n + s^2\curl\nu' + \cO(|s|^4)] \notag \\ 
	&\times [s\overline{\chi} + s^3\overline{ w'} + \cO(|s|^5)] \times [s\chi + s^3 w' + \cO(|s|^5)]. 
	\end{align}
	 where, recall, $\nu' := g(a'\sin\theta + z'\cos\theta)$. Recall that $\curl a^n = n$ and $e = g\sin\theta$. Expanding the wedge product of the second and third terms gives
	\begin{align}
	\int_{\Omega'} i(\curl\nu_s)\overline w_s\times w_s &= \int_{\Omega'} i[\frac ng^2 + s^2\curl\nu' + \cO(|s|^4)]\notag \\ 
	&\times [s^2 \overline{\chi}\times \chi + s^4(\overline{\chi}\times  w' + \overline{ w'}\times\chi) + \cO(|s|^6)].
	\end{align}
	Recall from Equation \eqref{chi-def} that $\chi$ is of the form $\chi = (\om,i\om)$, so $\overline\chi\times\chi = i|\chi|^2$ and $\overline\chi\times w' = i\overline\chi\cdot w'$. Therefore
	\begin{align}
	\int_{\Omega'} i(\curl\nu_s)\overline w_s\times w_s = \int_{\Omega'}& [in + s^2 i\curl\nu' + \cO(|s|^4)]\notag \\ 
	&\times [s^2 i |\chi|^2 + s^4 2\Re(i\overline{\chi} \cdot  w') + \cO(|s|^6)].
	\end{align}
	Expanding this product and using Lemma \ref{os-lem} gives
	\begin{align}
	\int_{\Omega'} i(\curl\nu_s)\overline w_s\times w_s = -s^2&\int_{\Omega'} n|\chi|^2 - s^4\int_{\Omega'} [2in\Im(\overline{\chi}\cdot w') \nonumber \\
	 &\qquad  - s^4\int_{\Omega'} (\curl\nu')|\chi|^2 + \cO(|s|^6). 
	\end{align}
	Recall that $\chi\in\Null(H_1(n))$ and $ w'$ is orthogonal to $\Null(H_1(n))$. Therefore the second term vanishes:
	\begin{align}
	\int_{\Omega'} i(\curl\nu_s)\overline w_s\times w_s = &-s^2\int_{\Omega'} n|\chi|^2 
	- s^4\int_{\Omega'} (\curl\nu')|\chi|^2 \nonumber\\
	 &\qquad \qquad \qquad \qquad \qquad \qquad \qquad + \cO(|s|^6). \label{EWS-asymp-8}
	\end{align}
	
	Plugging the Taylor expansions \eqref{s-asym} into the eigth term of \eqref{WS-energy-resc} gives
	\begin{align}
	\int_{\Omega'}& |\nabla\phi_s|^2 = \int_{\Omega'} |s^2\nabla\psi' + \cO(|s|^4)|^2.
	\end{align}
		Expanding the square and using Lemma \ref{os-lem} gives
	\begin{align}
	\int_{\Omega'}|\nabla\phi_s|^2 = s^4\int_{\Omega'} |\nabla\psi'|^2 + \cO(|s|^6). \label{EWS-asymp-9}
	\end{align}
	
	Plugging the Taylor expansions \eqref{s-asym} into the ninth term of \eqref{WS-energy-resc} and expanding the inner squares gives
	\begin{align}
	\int_{\Omega'}&\frac12\lambda (\phi_s^2 - \xi_s^2) \nonumber \\
	&= \int_{\Omega'} \frac12\lambda [\frac{2n}{g^2}  + s^2 2\frac{\sqrt{2n}}{g} (\xi' + \psi')
	-\frac{2n}{g^2} - s^2 2\frac{\sqrt{2n}}{g} \xi' + \cO(|s|^4)]^2 \notag \\
	&= \int_{\Omega'} \frac12\lambda [ s^2 2\frac{\sqrt{2n}}{g} \psi' + \cO(|s|^4)]^2. 
	\end{align}
	Expanding the outer square gives and using Lemma \ref{os-lem} gives
	\begin{align}
	\int_{\Omega'}&\frac12\lambda (\phi_s^2 - \xi_s^2) = s^4\int_{\Omega'} \frac{4\lambda n}{g^2} \psi'^2 + \cO(|s|^6). \label{EWS-asymp-10}
	\end{align}
	
	Adding \eqref{EWS-asymp-1} - \eqref{EWS-asymp-10} and dividing by $|\Omega'|$ gives \eqref{WS-energy-s-1}, where $R_\varepsilon$ collects the $\cO(|s|^6)$ remainder terms. $R_\varepsilon$ has continuous derivatives of all orders because it is a sum of integrals of the form \eqref{os-rem} with $f_s$ and $g_s$ coming from the continuously differentiable remainder terms $\cO(|s|^p)$ 
 of \eqref{s-asym}. \end{proof}

\DETAILS{\section{The Parametrization of Lattice Shapes}\label{sec:ls}

Let $\cL\subset\R^2$ be a lattice. In order to define the shape of $\cL$, it is convenient to identify $\R^2$ with $\C$ via $(x_1,x_2) \mapsto x_1 + i x_2$, and view $\cL$ as a subset of $\C$. It is a well-known fact (see e.g. \cite{Ahlfors}) that any lattice $\cL\subset\C$ can be given a basis $r,r'$ such that the ratio $\tau = \frac{r'}{r}$ satisfies the inequalities:
\begin{enumerate}[(i)]
	\item $\Im\tau > 0$;
	\item $|\tau| \geq 1$;
	\item $-\frac 12<\Re\tau<\frac 12$ and $\Re\tau\geq 0$ if $|\tau|=1$.
\end{enumerate}
Although the basis is not unique, the value of $\tau$ is, and we will use that as a measure of the shape of the lattice.

Using a rotational symmetry, we may assume that $\cL$ has as a basis $\{\rho,\rho\tau\}$, where $\rho$ is a positive real number.}

\section{Spectral analysis of the operator $-\Delta_{a^n}$}\label{saol}

Recall from the main text, but in vector notation, that 
$a^n:=  \frac{n}{2}x^\perp$, where $(x^1, x^2)^\perp = (-x^2, x^1)$,  $\n_{q} := \n - iq = (\n_1, \n_2)$, $\n_j :=  \partial_j - iq_j, \partial_j\equiv \partial_{x^j}$, and $\Delta_q:=\n_{q}^2=-\n_{q}^*\n_{q}$.
The next proof follows Section $5$ of \cite{CSS}.
\begin{proof}[Proof of Proposition \ref{prop:Landau-ham-spec}] 
	The self-adjointness of the operator $-\Delta_{ a^n}$ is well-known. To find its spectrum, we  introduce  the complexified covariant derivatives (harmonic oscillator annihilation and creation  operators), $\bar\p_{ a^n} $ and $\bar\p_{ a^n}^*=- \p_{ a^n}$, with 
	\begin{equation}
	\bar\p_{ a^n}  :=  (\nabla_{ a^n})_1 + i(\nabla_{ a^n})_2 =\partial_{x^1} + i\partial_{x^2} + \frac{1}{2} n (x^1 + i  x^2).
	\end{equation}
	One can redily verify that these operators satisfy the following relations:
	\begin{align} \label{commut-rel}[\bar\p_{ a^n}, (\bar\p_{ a^n})^*] &= \Curl a^n =  n;\\
	\label{Landau-cr-annih}    -\Delta_{ a^n} -  n &= (\bar\p_{ a^n})^*\bar\p_{ a^n}.  \end{align}
	As for the harmonic oscillator (see e.g. \cite{GS}), this gives explicit information  about the spectrum of $-\Delta_{ a^n}$, namely \eqref{spec-Landau-app}, with each eigenvalue is of the same multiplicity. 
	Furthermore, the above properties imply \eqref{nullL'}.

	We find $\Null \bar\p_{ a^n}$.
	A simple calculation gives the following operator equation
	\begin{equation*}
	e^{-\frac{ n}{2}(i x^1 x^2-(x^2)^2)}\bar\p_{ a^n} e^{\frac{ n}{2}(i x^1 x^2-(x^2)^2)} = \partial_{x^1} + i\partial_{x^2}.
	\end{equation*}
	(The transformation on the left-hand side is highly non-unique.)    
	This immediately proves that 
	\begin{equation}\label{lin-probl'} \bar\p_{ a^n} \psi = 0,\end{equation}
	if and only if $\theta = e^{-\frac{ n}{2}(i x^1 x^2-(x^2)^2)}\psi$ satisfies $(\partial_{x^1} + i\partial_{x^2})\theta = 0$.
	We now identify $x \in \R^2$ with $z = x^1 + ix^2 \in \C$ and see that this means that $\theta$ is analytic and
	\begin{equation}\label{psi-nullL}    \psi \left( x \right) =   e^{ -\frac{\pi n}{2\im\tau} (|z|^2-z^2) } \theta(z, \tau),\ z = (x^1+ i x^2)/ \sqrt{\frac{2\pi }{\im\tau} }.
	\end{equation}
	where we display the dependence of $\theta$ on $\tau$. 
	The quasiperiodicity of $\psi$ transfers to $\theta$ as follows:
	\begin{equation*}
	\theta(z + 1, \tau) = \theta(z, \tau), \qquad
	\theta(z + \tau, \tau) =  e^{ -2\pi inz } e^{ -in\pi\tau  } \theta(z, \tau).
	\end{equation*}
	
	The first relation ensures that $\theta$ have a absolutely convergent Fourier expansion of the form
	$    \theta(z, \tau) = \sum_{m=-\infty}^{\infty} c_m e^{2\pi m iz}.$ 
	The second relation, on the other hand, leads to relation for the coefficients of the expansion:        
	$ c_{m + n} = e^{-in\pi z} e^{i2m\pi\tau} c_m$, 
	which together with the previous statement implies \eqref{theta-repr}.              
\end{proof}

Next, we claim that the solution \eqref{psi-nullL} satisfies 
\begin{align}  \label{psi-parity} \psi (x) = \psi (- x). \end{align}
By \eqref{psi-nullL}, it suffices to show that $\theta (z) = \theta (- z)$.  We show this for $n=1$. Denote the corresponding $\theta$ by $\theta (z, \tau)$. Iterating the recursive relation for the coefficients in \eqref{theta-repr}, we obtain the following
standard representation 
for the theta function    
\begin{align}  \label{theta-series} &\theta (z, \tau) =   \sum_{m=-\infty}^{\infty} e^{2\pi i  ( \frac12 m^2\tau + m z)}. \end{align} 
We observe that $\theta (- z, \tau)=\theta (z, \tau)$ and therefore $\psi_0 ( - x)=\psi_0 ( x)$. Indeed, using the expression \eqref{theta-series}, we find, after changing $m$ to $-m'$, we find
\begin{equation}\theta (- z, \tau)=   \sum_{m=-\infty}^{\infty} e^{2\pi i  ( \frac12 m^2\tau - m z)}
=   \sum_{m'=-\infty}^{\infty} e^{2\pi i  ( \frac12 m'^2\tau + m' z)} =\theta (z, \tau).\end{equation}

\bigskip

\noindent{\bf Declarations}

\noindent The research on this paper is supported in part by NSERC Grant No. 2017-06588.

\medskip
\noindent The authors have no relevant financial or non-financial interests to disclose.



\end{document}